\documentclass[11pt]{article}
\usepackage{amssymb}
\usepackage{amsfonts}
\usepackage{amsmath,amsthm}
\usepackage{amstext}
\usepackage{graphicx}
\usepackage{natbib}
\usepackage{geometry}
\usepackage{hyperref}
\usepackage{authblk}
\usepackage{natbib}
\usepackage{xcolor}
\hypersetup{
	colorlinks,
	citecolor={blue!50!black},
}

\newtheorem{lemma}{\textsc{Lemma}}[section]
\newtheorem{proposition}{\textsc{Proposition}}[section]
\newtheorem{assumption}{\textsc{Assumption}}[section]

\newtheorem{theorem}{\textsc{Theorem}}[section]

\newtheorem{definition}{Definition}[section]
\newtheorem{example}{Example}[section]
\renewcommand{\baselinestretch}{1.3}

\renewcommand{\theequation}{\arabic{section}.\arabic{equation}}
\renewcommand{\thetable}{\arabic{section}.\arabic{table}}
\renewcommand{\thefigure}{\arabic{section}.\arabic{figure}}

\title{Optimal Taxation with Endogenous Default under Incomplete Markets}
\author{Demian Pouzo and Ignacio Presno \thanks{Pouzo: UC Berkeley, Dept. of Economics, 530-1 Evans \# 3880, Berkeley CA
		94720, Email: dpouzo@econ.berkeley.edu; Presno: Universidad de Montevideo, Dept. of Economics, 2544 Prudencio de Pena St., Montevideo, Uruguay 11600. Email: jipresno@um.edu.uy. 
		We are deeply indebted to Xiaohong Chen, Ricardo Lagos and Tom Sargent for their thoughtful advice and insightful discussions. We are also grateful to Arpad Abraham, Mark Aguiar, David Ahn, Andy Atkeson, Marco Basetto, Hal Cole, Jonathan Halket,
		Greg Kaplan, Juanpa Nicolini, Anna Orlik, Nicola Pavoni, Andres
		Rodriguez-Clare, Ana Maria Santacreu, Ennio Stacchetti, and
		especially to Ignacio Esponda and Constantino Hevia. We thank for comments and suggestions to seminar participants at Berkeley, CEF 2013, Central Bank of Chile, CREI-UPF, Federal Reserve Board of Governors, FRB of Chicago, Cleveland, Philadelphia and NY, LACEA-LAMES 2014, LBS, NYU, North American Summer Meeting 2014, Rochester, SED 2013, U of Chicago, UC Davis, UCL, UG, UIUC and UPenn. We are grateful to Nan Lu and Sandra Spirovska for excellent research
		assistance. Usual disclaimer applies.}
}

\begin{document}
	
\maketitle

\begin{abstract}

In a dynamic economy, we characterize the fiscal policy of the government when it levies distortionary taxes and issues defaultable bonds to finance its stochastic expenditure. Default may occur in equilibrium as it prevents the government from incurring in future tax distortions that would come along with the service of the debt. Households anticipate the possibility of default generating endogenous credit limits. These limits hinder the government's ability to smooth taxes using debt, implying more volatile and less serially correlated fiscal policies, higher borrowing costs and lower levels of indebtedness. 
In order to exit temporary financial autarky following a default event, the government has to repay a random fraction of the defaulted debt. We show that the optimal fiscal and renegotiation policies have implications aligned with the data.



\end{abstract}

\section{Introduction}

As originally indicated by \cite{BARRO_JPE79}, tax smoothing motives play a fundamental role in the design of optimal fiscal policies. The possibility to smooth taxes, however, relies significantly on the market structure for government debt. Relative to the seminal results by \cite{LS_JME83} for complete markets economies, \cite{AMSS_JPE02} shows that taxes typically display higher variability and lower serial correlation under incomplete markets with sufficiently stringent debt and asset limits. In this latter economy, the government is assumed to have access to one-period \emph{risk-free} bonds. But how are optimal tax and debt policies affected if the government is given the option to default and restructure its debt, as we have observed several times throughout history? In this paper we answer this question. We show how the presence of default risk and the actual default contingency gives rise to endogenous credit limits that hinder the government's ability to smooth shocks using debt. As a result, taxes become even more volatile, and less serially correlated than in the benchmark incomplete market framework.

We analyze the dynamic taxation problem of a benevolent government with access to distortionary labor
taxes and non-state-contingent debt in a closed economy. We assume, however, that the government cannot commit to pay back the debt. In case the government defaults, the economy enters temporary financial autarky wherein it faces
exogenous random offers to repay a fraction of the defaulted debt that
arrive at a given rate.\footnote{While in our model we allow only
  for outright default on government bonds, governments in practice could
  liquidate the real value of the debt and repayments through
  inflation risk, which could be viewed as a form of partial
  default. In several economies, however, this second option may not
  available, either because the country has surrendered the control
  over its monetary policy (for example, as in the eurozone, Ecuador,
  and Panama), or a significant portion of the government debt is
  either foreign-currency denominated, or local-currency denominated
  but indexed to the CPI or a similar index. We see our environment
  particularly appropriate for this class of economies.} The government has the
option to accept the offer --- and thereby exit financial autarky --- or to
stay in financial autarky awaiting new offers. During temporary
financial autarky, the defaulted debt still has some value as a fraction of it will be eventually repaid in the future. Hence, households can trade the defaulted debt in a
secondary market from which the government is excluded, giving rise to an equilibrium price of the debt during the
period of default. Finally, in line with the aforementioned optimal taxation literature, we assume that the government commits itself to its optimal path of taxes as long as the economy is not in financial autarky.


The government has three policy instruments: (1)
distortionary taxes, (2) government debt, and (3) default/repayment
decisions that consist of: (a) whether to default on the
outstanding debt and (b) whether to accept the offer to exit
temporary financial autarky. In order to finance the stochastic process of expenditures, the government faces a trade-off between levying distortionary
taxes and not defaulting, or issuing debt and thereby increasing the exposure to
default risk. Defaulting introduces some degree of
state contingency on the payoff of the debt since the financial
instrument available to the government becomes an option, rather
than a non-state-contingent bond. In equilibrium, the government may optimally
decide not to honor its debt contracts ---even though the bondholders
are the households whose welfare it cares about--- because default would prevent the government from incurring in the future tax distortions that would come along with the service of the debt. We believe this is a novel motive to default on government debt which, to our knowledge, had not been explored before in the literature.

The option to default, however, does not come free of charge: in equilibrium households anticipate the possibility of default, demanding a compensation for it embedded in the pricing of the bond; this originates a ``Laffer curve" type of pattern for the bond proceedings, thereby implying endogenous credit limits. Two key implications follow: First, in this sense, our model generates ``debt intolerance" endogenously, in contrast with \cite{AMSS_JPE02} wherein credit limits are exogenous. Second, in our framework, these credit limits together with higher borrowing costs prevent the government from completely spreading the tax burden over time, as in the risk-free debt economy.

Our theoretical model is motivated by some observations of tax and debt dynamics along with episodes of domestic defaults and debt restructuring throughout the history for a number of economies. To name a few: (1) the government policies and fiscal accounts for France in the 100 years preceding the French Revolution of 1789, wherein France defaulted recurrently and no tax-smoothing features stand out from its debt and tax dynamics; see \cite{SV_JPE95}.\footnote{In contrast, time-series for Great Britain debt broadly resemble Barro's random walk behavior. Great Britain honored all its debt contracts through this period.} (2) The debt restructuring plan proposed by Secretary of the Treasury Alexander Hamilton for the U.S. economy in 1790 and the policy debates around it.\footnote{As \cite{HALLSARGENT_WP13} states, through renegotiating the U.S. debt with no discrimination scheme across creditors, Hamilton was hoping to improve the federal government access to credit markets, which in turn would eventually allow for lower borrowing costs to finance temporary increases in government spending and thereby smooth out taxes. Few years later, in 1807, in his report to Congress, Secretary of Treasury Albert Gallatin was advocating for a fiscal policy largely in accord with Barro's tax smoothing idea.} (3) Finally, several emerging economies, where default is a recurrent event, taxes are more volatile, borrowing costs are sizable, and indebtedness levels are significantly lower than for developed economies.\footnote{Several domestic defaults for emerging economies were also external. In any case, empirical evidence seems to suggest that government default has a significant direct impact on domestic residents, either because a considerable portion of the foreign debt is in the hands of local investors, or because the government also defaults on domestic debt. For example, for Argentina's default in 2001, about 60 percent of the defaulted debt is estimated to have been in the hands of Argentinean residents; local pension funds alone held almost 20 percent of the total defaulted debt. For Russia's default in 1998 about 60 percent of the debt was held by residents. For Ukraine's default in 1997-98, residents --- Ukrainian banks and the National Bank of Ukraine (NBU) among others --- held almost 50 percent of the outstanding stock of T-bills. See \cite{SZ_book06}.}


In a benchmark case, with quasi-linear utility and i.i.d. government expenditure, we characterize analytically the
determinants of the optimal default decision and its effects on
the optimal taxes, debt and allocations. In particular, we first show that default is more likely when the government expenditure or debt
is higher, and that the government is more likely to accept any given
offer to pay a fraction of the defaulted debt when the level of
defaulted debt is lower; these theoretical results have implications for haircuts and duration of debt restructuring processes aligned with the data. Second, we find that prices --- both outside
and during financial autarky --- are non-increasing on the level of
debt, thus implying that spreads are non-decreasing and also implying
the existence of endogenous borrowing limits.
Third, we prove that the law of  motion of the optimal
government tax policy departs from the standard martingale-type behavior found in the standard incomplete market framework. \cite{BARRO_JPE79} conjectured that optimal debt and taxes should exhibit a random walk behavior. This result was re-affirmed by \cite{AMSS_JPE02} in a general equilibrium setup under some restrictions on asset limits.
In our paper we show how this result is altered once default risk is incorporated. More specifically, the law of motion of the optimal government tax policy will be affected, on the one hand, by the benefit
from having more state-contingency on the payoff of the bond, but, on the other hand, by the cost of having the option to default (manifested in higher borrowing costs). 



Finally, we conduct a series of numerical exercises to assess the quantitative performance of our model along the aforementioned dimensions. 

\medskip

\textbf{Related Literature.}
A growing literature has emerged from the seminal work of \cite{BARRO_JPE79} highlighting the role of tax-smoothing motives in the design of optimal fiscal and debt policy. In a partial equilibrium deterministic framework, \cite{BARRO_JPE79} assumes that the government needs to finance an exogenous sequence of public spending either by levying distortionary taxes or issuing non-state-contingent debt. Barro shows that the government wants to smooth tax distortions across periods by recurring to debt issuance to finance temporary increases in public spending. In a stochastic environment, the model predicts a random walk response of debt and taxes to public spending. 

\cite{LS_JME83} shows that this result 
does not survive in an environment with complete markets. In particular, in a general equilibrium setup where a Ramsey planner disposes of distortionary labor taxes and a complete set of Arrow-Debreu securities, \cite{LS_JME83} proves that optimal taxes (and debt) do not follow a random walk process but roughly inherit the stochastic properties from the government spending dynamics.

By extending \cite{LS_JME83} framework to incomplete markets, \cite{AMSS_JPE02} revitalizes \cite{BARRO_JPE79} and show that the Ramsey plan prescribes a near-random walk component into debt and taxes under certain conditions for asset/debt limits. Our model builds on \cite{AMSS_JPE02} by adding two key ingredients. First, we give the government the option to default on its debt, thus endogenizing the ad hoc government credit limits imposed in \cite{AMSS_JPE02}, as well as the return on government bonds. Second, in our model the government is confronted with an exogenous debt restructuring process following a default event.

\cite{FARHI_WP07} extends the setting of \cite{AMSS_JPE02} by introducing capital accumulation and letting the government to levy capital taxes in addition to labor ones. \cite{SHIN_WP08} studies the Ramsey fiscal policy in an environment akin to \cite{AMSS_JPE02} but with heterogeneous households facing idiosyncratic labor risks.\footnote{\cite{ANGELETOS} and \cite{BN_JME04} study the optimal maturity structure of government debt and show how non-contingent bonds of different maturities can be used to implement the allocations with state-contingent debt.} We see these papers as complementary to ours. A recent paper by \cite{BEGS_WP13} builds on \cite{AMSS_JPE02} by allowing the government to trade a single possibly risky asset. While in \cite{BEGS_WP13} the asset payoff follows an ad hoc exogenous process, in our model it is driven by the government optimal default decisions.

Our work also contributes to the literature on quantitative default models. We model the strategic default decision
of the government as in \cite{ARELLANO_AER08} and \cite{AG_JIE06}, who first adapted the theoretical framework of \cite{EG_RESTUD81} to study sovereign default risk and its interaction with the business cycles in emerging economies.\footnote{\cite{CHATTEE-EIY-AER12} extends this setup by incorporating long-term debt.} From this strand of literature, our paper is closely related to \cite{DODA_WP07} and \cite{CSS_WP09}. Both papers analyze the procyclicality of fiscal policy in developing countries by solving an optimal taxation problem of a government with distortionary labor taxes and incomplete financial markets.\footnote{\cite{AAG_WP08} also allow for default in a small open economy with capital where households do not have access to neither financial markets nor capital and provide labor inelastically. The authors' main focus is on capital taxation and the debt ``overhang'' effect.} Their models, however, differ from ours along several important dimensions. First, they consider a small open economy with foreign lenders, while we assume a closed economy.\footnote{From our viewpoint, little attention has been put on quantitative models with domestic default. A notable exception is \cite{DERASMO_MENDOZA}, where, in an economy with wealth inequality across domestic agents, redistributional motives influence the incentives to default. \cite{DGS_WP15} studies optimal policies for taxes, transfers, and both domestic and external debt in an open economy with competition between political parties.} In our economy, bondholders are the domestic households whose welfare our benevolent government wants to maximize, so tax-smoothing concerns are the dominant determinant in the decision to default or not. Second, we assume a debt restructuring process while in \cite{DODA_WP07} and \cite{CSS_WP09} the government exits autarky exempt of any repayment of the defaulted debt. Third, these papers assume that the government cannot commit to fiscal policies, while we do, in line with \cite{AMSS_JPE02}.
Finally, we address a different class of question. Instead of exploring fiscal regularities in emerging economies, our work is rather centered on the normative analysis of optimal taxation in the context of government default and provides an analytical characterization of optimal fiscal and debt policies.

\cite{BENJAMIN_WRIGHT_WP09}, \cite{PW_WP08}, \cite{YUE_WP07} and \cite{BAI_ZHANG_JIE12} propose alternative ways of modeling the entire
  debt restructuring process. Although our mechanism to reach a debt settlement is not fully endogenous as theirs, it is sufficiently rich to replicate key features of debt renegotiation episodes in the data.

We assume that the government has the ability to commit to a tax policy at any time, except in the periods of debt renegotiation when it regains access to financial markets, in which case it can revise and reset its fiscal policy. This assumption is to some extent similar to \cite{DEBORTOLI_NUNES_WP08}. \cite{DEBORTOLI_NUNES_WP08} studies the dynamics of debt in a setting similar to \cite{LS_JME83} but with the peculiarity that at each time $t$, with some given probability, the government can lose its ability to commit to taxes and re-optimize; a feature labeled by the authors as ``loose commitment.'' Thus, our model can be viewed as providing a mechanism that ``rationalizes'' this probability of ``loosing commitment'' by allowing for endogenous default, and resetting of fiscal policy when a debt settlement is reached. 

\medskip

\textbf{Roadmap.} The paper is organized as follows. Section \ref{sec:model} introduces the model. Section
\ref{sec:eqm} presents the competitive equilibrium. Section
\ref{sec:govt} presents the government's problem. Section
\ref{sec:bench} derives analytical results. Section
\ref{sec:num-simul} contains some numerical exercises. Section \ref{sec:conclusion} briefly concludes.
All proofs are gathered in the appendices.


\section{The Economy}
\label{sec:model}

In this section we describe the stochastic structure of the model, the timing and policies of the government and present the household's problem.

\subsection{The Setting}

Let time be indexed as $t=0,1,...$. Let $(g_{t},\delta_{t})$ be the vector of government expenditure at time $t$ and the fraction of the defaulted debt which is to be repaid when exiting autarky, respectively. If the economy is not in financial autarky, $\delta_{t}$ is either one or zero in order to model the option of the government to repay the totality of the debt or to default. These are the exogenous driving random variables of this economy. Let $\omega_{t} \equiv (g_{t},\delta_{t}) \in  \mathbb{G} \times \bar{\Delta}$, where $\mathbb{G} \subset \mathbb{R}$, $\bar{\Delta} \equiv \Delta \cup \{1\} \cup \{ \bar{\delta} \}$ and $\Delta \subset [0,1)$, and in order to avoid technical difficulties, we assume cardinal of $\mathbb{G}$ and $\Delta$ are finite. The set $\Delta$ models the offers --- as fractions of outstanding debt --- to repay the defaulted debt, and $\bar{\delta}$ is designed to capture situations where the government does not receive any offer to repay. For any $t \in \{1,....,\infty\}$, let $\Omega^{t} = (\mathbb{G} \times \bar{\Delta})^{t}$ be the space of histories of exogenous shocks up to time $t$; a typical element is $\omega^{t} = (\omega_{0},\omega_{1},...,\omega_{t})$.

\subsection{The Government Policies and Timing}

In this economy, the government finances exogenous government expenditures by levying labor distortionary taxes and trading one-period, discount bonds with households. The government, however, cannot commit to repay and may default on the bonds at any point in time.

Let $\mathbb{B} \subseteq \mathbb{R}$ be compact. Let $B_{t+1} \in \mathbb{B}$ be the quantity of bonds issued at time $t$ to be paid at time $t+1$ so that $B_{t+1}>0$ means that the government is borrowing at time $t$ from households. Let $\tau_{t}$ be the linear labor tax. Also, let $d_{t}$ be the default decision, which takes value 1 if the government decides to default and 0 otherwise. Finally, let $a_{t}$ be the decision of accepting an offer to repay the defaulted debt. It takes value 1 if the offer is accepted and 0 otherwise.

For any $t$, let $\phi_{t}$ be the variable that takes value 0
if at time $t$ the government cannot issue bonds during this period,
and value 1 if it can. The implied law of motion for $\phi_{t}$ is
$\phi_{t} \equiv \phi_{t-1} (1-d_{t}) + (1-\phi_{t-1}) a_{t} $. That means that if at time $t-1$, the government could issue bonds, then $\phi_{t}
= (1-d_{t})$, but if it was in financial autarky, then $\phi_{t} =
a_{t}$, reflecting the fact that the government regains access to
financial markets only if the government decides to renegotiate the
defaulted debt.

The timing for the government is as follows. Following a period with financial access, after observing the current government expenditure, the government has the option to default
on the totality of the outstanding debt carried from last period, $B_{t}$.

\begin{figure}
  \centering
  \includegraphics[height=3.0in,width=5.5in]{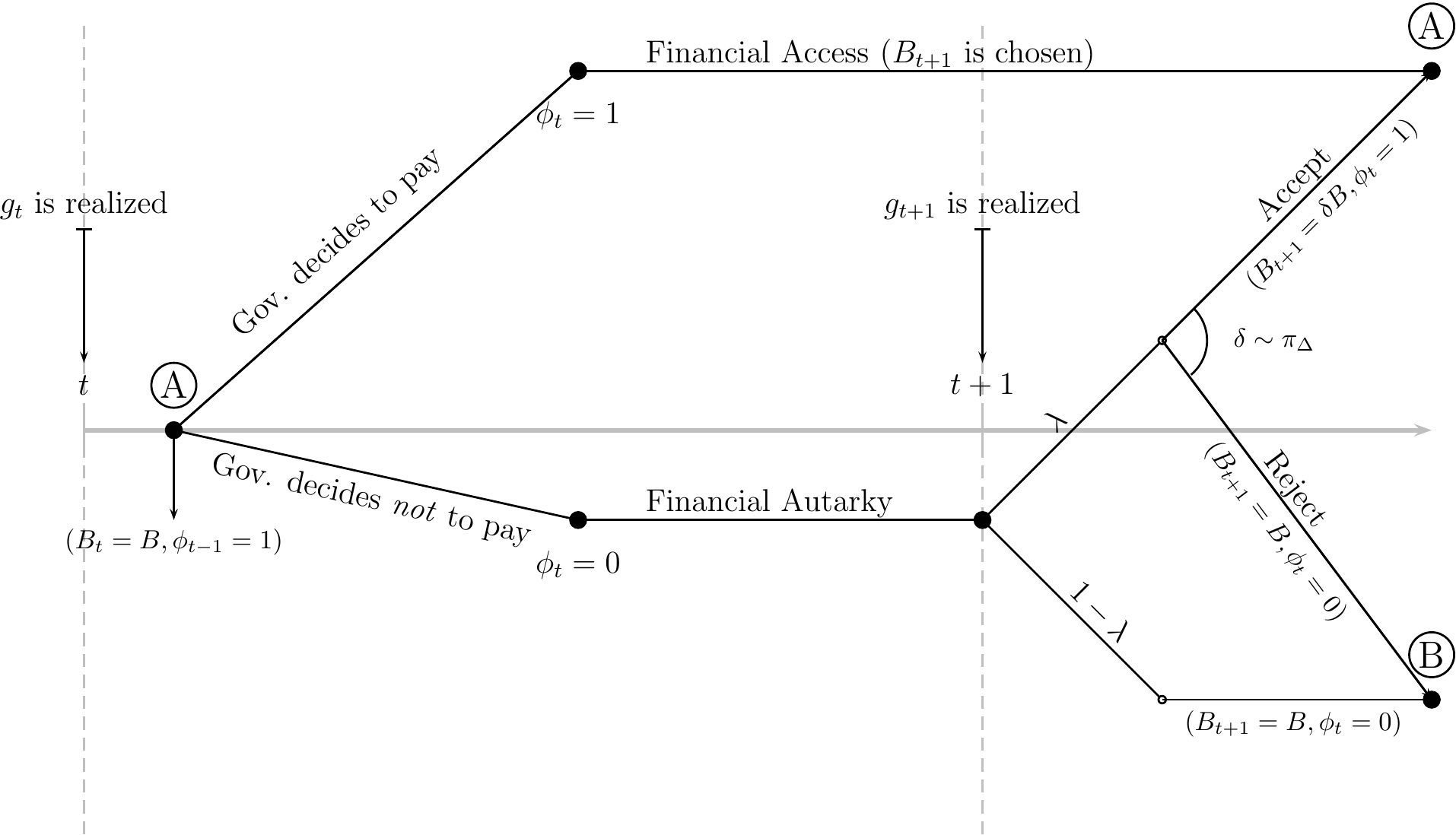}
  \caption{Timing of the Model}
  \label{fig:timing}
\end{figure}

As shown in figure \ref{fig:timing}, if the government exercises the option to default at time $t$, it cannot issue bonds in that period and runs a balanced budget, i.e., tax revenues equal government expenditure. At the beginning of next period, time $t+1$, with probability $1-\lambda$, the government remains in temporary
financial autarky for that period (node $B$). With probability $\lambda$, the government receives a random offer to repay a fraction $\delta$ of the debt, and has the option to accept or reject it. If the government
accepts the offer, it pays the restructured amount (the outstanding defaulted debt times the
fraction $\delta$), and can issue new bonds for
the following period (node $A$). If the government rejects the offer, it stays in
temporary autarky (node $B$).

Finally, if the government decides not to default, it levies distortionary labor
taxes, and allocates discount bonds to the households to cover the expenses $g_{t}$ and liabilities carried from last period. Next period, it has again the option to default for the new values of outstanding debt and government expenditure (node $A$).

As it will become clear later, default on bonds can be seen as a
negative lump-sum transfer to households, but a costly one. Default
will turn to be costly for two reasons. First, households anticipate
the government default strategies and demand higher returns to buy
the bond. Second, default is followed by temporary
financial autarky. During autarky, the government is not only unable
to issue debt but also could be subject to an ad hoc output
cost, as shown later.

We now formalize the probability model. Let $\pi_{\mathbb{G}} : \mathbb{G} \rightarrow \mathcal{P}(\mathbb{G})$ be the Markov transition probability function for the process of government expenditures and let $\pi_{\Delta} \in \mathcal{P}(\Delta)$ be the probability measure over the offer space $\Delta$.\footnote{For a finite set $\mathbb{X}$, $\mathcal{P}(\mathbb{X})$ is the space of all probability measures defined over $\mathbb{X}$. Also, for any $A \subseteq \mathbb{X}$, the function $\mathbf{1}_{A}(\cdot)$ takes value 1 over the set $A$ and 0 otherwise.}

\begin{assumption}
\label{ass:prob_def} \label{ass:prob_iid}
 For any $(t,\omega^{t})$, $\Pr( g_{t} = g | \omega^{t-1} ) =
 \pi_{\mathbb{G}}( g | g_{t-1})$ for any $g \in \mathbb{G}$ and
 \begin{align*}
   \Pr
 (\delta_{t} = \delta | g_{t}, \omega^{t-1}) = \left\{
     \begin{array}{ cc}
       \mathbf{1}_{\{ 1 \}}(\delta) &if~\phi_{t-1} = 1 \\
       (1-\lambda ) \mathbf{1}_{\{ \bar{\delta} \}} (\delta) +
       \lambda  \pi_{\Delta}(\delta) &if ~\phi_{t-1}=0
     \end{array}
\right.
 \end{align*}
for any $\delta \in \bar{\Delta}$.\footnote{It is easy to generalize this to a more general
   formulation such as $\lambda$ and $\pi_{\Delta}$ depending on $g$. For instance, we could allow for, say, $\pi_{\Delta}(\cdot | g_{t},B_{t},d_{t},d_{t-1},...,d_{t-K})$ some $K>0$, denoting that possible partial payments depend on the credit history and level of debt. See
   	\cite{RR_WP03}, \cite{RR_WP08} and \cite{YUE_WP07} for an intuition behind
   	this structure.}
\end{assumption}

Essentially, this assumption imposes a Markov restriction on the
probability distribution over government expenditures and also additional restrictions over the probability of offers. In particular, this assumption implies that in financial autarky with probability $1-\lambda$, $\delta =
\bar{\delta}$ (i.e., receiving no offer) and with probability $\lambda$, an offer from the offer space is drawn according to $\pi_{\Delta}$. Also, if $\phi_{t-1}=1$ (i.e., the government was not in financial autarky at period $t-1$), then $\delta_{t}=1$ with probability one, which implies that if the government decides \emph{not to default} at time $t$, it will pay the totality of the outstanding debt. 

Finally, we use $\Pi$ to denote the probability distribution over $\Omega^{\infty}$ generated by assumption \ref{ass:prob_def}, and $\Pi(\cdot|\omega^{t})$ to denote the conditional probability over $\Omega$, given $\omega^{t}$.

The next definitions formalize the concepts of government policy, allocation, prices of bonds and the government budget constraint. In particular, it formally introduces the fact that taxes, default decisions and debt depend on histories of past realizations of shocks, and in particular that debt is non-state contingent (i.e., $B_{t+1}$ only depends on the history up to time $t$, $\omega^{t}$).

\begin{definition}
\label{def:sig}
 A government policy is a collection of stochastic processes $\boldsymbol{\sigma} = (B_{t+1},\tau_{t},d_{t},a_{t})_{t=0}^{\infty}$, such that for each $t$, $(B_{t+1},\tau_{t},d_{t},a_{t}) \in \mathbb{B} \times [0,1] \times \{0,1\}^{2}$ are measurable with respect to $\omega^{t}$ and $(B_{0},\phi_{-1})$.

\end{definition}

\begin{definition}
\label{def:allocation}
An allocation is a collection of stochastic processes $(g_{t},c_{t},n_{t})_{t=0}^{\infty}$ such that for each $t$, $(g_{t},c_{t},n_{t}) \in \mathbb{G} \times \mathbb{R}_{+} \times [0,1]$ are measurable with respect to $\omega^{t}$ and $(B_{0},\phi_{-1})$.
\end{definition}

Given a government policy, we say an allocation is \emph{feasible} if for any $(t,\omega^{t})$
\begin{align}
  c_{t}(\omega^{t}) + g_{t} = \kappa_{t}(\omega^{t}) n_{t}(\omega^{t}), \label{eqn:feas}
\end{align}
where $\kappa_{t} : \Omega^{t} \rightarrow \mathbb{R}_{+}$ is such that $\kappa_{t}(\omega^{t})$ is the productivity at period $t$, given history $\omega^{t}$. For simplicity, we set  $\kappa_{t}(\omega^{t}) = \phi_{t}(\omega^{t}) + \kappa (1-\phi_{t}(\omega^{t}))$ with $\kappa \in (0,1]$. The fraction $1-\kappa$ represents direct output loss following a default event, associated for example with financial disruption in the banking sector, limited insurance against idiosyncratic risk, among others.

\begin{definition}
\label{def:price}
  A price process is an stochastic process $(p_{t})_{t=0}^{\infty}$ such that for each $t$, $p_{t} \in \mathbb{R}_{+}$ is measurable with respect to $\omega^{t}$ and $(B_{0},\phi_{-1})$.
\end{definition}

Note that $p_{t}$ denotes the price of one unit of debt in any state of the world, both with access to financial markets and during autarky, where it represents the price of defaulted debt in secondary markets. Finally, we introduce the government budget constraint.

\begin{definition}
\label{def:sig-att}
  A government policy $\boldsymbol{\sigma}$ is attainable, if for all $(t,\omega^{t})$,
     \begin{align}\label{eqn:sig-att-1}
       g_{t} + \phi_{t}(\omega^{t}) \delta_{t} B_{t}(\omega^{t-1}) \leq  \kappa_{t}(\omega^{t})\tau_{t}(\omega^{t}) n_{t}(\omega^{t})  + \phi_{t}(\omega^{t}) p_{t}(\omega^{t}) B_{t+1}(\omega^{t}),
     \end{align}
and $d_{t}(\omega^{t}) = 1$ if $\phi_{t-1}(\omega^{t-1}) = 0$ and
$a_{t}(\omega^{t})=0$ if $\phi_{t-1}(\omega^{t-1}) = 1$ or $\delta_{t}
= \bar{\delta}$.\footnote{The inequality in equation
  \ref{eqn:sig-att-1} implies that the government can issue lump-sum
  transfers to the households. Lump-sum taxes are not permitted.}
\end{definition}

Observe that in equation \ref{eqn:sig-att-1}, if the government is in financial autarky ($\phi_{t}(\omega^{t}) = 0$), its budget constraint boils down to $g_{t} \leq \kappa_{t}(\omega^{t})\tau_{t}(\omega^{t}) n_{t}(\omega^{t})$. On the other hand, if the government has access to financial markets ($\phi_{t}(\omega^{t}) = 1$), then it has liabilities to be repaid for $\delta_{t} B_{t}$ and can issue new debt.\footnote{If the government had access to financial markets at time $t-1$ ($\phi_{t-1}=1$), then by assumption \ref{ass:prob_def}, $\delta_{t} = 1$ and the outstanding debt if simply $B_{t}$.} The final restriction on $d_{t}(\omega^{t})$ and $a_{t}(\omega^{t})$ simply states that if last period the government was in financial autarky, then it trivially cannot choose to default at time $t$, and if $\delta_{t} = \bar{\delta}$ or if last period the government had access to financial markets $a_{t}(\omega^{t})$ is set to 0.

A few final remarks about the ``debt-restructuring process'' are in order.  This process intends to capture
the fact that after defaults, economies see their access to credit severely
hindered.\footnote{The duration of debt restructuring after sovereign
	defaults in particular on external debt has received considerable attention in the
	literature. For instance, for Argentina's default in 2001 the settlement with the
	majority of the creditors was reached in 2005. In the default
	episodes of Russia (1998), Ecuador (1999) and Ukraine (1998), the
	renegotiation process lasted 2.3, 1.7 and 1.4 years, respectively,
	according to \cite{BENJAMIN_WRIGHT_WP09}. In general, domestic debt
	restructuring periods tend to be not as long as in the case of external debt. For example, as documented by \cite{SZ_book06}, after the default by Russia in 1998 it took six months to restructure the domestic GKO bonds.} The parameters $(\lambda,\pi_{\Delta})$ capture the fact that
debt restructuring is time-consuming but, generally, at the end a
positive fraction of the defaulted debt is honored.

\subsection{The Household's Problem}

There is a continuum of identical households that are price takers and have time-separable preferences for consumption and labor processes. They also make savings decisions by trading government bonds. Formally, we define a household debt process as a stochastic process given by $(b_{t+1})_{t=0}^{\infty}$ where $b_{t+1} : \Omega^{t} \rightarrow  [\underline{b},\overline{b}]$ is the household's savings in government bonds at time $t+1$ for any history $\omega^{t}$.\footnote{We assume $b_{t+1} \in [\underline{b},\overline{b}]$ with $[\underline{b},\overline{b}] \supset \mathbb{B}$ so in equilibrium these restrictions will not be binding.}

For convenience, let $q_{t}$ denote the price of defaulted debt at time $t$, i.e., $q_{t} = p_{t}$ if $\phi_{t}=0$. Given a government policy $\boldsymbol{\sigma}$, for each $t$, let $\varrho_{t} : \Omega^{t} \rightarrow \mathbb{R}$ be the payoff of a government bond at period $t$, i.e.
\begin{align} \label{eqn:payoff}
   \varrho_{t}(\omega^{t}) = \phi_{t}(\omega^{t})\delta_{t} + (1-\phi_{t}(\omega^{t})) q_{t} (\omega^{t}).
\end{align}
From the households' point of view (which takes government actions as given) the debt is an asset with a state-dependent payoff. This dependence clearly illustrates that default decisions add certain degree of state contingency to the government debt. In particular, if $\phi_{t}(\omega^{t}) = 1$, then $\varrho_{t}(\omega^{t}) = \delta_{t}$ denoting the fact that the government pays a fraction $\delta_{t}$.
If the government defaults or rejects the repayment option, the household can sell each unit of government debt in the secondary market at a price $\varrho_{t}(\omega^{t}) = q_{t}(\omega^{t})$.

The household's problem consists of choosing consumption, labor and debt processes in order to maximize the expected lifetime utility. That is, given  $(\omega_{0},b_{0})$ and $\boldsymbol{\sigma}$,
\begin{align*}
 \sup_{(c_{t},n_{t},b_{t+1})_{t=0}^{\infty} \in \mathbb{C}(g_{0},b_{0};\boldsymbol{\sigma}) } E_{\Pi(\cdot|\omega_{0})} \left[  \sum_{t=0}^{\infty}  \beta^{t} u(c_{t}(\omega^{t}),1-n_{t}(\omega^{t})) \right]
\end{align*}
where $\beta \in (0,1)$ is the discount factor, $E_{\Pi(\cdot|\omega_{0})}[\cdot]$ is the expectation using the conditional probability $\Pi(\cdot|\omega_{0})$, and $\mathbb{C}(g_{0},b_{0};\boldsymbol{\sigma})$ is the set of household's allocations and debt process, given government policy $\boldsymbol{\sigma}$, such that for all $t$ and all $\omega^{t} \in \Omega^{t}$,
\begin{align*}
  c_{t}(\omega^{t}) + p_{t}(\omega^{t}) b_{t+1}(\omega^{t}) = (1-\tau_{t}(\omega^{t}))
  \kappa_{t}(\omega^{t}) n_{t}(\omega^{t}) + \varrho_{t}(\omega^{t}) b_{t}(\omega^{t-1}) +T_{t}(\omega^{t}),
 \end{align*}
where $T_{t}(\omega^{t}) \geq 0$ are lump-sum transfers from the government.

\section{Competitive Equilibrium}
\label{sec:eqm}

We now define a competitive equilibrium for a given government policy and derive the equilibrium taxes and prices.

\begin{definition}\label{def:CEG}
 Given $\omega_{0},B_{0} = b_{0}$ and $\phi_{-1}$, a competitive equilibrium is a government policy, $\boldsymbol{\sigma}$, an allocation, $(g_{t},c_{t},n_{t})_{t=0}^{\infty}$, a household debt process, $(b_{t+1})_{t=0}^{\infty}$, and a price process $(p_{t})_{t=0}^{\infty}$ such that:
  \begin{enumerate}
  \item Given the government policy and the price process, the allocation and debt process solve the household's problem.
   \item The government policy, $\boldsymbol{\sigma}$, is attainable.
   \item The allocation is feasible. 
    \item For all $(t,\omega^{t})$, $B_{t+1}(\omega^{t}) = b_{t+1}(\omega^{t})$, and $B_{t+1}(\omega^{t}) = B_{t}(\omega^{t-1})$ if $\phi_{t}(\omega^{t}) = 0$.
  \end{enumerate}
\end{definition}

The market clearing for debt imposes that, if the economy is in financial autarky --- where the government cannot issue debt, and thus agents can only trade among themselves ---, $B_{t+1}(\omega^{t})= B_{t}(\omega^{t-1})$. This implies, since agents are identical, that in equilibrium $b_{t}(\omega^{t-1}) = b_{t+1}(\omega^{t})$, i.e., agents do not change their debt positions.


\subsection{Equilibrium Prices and Taxes}

In this section we present the expressions for equilibrium taxes and
prices of debt. The former quantity is standard
(e.g. \cite{AMSS_JPE02} and \cite{LS_JME83}); the latter quantity,
however, incorporates the possibility of default of the
government. The following assumption is standard and ensures that $u$
is smooth enough to compute first order conditions.

\begin{assumption}\label{ass:U_prop}
  $u \in \mathbb{C}^{2}(\mathbb{R}_{+} \times [0,1], \mathbb{R})$ with $u_{c} > 0$, $u_{cc} < 0$, $u_{l} > 0$ and $u_{ll} < 0$, and $\lim_{l \rightarrow 0} u_{l}(l) = \infty$.\footnote{$\mathbb{C}^{2}(X,Y)$ is the space of twice continuously differentiable functions from $X$ to $Y$, the subscript $c$ denotes the derivative with respect to the first argument and the subscript $l$ with respect to the second one. The assumption $u_{cc} < 0$ could be relaxed to include $u_{cc} = 0$ (see section \ref{sec:bench} below).}
\end{assumption}

Henceforth, for any $(t,\omega^{t})$, we use $u_{c}(\omega^{t})$ as $u_{c}(c_{t}(\omega^{t}),1-n_{t}(\omega^{t}))$ and proceed similarly for other derivatives and functions. From the first order conditions of the optimization problem of the
households (assuming an interior solution) the following equations
hold for any $(t,\omega^{t})$,\footnote{See appendix \ref{app:HOUSE}
  for the derivation.}
\begin{align}\label{eqn:eq_tax}
  \frac{u_{l}(\omega^{t})}{u_{c}(\omega^{t})} = (1-\tau_{t}(\omega^{t})) \kappa_{t}(\omega^{t}),
\end{align}
and
\begin{align} \notag
  p_{t}(\omega^{t})= & E_{\Pi(\cdot|\omega^{t})} \left[ \beta  \frac{u_{c}(\omega^{t+1})}{u_{c}(\omega^{t})} \varrho_{t+1}(\omega^{t+1})  \right] \\ \label{eqn:eq_price}
 = & \beta E_{\Pi(\cdot|\omega^{t})} \left[   \frac{u_{c}(\omega^{t+1})}{u_{c}(\omega^{t})}  \phi_{t+1}(\omega^{t+1})\delta_{t+1}   \right] + \beta E_{\Pi(\cdot|\omega^{t})} \left[   \frac{u_{c}(\omega^{t+1})}{u_{c}(\omega^{t})} (1-\phi_{t+1}(\omega^{t+1})) q_{t+1} (\omega^{t+1})   \right]
\end{align}

Given the definition of $\varrho$ and the restrictions on $\Pi$, equation \ref{eqn:eq_price} implies for $\phi_{t}(\omega^{t})=1$,\footnote{The notation $(\omega^{t},g,\delta)$ denotes the partial history $\omega^{t+1}$ where $(g_{t+1},\delta_{t+1}) = (g,\delta)$. As it will become clear below, the price $q_{t}$ does not depend on $\delta_{t}$, so we omit it from the notation.  Also, when $\phi_{t+1}(\omega^{t+1}) = 0$, $u_{c}(\omega^{t+1})$ is actually only a function of $g_{t+1}$ (not the entire past history $\omega^{t+1}$) because in equilibrium the government runs a balanced budget.}
\begin{align} \notag
  p_{t}(\omega^{t})  = & \beta \int_{\mathbb{G}} \left(   \frac{u_{c}(\omega^{t},g',1)}{u_{c}(\omega^{t})} (1-d_{t+1}(\omega^{t},g',1)) \right) \pi_{\mathbb{G}}(dg'|g_{t}) \\
& +  \beta  \int_{\mathbb{G}} \frac{u_{c}(\omega^{t},g',1)}{u_{c}(\omega^{t})} d_{t+1}(\omega^{t},g',1) q_{t+1}(\omega^{t},g') \pi_{\mathbb{G}}(dg'|g_{t})  \label{eqn:price-1},
\end{align}
and for $\phi_{t}(\omega^{t})=0$
\begin{align}\notag
  q_{t}(\omega^{t}) = & \beta \lambda \int_{\mathbb{G}} \int_{\Delta} \left(   \frac{u_{c}(\omega^{t},g',\delta')}{u_{c}(\omega^{t})}\delta' a_{t+1}(\omega^{t},g',\delta')  \right) \pi_{\Delta}(d\delta') \pi_{\mathbb{G}}(dg'|g_{t}) \\ \notag
 & + \beta \lambda \int_{\mathbb{G}} \left\{ \int_{\Delta} \left(   \frac{u_{c}(\omega^{t},g',\delta')}{u_{c}(\omega^{t})} (1-a_{t+1}(\omega^{t},g',\delta'))  \pi_{\Delta}(d\delta')   \right) \right\} q_{t+1}(\omega^{t},g') \pi_{\mathbb{G}}(dg'|g_{t}) \\
 & + \beta (1-\lambda) \int_{\mathbb{G}}  \left(   \frac{u_{c}(\omega^{t},g',\bar{\delta})}{u_{c}(\omega^{t})} \right) q_{t+1}(\omega^{t},g') \pi_{\mathbb{G}}(dg'|g_{t}). \label{eqn:price-2}
\end{align}

Equation \ref{eqn:eq_price} reflects the fact that in equilibrium households anticipate the default strategies of the government and demand higher returns to compensate for the default risk. The second line in the Euler equation \ref{eqn:price-1} shows that, due to the possibility of partial repayments in the future, defaulted debt has positive value and agents can sell it in a secondary market at price $q_{t+1}(\omega^{t+1})$. Equation \ref{eqn:price-2} characterizes this price. Each summand in the right-hand side corresponds to a ``branch'' of the tree depicted in
figure \ref{fig:timing}. The first line represents the value of
one unit of debt when an offer arrives and the government decides to repay the realized fraction of the defaulted debt next period. The second and third lines capture the value of one unit of debt when either the government decides to reject the repayment offer, or it does not receive one.

It is easy to show that if $u_{c}=1$, $\lambda = 0$, the equilibrium price described in equations \ref{eqn:price-1} and \ref{eqn:price-2} coincide with that in the standard sovereign default model (e.g., \cite{ARELLANO_AER08} and \cite{AG_JIE06}).
The novelty of our pricing equations with respect to the
standard sovereign default model is the presence of secondary market prices,
$q_{t}$. By allowing for a positive recovery rate, the model is able to deliver a positive price of defaulted debt during the financial autarky period. In sections \ref{sec:bench} and \ref{sec:num-simul}, we analyze further the pricing implications of this model. 


\subsection{Characterization of the Competitive Equilibrium}
\label{sec:charc-CE}

In this environment, the set of competitive equilibria can be characterized by a sequence of non-linear equations which impose restrictions on $(d_{t},a_{t},B_{t+1},n_{t})_{t=0}^{\infty}$ and are derived from the first order conditions of the household, the budget constraint of the government and the feasibility condition. 
The next theorem formalizes this claim.

Henceforth, we call $(d_{t},a_{t},B_{t+1},n_{t})_{t=0}^{\infty}$ an \emph{outcome path} of allocations. We say an outcome path is \emph{consistent with a competitive equilibrium} if the outcome path and $(c_{t},p_{t},b_{t+1},\tau_{t},g_{t})_{t=0}^{\infty}$, derived using the market clearing, feasibility and first order conditions, is a competitive equilibrium. Also, let
\begin{align}
  Z_{t}(\omega^{t}) \equiv z(\kappa_{t}(\omega^{t}),n_{t}(\omega^{t}),g_{t}) = \left( \kappa_{t}(\omega^{t}) - \frac{u_{l}(\omega^{t})}{u_{c}(\omega^{t})} \right) n_{t}(\omega^{t}) - g_{t}
\end{align}
be the primary surplus (if it is negative, it represents a deficit) at time $t$ given history $\omega^{t} \in \Omega^{t}$.

\begin{theorem}
\label{thm:charac-CEG}
  Given $\omega_{0},B_{0} = b_{0}$ and $\phi_{-1}$, the outcome path $(d_{t},a_{t},B_{t+1},n_{t})_{t=0}^{\infty}$ is consistent with a competitive equilibrium iff for all $(t,\omega^{t}) \in \{0,1,2,...\} \times \Omega^{t}$, the following holds:
\begin{align}
   Z_{t}(\omega^{t}) u_{c}(\omega^{t}) + \phi_{t}(\omega^{t}) \{ p_{t}(\omega^{t}) u_{c}(\omega^{t}) B_{t+1}(\omega^{t}) - \delta_{t} u_{c}(\omega^{t}) B_{t}(\omega^{t-1}) \} \geq 0  \label{eqn:IC-time-3},\\
B_{t+1}(\omega^{t}) = B_{t}(\omega^{t-1}) \text{ if } \phi_{t}(\omega^{t}) = 0, \notag 
\end{align}
and
$c_{t}(\omega^{t})=\kappa_{t}(\omega^{t})n_{t}(\omega^{t})-g_{t}(\omega^{t})$
and equations \ref{eqn:eq_tax} and \ref{eqn:price-1} hold.
\end{theorem}

For any $(\omega,B,\phi)\in (\mathbb{G} \times \overline{\Delta}) \times
\mathbb{B} \times \{0,1\}$, let $CE_{\phi}(\omega,B)$ denote the set
of all outcome paths that are consistent with competitive equilibria, given $\omega_{0}
= \omega$, $\phi_{0}(\omega_{0}) = \phi$ and where $B$ is the
outstanding debt of time $0$, after any potential debt restructuring
in that period.\footnote{Constructing the set $CE_{\phi}(\omega,B)$ is useful since, in order to make a default/repayment decision, the default authority evaluates alternative utility values both for repayment and for autarky that are sustained by competitive equilibrium allocations.} We observe that by setting $\phi_{0}(\omega_{0}) = \phi$ we are implicitly imposing restrictions on $a_{0},d_{0},\phi_{-1}$ and $\delta_{0}$.\footnote{For example, if $\phi_{0}=1$ we could only arrive to it because either $\phi_{-1}=1$ and $d_{0}=0$, given $(g_{0},B_{0})=(g,B)$, or because $\phi_{-1}=0$ with defaulted debt $\widetilde{B}_{0} = B_{0}/\delta_{0}$ and offer $\delta_{0}$ is accepted ($a_{0}=1$). }

Equation \ref{eqn:IC-time-3} summarizes the budget constraint of the government but replacing prices and taxes by the first order conditions, as in the ``primal approach'' used by \cite{LS_JME83} and \cite{AMSS_JPE02}.

\section{The Government Problem}
\label{sec:govt}

The government is benevolent and maximizes the welfare of the representative household
by choosing policies. The government, however, cannot commit to
repaying the debt, but commits to previous tax promises until a debt restructuring takes place. This assumption facilitates the comparison to the optimal taxation literature in the spirit of \cite{LS_JME83} and \cite{AMSS_JPE02}.

For autarky states, the government chooses taxes that balance its budget. Once the government accepts an offer to restructure the debt, it regains access to financial markets and starts anew, without any outstanding tax
promises, by assumption. A similar feature is present in
\cite{DEBORTOLI_NUNES_WP08}, where the government can randomly
re-optimize and reset fiscal policies with a given exogenous
probability.\footnote{In our model, however, the \emph{resetting} event,
  given by the debt restructuring, is an equilibrium outcome that
  emerges endogenously.}

The government problem can thus be
viewed as a problem involving two types of authorities:
a \emph{default authority} and a \emph{fiscal authority}. On the one
hand, the default authority can be seen as comprised of a sequence of
one-period administrations, where the time-$t$ administration makes
the default and repayment decision in period $t$, taking as given the
behavior of all the other agents including the fiscal authority. On
the other hand, the \emph{fiscal authority} can be viewed as a
sequence of consecutive administrations, each of which stays in office
until there is a debt renegotiation. While ruling, a fiscal
administration has the ability to commit, and chooses the optimal
fiscal and debt processes, taking as given the behavior of the default
authority. When debt is renegotiated, the fiscal administration is
replaced by a new one, which is not bound by previous tax promises,
and is free to reset the fiscal and debt policy.\footnote{We focus
  exclusively on symmetric strategies for households, where all of
  them take the same decisions along the equilibrium path. Similarly,
  we assume that all default and fiscal administrations choose
  identical actions conditional on the same state of the economy, thereby introducing a Markovian structure for
  optimal strategies. } 

\subsection{The Government Policies}

For any $t \in \{0,1,...\}$, let $h_{t} \equiv
(\phi_{t-1},B_{t},\omega_{t})$ and
$h^{t} \equiv (h_{0},h_{1},....,h_{t})$ be the \emph{public history}
until time $t$.
We use $\mathbb{H}^{t}$ to denote the set of all
public histories until time $t$.

A \emph{government strategy} is given by a strategy for the default
and fiscal authorities, $\boldsymbol{\gamma} \equiv
(\boldsymbol{\gamma}^{D},\boldsymbol{\gamma}^{F})$. The strategy for
the default authority $\boldsymbol{\gamma}^{D}$ specifies a default
and a repayment decision for any period $t$ and any public history
$h^{t} \in \mathbb{H}^{t}$, i.e., $\boldsymbol{\gamma}^{D} =
(\gamma^{D}_{t}(\cdot))_{t=0}^{\infty}$ with $\gamma^{D}_{t}(h^{t})
\equiv ( d_{t}(h^{t}),a_{t}(h^{t}))$ for any $h^{t} \in \mathbb{H}^{t}$. The strategy for the fiscal
authority, $\boldsymbol{\gamma}^{F}$, specifies next period's debt
level for any public history $h^{t} \in \mathbb{H}^{t}$ and any
$\phi_{t}$, i.e., $\boldsymbol{\gamma}^{F} =
(\gamma^{F}_{t}(\cdot,\cdot))_{t=0}^{\infty}$ with
$\gamma^{F}_{t}(h^{t},\phi_{t}) \equiv  B_{t+1}(h^{t},\phi_{t})$ for
any $(h^{t},\phi_{t}) \in \mathbb{H}^{t} \times \{0,1\}$. The fact
that $\gamma^{F}_{t}(h^{t},\phi_{t})$ depends on $\phi_{t}$ reflects
our assumption on the timing protocol by which the default authority
moves first in each period.\footnote{We omit labor taxes (or labor
directly) as part of the government strategy because, given
$(h^{t},\phi_{t})$ and $\gamma^{F}_{t}(h^{t},\phi_{t})$, labor taxes
are obtained from the budget constraint. For this reason we do  not include them as part of the public history.}
To stress
that a particular policy action, say $B_{t+1}(h^{t},\phi_{t})$,
belongs to a given strategy we use $B_{t+1}(\boldsymbol{\gamma})
(h^{t},\phi_{t})$.

Let $\boldsymbol{\gamma}|_{(h^{t},\phi_{t})}$ denote the continuation of strategy
$\boldsymbol{\gamma}$ after history $(h^{t},\phi_{t}) \in \mathbb{H}^{t} \times
\{0,1\}$.\footnote{Observe that while a strategy prescribes that the default authority
  moves first at $t = 0$, with the continuation strategy,
as we defined, the fiscal authority is moving first at $t$ and then
the default authority moves at $t + 1$.} We say a strategy
$\boldsymbol{\gamma}$ is \emph{consistent with a competitive
  equilibrium}, if after any $(h^{t},\phi_{t})  \in \mathbb{H}^{t}
\times \{0,1\}$, the outcome path generated by
$\boldsymbol{\gamma}|_{(h^{t},\phi_{t})}$ belongs to
$CE_{\phi_{t}}(\omega_{t},B)$ with outstanding debt 
$B= (\delta_{t}\phi_{t} +(1-\phi_{t}))
B_{t}(\boldsymbol{\gamma})
(h^{t-1},\phi_{t-1}(\boldsymbol{\gamma})(h^{t-1}))$.
That is, if there is full repayment (i.e. no default) or not borrowing at all in the current period $t$, the debt level $B$ is given by the bond holdings carried over from last period. Otherwise, if the government just regained access by accepting an offer $\delta_{t}$, $B$ is the restructured debt level. For any $h_{0} \in \mathbb{H}$ and $\phi_{0} \in \{0,1\}$, we use $\mathcal{S}(h_{0},\phi_{0})$ to denote the set of such strategies; see appendix \ref{app:govt} for the formal expression of $\mathcal{S}(h_{0},\phi_{0})$. Henceforth, we only consider strategies that are consistent with competitive equilibrium.

Finally, for any public history $h^{t} \in \mathbb{H}^{t}$, $\phi \in \{0,1\}$ and $\boldsymbol{\gamma} \in \mathcal{S}(h_{0},\phi)$, let
\begin{align}
  \label{eq:8}
  V_{t}(\boldsymbol{\gamma})(h^{t},\phi) = E_{\Pi(\cdot|\omega_{t})} \left[
    \sum_{j=0}^{\infty} \beta^{j} u(\kappa_{t+j}(\omega^{t+j})
    n_{t+j}(\boldsymbol{\gamma})(\omega^{t+j}) - g_{t+j},1-n_{t+j}(\boldsymbol{\gamma})(\omega^{t+j}))  \right]
\end{align}
be the expected lifetime utility of the representative household at time $t$, given strategy $\boldsymbol{\gamma}|_{(h^{t},\phi)}$. 


\subsection{Default and Renegotiation Policies}

As mentioned before, the default authority can be viewed as comprised by a sequence of
one-period administrations, each of which makes the default and renegotiation decision in its respective period, taking as given the
behavior of all the other agents including the other default and fiscal administrations. It is
easy to see that, for each public history $h^{t} \in \mathbb{H}^{t}$, the default authority will optimally choose as follows: if $\phi_{t-1}=1$
\begin{align} \label{eqn:PF-d1}
  d^{\ast}_{t}(\boldsymbol{\gamma})(h^{t}) = \left\{
    \begin{array}{cc}
      0 & if~ V_{t}(\boldsymbol{\gamma})(h^{t},1) \geq V_{t}(\boldsymbol{\gamma})(h^{t},0) \\
      1 & if ~ V_{t}(\boldsymbol{\gamma})(h^{t},1) < V_{t}(\boldsymbol{\gamma})(h^{t},0)
    \end{array}
\right.
\end{align}
and if $\phi_{t-1}=0$
\begin{align} \label{eqn:PF-a1}
  a^{\ast}_{t}(\boldsymbol{\gamma})(h^{t}) = \left\{
    \begin{array}{cc}
      1 & if~ V_{t}(\boldsymbol{\gamma})(h^{t},1) \geq V_{t}(\boldsymbol{\gamma})(h^{t},0) \\
      0 & if ~ V_{t}(\boldsymbol{\gamma})(h^{t},1) < V_{t}(\boldsymbol{\gamma})(h^{t},0)
    \end{array}
\right.
\end{align}
The
dependence on $\boldsymbol{\gamma}$ denotes the fact that
$d^{\ast}_{t}$ and $a^{\ast}_{t}$ are associated with the strategy of
the fiscal authority $\boldsymbol{\gamma}^{F}$. Indeed, to specify the optimal
default and repayment decisions at any history $h^{t} \in \mathbb{H}^{t}$ we need to know the value of repayment and the value of default, $V_{t}(\boldsymbol{\gamma})(h^{t},1)$ and $V_{t}(\boldsymbol{\gamma})(h^{t},0)$, respectively, which are functions of $\boldsymbol{\gamma}^{F}$.\footnote{Also, recall that by assumption
	$a^{\ast}_{t}(\boldsymbol{\gamma})(h^{t}) = 0$ if $\phi_{t-1}=1$ or
	$\delta_{t} = \bar{\delta}$ and
	$d^{\ast}_{t}(\boldsymbol{\gamma})(h^{t}) = 1$ if $\phi_{t-1}=0$. }

\subsection{Recursive Representation of the Government Problem}

Taking as given the optimal decision rules $\ref{eqn:PF-d1}$ and
$\ref{eqn:PF-a1}$ for the default authority, we now turn to the
optimization problem of the fiscal authority and the recursive
representation of the government problem. To do so, we adopt a
recursive representation for the competitive equilibria by introducing
an adequate state variable. 
Following \cite{KYDLAND_PRESCOTT_JEDC80} and
\cite{CHANG_JET98} among others, it follows that the relevant
(co-)state variable is the ``promised''  marginal utilities of consumption.\footnote{By keeping track of the profile of ``promised'' marginal utilities of consumption, we ensure that the fiscal authority commits to deliver the ``promised'' marginal utility as long as the default authority does not restructure the debt. If the debt is restructured and a new
	fiscal administration takes power, it sets the current marginal
	utility at its convenience, which in equilibrium is anticipated by
	the households.}



For any $h_{0} = (\phi_{-1},B_{0},g_{0},\delta_{0}) \in \mathbb{H}$ and $\phi \in \{0,1\}$, let $\Omega(h_{0},\phi)$ be the set of all marginal
utility values $(\mu)$ and lifetime utilities ($v$) that can be sustained
in a competitive equilibrium, wherein the default authority reacts
optimally from next period on; see appendix \ref{app:govt} for the formal expression of $\Omega(h_{0},\phi)$. This set differs from the standard set of equilibrium promised marginal utilities in
\cite{KYDLAND_PRESCOTT_JEDC80} along some dimensions. In particular,
in  an standard Ramsey problem it would suffice to only specify the set of
promised marginal utilities, but in our framework with endogenous
default decisions we find it necessary to also specify continuation values to
evaluate alternative courses of action of the default authority. By
the same token, we compute this set for any $\phi$, even for the value
of $\phi$ not optimally chosen by the default authority through its
policy action.

For any $(g,B,\mu) \in \mathbb{G} \times \mathbb{B} \times \mathbb{R}_{+}$, let $V^{\ast}_{1}(g,B,\mu)$ be the value of a fiscal authority that had access to financial markets last period and continue to have it this current period (i.e., $\phi_{-1} = \phi = 1$) and that takes as given the optimal behavior of the default and subsequent fiscal authorities, with outstanding debt $B$ and a promised marginal utility of $\mu$ and government expenditure $g$. Similarly, let $V^{\ast}_{0}(g,B)$ be the value of a fiscal authority that does not have access to financial markets (i.e., $\phi=0$). 
Observe that since in financial autarky the government ought to run a balanced budget, $V^{\ast}_{0}$ does not depend on $\mu$.

Finally, let $\overline{V}^{\ast}_{1}(g,\delta B)$ be the value
 of a ``new'' fiscal authority (i.e., when $\phi_{-1} = 0$ and $\phi=1$) that takes as given the optimal behavior of the default and subsequent fiscal authorities, when an offer $\delta$ is accepted, government spending equals $g$ and the outstanding defaulted debt is $B$. In this case the fiscal authority does not have any outstanding ``promised'' marginal utility and thus it sets the current marginal utility at its convenience. Thus, by definition of $\Omega$,
\begin{align}
  \overline{V}^{\ast}_{1}(g, \delta B) = \max \{ v | (\mu,v) \in \Omega(0,B,g,\delta ,1) \},
\end{align}
and let $\overline{\mu}(g,\delta,B)=\{\mu|(\mu,\overline{V}^{\ast}_{1}(g,\delta B)) \in \Omega(0,B,g,\delta,1)\}$ be the associated marginal utility.

Given the aforementioned value functions, the optimal policy functions of the default authority in expressions \ref{eqn:PF-d1} and \ref{eqn:PF-a1} become\footnote{As indicated before, by assumption, $\mathbf{d}^{\ast}(g,B,\mu)=1$ if $\phi_{-1}=0$ and $\mathbf{a}^{\ast}(g,\delta,B) =0$ if $\phi_{-1}=1$ or $\delta = \bar{\delta}$.}
\begin{align} \label{eqn:PF-d}
  \mathbf{d}^{\ast}(g,B,\mu) = \left\{
    \begin{array}{cc}
      0 & if~ V^{\ast}_{1}(g,B,\mu) \geq V^{\ast}_{0}(g,B) \\
      1 & if ~ V^{\ast}_{1}(g,B,\mu) < V^{\ast}_{0}(g,B)
    \end{array}
\right.
\end{align}
and
\begin{align}\label{eqn:PF-a}
  \mathbf{a}^{\ast}(g,\delta,B) = \left\{
    \begin{array}{cc}
      1 & if~ \overline{V}^{\ast}_{1}(g,\delta B) \geq V^{\ast}_{0}(g,B) \\
  0 & if ~ \overline{V}^{\ast}_{1}(g,\delta B) < V^{\ast}_{0}(g,B)
    \end{array}
\right.
\end{align}

The next theorem presents a recursive formulation for the value
functions. In what follows, we denote the marginal utility of consumption in financial autarky as $m_{A}(g)$ for all $g \in \mathbb{G}$.\footnote{Formally, $m_{A}(g)=u_{c}(\kappa \mathbf{n}^{\ast}_{0}(g)-g,1-\mathbf{n}^{\ast}_{0}(g))$ where $\mathbf{n}^{\ast}_{0}(g) = \arg\max_{n \in [0,1]} \{ u(\kappa n-g,1-n) :  z(\kappa,n,g) = 0\}$. See appendix \ref{app:govt} for details.}


\begin{theorem}
\label{thm:VF-rec}
The value functions $V^{\ast}_{0}$ and $V^{\ast}_{1}$ satisfy the following recursions: For any $(g,B,\mu)$,
  \begin{align}
    V^{\ast}_{1}(g,B,\mu) = \max_{(n,B',\mu'(\cdot)) \in \Gamma(g,B,\mu)} \left\{ u(n-g,1-n) + \beta \int_{\mathbb{G}} \max\{ V^{\ast}_{1}(g',B',\mu'(g')) , V^{\ast}_{0}(g',B') \}  \pi_{\mathbb{G}}(dg'|g)  \right\} ,
  \end{align}
and
\begin{align}
   V^{\ast}_{0}(g,B) =&   u(\kappa \mathbf{n}^{\ast}_{0}(g)-g,1-\mathbf{n}^{\ast}_{0}(g)) + \beta \lambda \int_{\mathbb{G}} \int_{\Delta} \max\{ \overline{V}^{\ast}_{1}(g', \delta' B) , V^{\ast}_{0}(g',B) \}  \pi_{\Delta}(d\delta') \pi_{\mathbb{G}}(dg'|g) \notag \\
 & + \beta (1-\lambda) \int_{\mathbb{G}} V^{\ast}_{0}(g',B) \pi_{\mathbb{G}}(dg'|g)
\end{align}
where,
\begin{align}\notag
  \Gamma(g,B,\mu) = & \left\{ (n,B',\mu'(\cdot)) \in [0,1] \times \mathbb{B} \times \mathbb{R}^{|\mathbb{G}|} :     \right.\\ \notag
& \left. (B',\mu'(g'),V^{\ast}_{1}(g',B',\mu'(g')) \in
  Graph(\Omega(1,\cdot, g',1,1)),~\forall g' \in \mathbb{G} \right.\\
 & \left. \mu = u_{c}(n-g,1-n)~and~z(1,n,g)\mu + \mathcal{P}^{\ast}_{1}(g,B',\mu'(\cdot))B' - B \mu \geq 0  \right\}
\end{align}
and, for any $(B',\mu'(\cdot))$,
\begin{align*}
\mathcal{P}^{\ast}_{1}(g,B',\mu'(\cdot)) =& \beta \int_{\mathbb{G}} \left( (1-\mathbf{d}^{\ast}(g',B',\mu'(g')))\mu'(g')  + \mathbf{d}^{\ast}(g',B',\mu'(g')) m_{A}(g') \mathcal{P}^{\ast}_{0}(g',B')\right)  \pi_{\mathbb{G}}(dg'|g)\\
  \mathcal{P}^{\ast}_{0}(g,B') = &  \beta \int_{\mathbb{G}} \left(  \int_{\Delta} \overline{\mu}(g^{\prime},\delta^{\prime},B') \delta'  \mathbf{a}^{\ast}(g',\delta',B') \pi_{\Delta}(d\delta') +  \pi^{\ast}_{A}(g',B') m_{A}(g') \mathcal{P}^{\ast}_{0}(g',B') \right) \pi_{\mathbb{G}}(dg'|g)
\end{align*}
where $\pi^{\ast}_{A}(g,B) \equiv \left\{ (1-\lambda) + \lambda \int_{\Delta} (1- \mathbf{a}^{\ast}(g ,\delta ,B)) \pi_{\Delta}(d\delta)    \right\}$.
\end{theorem}

Below we present some particular cases of special interest in order to illustrate the objects in the previous theorem.

\begin{example}[non-defaultable debt]
Consider an economy with (ad hoc) risk-free debt. The value function $V^{\ast}_{0}$ is irrelevant and $V^{\ast}_{1}$ boils down to
\begin{align*}
  V^{\ast}_{1}(g,B,\mu) = \max_{(n,B',\mu'(\cdot)) \in \Gamma(g,B,\mu)} \left\{ u(n-g,1-n) + \beta \int_{\mathbb{G}}  V^{\ast}_{1}(g',B', \mu'(g')) \pi_{\mathbb{G}}(dg'|g)   \right\}
\end{align*}
where
\begin{align*}
   \Gamma(g,B,\mu)  =  \left\{
    \begin{array}{l}
      (n,B',\mu'(\cdot)) : z(1,n,g) \mu + \beta E_{\pi_{\mathbb{G}}(\cdot|g)}[\mu'(g')] B' - B \mu \geq 0 \\
 ~where~\mu = u_{c}(n - g, 1-n)
    \end{array}
\right\}.
\end{align*}
In addition, since there is no ``re-setting'' time, $\bar{V}^{\ast}_{1}$ coincides with  the value function at time 0 with $\mu$ chosen optimally. This case is precisely the type of model studied in \cite{AMSS_JPE02}.$\square$
\end{example}

The following example is analogous to the models studied by \cite{ARELLANO_AER08} and \cite{AG_JIE06} among others, but with a ``non-standard" per-period utility that reflects the distortive nature of the labor tax.

\begin{example}[quasi-linear per-period payoff, $\lambda \geq 0$, and $\pi_{\Delta}=\mathbf{1}_{\{0\}}$] \label{exa:quasi}
 Assume that $u(c,1-n) = c + H(1-n)$ for some function $H$ consistent
 with assumption \ref{ass:U_prop}. Under this assumption, $\mu$ can be dropped
 as a state variable since $u_{c}=1$ and thus it does not affect the pricing equation. In this case, the value during financial autarky is given by
\begin{align*}
  V^{\ast}_{0}(g) = & \kappa \textbf{n}^{\ast}_{0}(g) - g  + H(1- \textbf{n}^{\ast}_{0}(g))  + \beta \int _{\mathbb{G}} \left( \lambda V^{\ast}_{1}(g',0) + (1-\lambda) V^{\ast}_{0}(g')\right) \pi_{\mathbb{G}}(dg' |g ).
\end{align*}
Note that there is no need to keep
the debt $B$ as part of the state during financial autarky since
none of the defaulted debt is ever repaid, and all the offers of zero repayment are accepted by the government. The value during financial access is given by
\begin{align*}
  V^{\ast}_{1} (g,B) = \max_{(n,B') \in \Gamma(g,B)} \left\{ n-g + H(1-n) + \beta \int_{\mathbb{G}}  \max\{ V^{\ast}_{1}(g',B') , V^{\ast}_{0}(g') \} \pi_{\mathbb{G}}(dg'|g) \right\},
\end{align*}
where $\Gamma(g,B) \equiv \{ (n,B')~:~z(1,n,g) + \beta E_{\pi_{\mathbb{G}}(\cdot|g)}[\mathbf{1}_{ \{ g' :  V^{\ast}_{1}(g',B') \geq  V^{\ast}_{0}(g') \} }(g')]B' - B \geq 0 \}$.

Moreover, assuming $H'(1) <1$ and $2 H''(l) < H'''(l)(1-l)$,  we can view the
government problem as
directly choosing tax revenues $R$ with a per-period payoff given
by $W_{\kappa}(R) = \kappa n_{\kappa}(R) + H(1-n_{\kappa}(R))$ where $n_{\kappa}(R)$ is the
amount of labor needed to collect revenues equal to $R$, given
$\kappa$. Under our assumptions, $W_{\kappa}$ is non-increasing and
concave function. The Bellman equation of the value of repayment is given by
\begin{align}
  \label{eq:1}
  V^{\ast}_{1} (g,B) = \max_{(R,B')} \left\{ W_{1}(R) -g + \beta \int_{\mathbb{G}}  \max\{ V^{\ast}_{1}(g',B') , V^{\ast}_{0}(g') \} \pi_{\mathbb{G}}(dg'|g) \right\},
\end{align}
subject to $R + \beta E_{\pi_{\mathbb{G}}(\cdot|g)}[\mathbf{1}_{ \{ g' :  V^{\ast}_{1}(g',B') \geq  V^{\ast}_{0}(g') \} }(g') ]B' \geq g + B$, and
\begin{align}
  \label{eq:2}
   V^{\ast}_{0}(g) = & W_{\kappa}(g) -g + \beta \int _{\mathbb{G}} \left( \lambda V^{\ast}_{1}(g',0) + (1-\lambda) V^{\ast}_{0}(g')\right) \pi_{\mathbb{G}}(dg' |g ).
\end{align}

This problem is analogous to that studied in
\cite{ARELLANO_AER08} and \cite{AG_JIE06} among others, where the
government chooses how much to ``consume'', captured by $-R$, given an exogenous
process of ``income'', $-g$. An important difference, however, is the non-standard per-period payoff which
has a satiation
point at $R=0$ (i.e., zero distortive taxes).\footnote{Another subtle difference with
  the standard sovereign default literature is that while in our economy
  government and bondholders share the same preference, in this literature
 they do not. In particular, the government tends to be more impatient than (foreign) investors, thus
  bringing about incentives to front-load consumption through borrowing.}
$\square$
\end{example}

	In the previous example, the expression for the price function, $E_{\pi_{\mathbb{G}}(\cdot|g)}[\mathbf{1}_{ \{ g' :  V^{\ast}_{1}(g',\cdot) \geq  V^{\ast}_{0}(g') \} }(g') ]$, highlights an important
	difference between our default model and a model with risk-free debt
	such as \cite{AMSS_JPE02} (henceforth, AMSS). Since $u_{c}=1$, the market
	stochastic discount factor is equal to $\beta$, and thus in the AMSS model the
	government cannot manipulate the return of the discount bond. In our
	economy, however, the government is still able to manipulate the return of the discount bond by altering its payoff through the decision of default.
	

\section{Analytical Results}
\label{sec:bench}

We present analytical results for a benchmark
model characterized by quasi-linear per-period utility and i.i.d. government expenditure shocks. 
The proofs for the results are gathered in appendix \ref{app:bench}.

\begin{assumption}\label{ass:linear_c}
  (i) $\kappa = 1$; (ii) $u(c,n) = c + H(1-n)$ where $H \in \mathbb{C}^{2}((0,1),\mathbb{R})$ with $H'(0) = \infty$, $H'(l) > 0$, $H'(1)<1$, $H''(l) < 0$ and $2 H''(l) < H'''(l)(1-l)$
\end{assumption}

Part (i) implies that there are no direct cost of defaults in terms of output. Part (ii) of this assumption imposes that the per-period utility of the households is quasi-linear and it is analogous to assumption in p. 10 in AMSS. As noted above, under this assumption, $\mu$ can be dropped as a state variable. This implies that the value functions $V^{\ast}_{0}$, $V^{\ast}_{1}$ are only functions of $(g,B)$ and the same holds true for the optimal policy functions. We also assume that government expenditure are i.i.d., formally
\begin{assumption}\label{ass:g-iid}
 For any $g' \ne g$, $\pi_{\mathbb{G}}(\cdot|g) = \pi_{\mathbb{G}}(\cdot|g')$.
\end{assumption}

With a slight abuse of notation and to simplify the exposition we use $\pi_{\mathbb{G}}(\cdot)$ to denote the probability measure of $g$. Finally, to further simplify the technical details, we assume that
$\mathbb{B}$ has only finitely many points, unless stated otherwise.\footnote{This assumption is made for simplicity. It can be relaxed to allow for general compact subsets, but some of the arguments in the proofs will have to be changed slightly. Also, $\mathbb{B} \equiv \{B_{1},...,B_{|\mathbb{B}|}\}$ is only imposed for the government; the households can still choose from convex sets; only in equilibrium we impose $\{B_{1},...,B_{|\mathbb{B}|}\}$.} For the rest of the section, assumptions \ref{ass:linear_c} and \ref{ass:g-iid} hold and will not be referenced explicitly.

\subsection{Characterization of Optimal Default Decisions}

The next proposition characterizes the optimal decisions to default and
to accept offers to repay the defaulted debt as ``threshold
decisions''. These results extends those in \cite{ARELLANO_AER08} to our setting, where, among other things, we allow for partial repayments of government debt. Recall that $\mathbf{d}^{\ast}(g,B)$ and
$\mathbf{a}^{\ast}(g,\delta,B)$ are the optimal decision of default
and of renegotiation, respectively, given the state $(g,\delta,B)$.

\begin{proposition}
\label{pro:opt-dec}
  There exists $\bar{\lambda}$ such that for all $\lambda \in [0, \bar{\lambda}]$, the following holds:
  \begin{enumerate}
  \item There exists a $\hat{\delta} : \mathbb{G} \times \mathbb{B} \rightarrow \Delta$ such that $\mathbf{a}^{\ast}(g,\delta,B) = \mathbf{1}_{\{\delta :  \delta \leq \hat{\delta}(g,B)\}}(\delta)$ and $\hat{\delta}$ non-increasing as a function of $B$.
  \item There exists a $\bar{g} : \mathbb{B}
    \rightarrow \mathbb{G}$ such that $\mathbf{d}^{\ast}(g,B) =
    \mathbf{1}_{\{ g :  g \geq \bar{g}(B) \}}(g)$ and $\bar{g}$
    non-increasing for all $B > 0$.
  \end{enumerate}
\end{proposition}

This result shows that for a (non-trivial) range of probabilities of
receiving outside offers, $\lambda \in [0,
\bar{\lambda}]$, default is more likely to occur for high levels of debt, and so are rejections of offers to exit financial autarky. The latter result implies that the average recovery rate, $E_{\pi_{\mathbb{G}}}[\int_{\delta' \in \Delta} \delta' \mathbf{1}_{\{\delta~:~ \delta \leq \hat{\delta}(g,B) \}}(\delta') \pi_{\Delta}(d\delta')]$,  is decreasing in the level of debt, as documented by \cite{YUE_WP07} in the data. It also follows that other things equal, higher debt levels are on average associated with longer financial autarky periods. 
 Thus, these two results imply a positive co-movement between the (observed) average haircut and the average length of financial autarky. This last fact is consistent with the data: See fact 3 in \cite{BENJAMIN_WRIGHT_WP09} and also \cite{Cruces_13} found a similar relationship for 180 sovereign debt restructuring cases of 68 countries between 1970 and 2010.\footnote{It is important to note, however, that we derived the implications by looking at \emph{exogenous} variations of the debt level; in the data this quantity is endogenous and, in particular, varies with $g$. This endogeneity issue taken into account in the numerical simulations.}

\subsection{Implications for Equilibrium Prices and Taxes}

We now study the implications of the above results on equilibrium
prices and taxes.\\

\textbf{Equilibrium prices and endogenous debt limits.}  Under assumption \ref{ass:g-iid} equilibrium prices do not depend on $g$, i.e., $ \mathcal{P}^{\ast}_{\phi}(\cdot) \equiv  \mathcal{P}^{\ast}_{\phi}(g,\cdot)$ for any $g \in \mathbb{B}$. By proposition \ref{pro:opt-dec} it follows that, for any $B' \in \mathbb{B}$,
\begin{align}\label{eqn:p-bench1}
  \mathcal{P}^{\ast}_{1}(B') = & \beta \int_{\mathbb{G}} \mathbf{1}_{\{ g' \leq \bar{g}(B') \}}(g')   \pi_{\mathbb{G}}(dg') + \left( \beta \int_{\mathbb{G}} \mathbf{1}_{\{ g' > \bar{g}(B') \}}(g')   \pi_{\mathbb{G}}(dg') \right) \mathcal{P}^{\ast}_{0}(B')
\end{align}
and\footnote{See lemma \ref{lem:T-q}(3) in the appendix for the derivation.}
\begin{align} \label{eqn:q-bench}
  \mathcal{P}^{\ast}_{0}(B) =  \frac{\beta \lambda   \int_{\Delta} \left( \int_{\mathbb{G}} \mathbf{1}_{\{\delta :  \delta \leq \hat{\delta}(g',B) \}}(\delta) \pi_{\mathbb{G}}(dg') \right) \delta \pi_{\Delta}(d\delta) }{1 - \beta + \beta\lambda \int_{\Delta} \int_{\mathbb{G}}  \mathbf{1}_{\{\delta :  \delta \leq \hat{\delta}(g',B) \}}(\delta)  \pi_{\mathbb{G}}(dg') \pi_{\Delta}(d\delta) }.
\end{align}

A key feature of endogenous default models is the existence of endogenous borrowing limits. A necessary condition for this result to hold is that, due to the possibility of default, equilibrium prices are non-increasing as a function of debt, thus implying a ``Laffer-type curve'' for the revenues coming from selling bonds. In an economy without debt repayment (e.g., $\pi_{\Delta} = \mathbf{1}_{\{0\}}$), it follows that $\mathcal{P}^{\ast}_{0}=0$ and $\mathcal{P}^{\ast}_{1}(B') =  \beta \int_{\mathbb{G}} \mathbf{1}_{\{g :  g \leq \bar{g}(B') \}}(g')   \pi_{\mathbb{G}}(dg') $ which is non-increasing in $B'$ by proposition \ref{pro:opt-dec}. Moreover, it takes value zero for sufficiently high $B'$. Therefore, there exists an endogenous debt limit, i.e., finite value of $B'$ that maximize the debt revenue $\mathcal{P}^{\ast}_{1}(B')B'$. In an economy where we allow for debt renegotiation, by inspection of equation \ref{eqn:p-bench1} and the fact that $\mathcal{P}^{\ast}_{0} \geq 0$, it is easy to see that, other things equal, the previous result is attenuated by
the presence of (potential) defaulted debt payments and secondary
markets.  
The next proposition shows that for (non-random) repayment offers, the price is non-increasing on the level
of debt and there are endogenous borrowing limits. Hence, high levels of debt are associated with higher
return on debt, both before and during financial autarky.

\begin{proposition}
\label{pro:preqn-mon0}
  Suppose $\pi_{\Delta}(\cdot) = \mathbf{1}_{\delta_{0}}(\cdot)$ for some $\delta_{0} \in [0,1]$. Then there exists a $\bar{\lambda} > 0$, such that for all $\lambda \in
  [0,\bar{\lambda}]$, $\mathcal{P}^{\ast}_{i}(\cdot)$ is
  non-increasing for $B>0$ and for $i=0,1$. 
\end{proposition}



This result is consistent with the evidence of the positive relationship between debt-to-output levels and default risk measures; see section \ref{sec:facts} in the supplementary material. In addition, the existence of endogenous borrowing limits implies that the ability to roll over high levels of debt is hindered. Since the primary surplus function $z(1,\cdot,g)$ is concave in $n$, as shown in lemma \ref{lem:charac-z} of appendix \ref{app:govt}, labor is more ``sensitive'' to fluctuations in government expenditure when the indebtedness level is high. This feature is consistent with the stylized fact that on average higher volatility of tax revenue-to-output ratios is observed when debt and default risk are high. We further explore this mechanism in the numerical simulations.\\

\textbf{Default risk and the law of motion of equilibrium
  taxes.} \label{sec:LOM} In order to analyze the ex-ante effect of default risk on
the law of motion of taxes, we look at the case $\lambda = 0$ (i.e.,
autarky is an absorbing state) to simplify the analysis. We also strengthen assumption \ref{ass:linear_c} by requiring that $H''(l) < H'''(l)(1-l)$. By proposition
\ref{pro:opt-dec}, the default decision is a threshold decision, so
for each history $\omega^{\infty} \in \Omega^{\infty}$ we can define
$T(\omega^{\infty}) = \inf\{ t : g_{t} \geq
\bar{g}(B_{t}(\omega^{t-1}))\}$ (it could be infinity) as the first
time the economy enters in default. For all $t \leq T(\omega^{\infty})$ the
economy is not in financial autarky, and the implementability constraint is given by
\begin{align*}
  B_{t}(\omega^{t-1})  + g_{t} \leq & \left(1- H'(1-n_{t}(\omega^{t}))   \right) n_{t}(\omega^{t})  + \mathcal{P}^{\ast}_{1}(B_{t+1}(\omega^{t})) B_{t+1}(\omega^{t}),
\end{align*}
where $\mathcal{P}^{\ast}_{1}(g_{t},B_{t+1}(\omega^{t})) \equiv E_{\pi_{\mathbb{G}}}[1-\mathbf{d}^{\ast}(g',B_{t+1}(\omega^{t}))]
$. Let $\nu_{t}(\omega^{t})$ be the Lagrange multiplier associated to this
restriction in the optimization problem of the government, given
$\omega^{t} \in \Omega^{t}$. In appendix \ref{app:LM-LoM} we derive the FONC of the government and provide a closed-form expression for $\nu_{t}(\omega^{t})$ as a \emph{decreasing} nonlinear function of $n_{t}(\omega^{t})$; see equation \ref{eqn:nu-labor}. Hence, as noted by AMSS, studying the law of motion of $\nu_{t}$ we can shed light on the law of motion of taxes.\footnote{Under our assumptions, $\tau_{t}$ are decreasing in $n_{t}(\omega^{t})$ which, in turn, implies a positive relationship between $\nu_{t}$ and $\tau_{t}$.}



From the FONC of the government it follows (see appendix \ref{app:LM-LoM} for the derivation)\footnote{This derivation assumes that $\mathbb{B}$ is a convex set and $\pi_{\mathbb{G}}$ has a density with respect to the Lebesgue measure, so as to make sense of differentiation. It also assumes differentiability of $V^{\ast}$.}$^{,}$\footnote{The martingale property is also preserved if capital is added to the economy; see \cite{FARHI_WP07}.}
\begin{align}\label{eqn:LM-lom}
  \nu_{t}(\omega^{t}) \left( 1  + \frac{d\mathcal{P}^{\ast}_{1}(B_{t+1}(\omega^{t}))}{dB_{t+1}} \frac{B_{t+1}(\omega^{t})}{\mathcal{P}^{\ast}_{1}(B_{t+1}(\omega^{t}))} \right) = \int_{\mathbb{G}} \nu_{t+1}(\omega^{t},g') \frac{1\{ g' \leq \bar{g}(B_{t+1}(\omega^{t})) \}}{\int_{\mathbb{G}} 1\{ g' \leq \bar{g}(B_{t+1}(\omega^{t})) \} \pi_{\mathbb{G}}(dg') } \pi_{\mathbb{G}}(dg').
\end{align}

 Equation \ref{eqn:LM-lom} reflects the role of debt for tax-smoothing purposes and the trade-off the government is confronted with. The Lagrange multiplier associated with the implementability condition is constant in \cite{LS_JME83} and, thus, trivially a martingale. In \cite{AMSS_JPE02}, away from the asset limits, the Lagrange multiplier associated with the implementability condition is a martingale with respect to the probability measure $\pi_{\mathbb{G}}$; i.e., the government spreads over time the tax burden up to the point that $\nu_{t}$ is equal to $\nu_{t+1}$ in expected terms. Equation \ref{eqn:LM-lom} implies that in our economy this is not the case; the presence of default risk affects the law of motion of the Lagrange multiplier in two important ways. First, the expectation is computed
 under the so-called \emph{default-adjusted probability measure}, given by $\frac{1\{ \cdot \leq \bar{g}(B_{t+1}(\omega^{t})) \}}{\int_{\mathbb{G}} 1\{ g' \leq \bar{g}(B_{t+1}(\omega^{t})) \} \pi_{\mathbb{G}}(dg') } \pi_{\mathbb{G}}(\cdot)$. The default-adjusted probability measure is first-order dominated by $\pi_{\mathbb{G}}$, which reflects the fact that the option to default adds ``some'' degree of state-contingency to the payoff of the government debt; in particular it implies that only the states tomorrow in which there is repayment, are relevant  for the law of motion of $\nu_{t}$.
  Moreover, to the extent that $\nu_{t}(\omega^{t},.)$ is increasing, the left-hand side of the equation is lower than the expectation of $\nu_{t+1}$ under $\pi_{\mathbb{G}}$. Second, $\nu_{t}(\omega^{t})$ in the left-hand side is multiplied by $\left( 1  + \frac{d\mathcal{P}^{\ast}_{1}(B_{t+1}(\omega^{t}))}{dB_{t+1}} \frac{B_{t+1}(\omega^{t})}{\mathcal{P}^{\ast}_{1}(B_{t+1}(\omega^{t}))} \right)$, which can be interpreted as the ``markup'' that the government has to pay for having the option to default. In the presence of default risk the markup is less than one and hence $\nu_{t}$ is higher than the expectation of $\nu_{t+1}$ under the default-adjusted probability measure. This reflects the fact that default risk entails higher borrowing costs and limited debt issuance which prevent the government from smoothing taxes completely, as it would be the case if $\nu_{t}$ were a Martingale, like in the risk-free debt economy.

 These two forces act in opposite directions, and is not clear which one will prevail. In order to shed more light on this issue, in section \ref{sec:num-simul} we explore quantitatively this trade-off by studying the impulse responses of $\nu_{t}$.

\section{Numerical Results}
\label{sec:num-simul}

Throughout this section, we run a battery of numerical exercises in order to assess the impact of endogenous default risk on fiscal policy and the overall economy's dynamics. We compare our findings with an economy in which the option to default is not present---precisely the
model considered in \cite{AMSS_JPE02}. We denote the variables
associated with this model with a (sub)superscript ``AMSS'';
variables associated to our economy are denoted with a
(sub)superscript ``ED'' (short for Economy with Default).

A natural question that arises in this context is what characteristics of an economy will prompt it to behave as the AMSS- or ED-type models prescribe, in particular when it comes to the propensity or willingness to honor debt contracts. Several reasons have been put forth to explain why some governments always repay while some others do not. One reason is attributed to factors unrelated to the model, such as political instability and polarization, which could result in lower discounting of future consumption by the incumbent ruling party and hence stronger incentives to default.\footnote{See for example \cite{CS_WP08} and \cite{Derasmo}. In a small open political economy where different political groups compete with each other for access to government resources, \cite{Amador} arrives to a different conclusion. The same arguments why political interactions lead to over-spending actually provide incentives not to repudiate sovereign debt contracts.} An alternative explanation, in line with our model, is that for AMSS-type economies, default is more costly because they are financially more integrated, and the inability to borrow from capital markets after a default could have a more severe impact on financing of the firms, thus lowering their productivity (in our model represented by a lower $\kappa$).\footnote{Along this line, in a general equilibrium setup \cite{MY_WP08} generates an endogenous loss of output resulting from the substitution of imported inputs by less-efficient domestic ones as firms' credit lines are cut during default episodes. Instead, in \cite{Sosa_Padilla} a financial disruption in the banking sector occurs after a default event, as banks' balance sheets deteriorate, causing a domestic credit crunch followed by an output drop.} A third explanation is connected to income inequality and re-distributional motives. High income dispersion across domestic households provides stronger incentives to default to the government, expropriate wealth and redistribute those resources to reduce inequality.\footnote{\cite{DERASMO_MENDOZA} and \cite{DGS_WP15} elaborate along this dimension.} Although we consider these issues important, we think they are out of the scope of the present paper and do not explore them herein. We take them as given and proceed to characterize the optimal government policies in this economic environment.

The main focus of the numerical results is to describe some features pointed out in the normative analysis, and not to replicate business cycle dynamics of any particular economy or historical default event. The numerical simulations show that our model can spawn recurrent default episodes and is able to generate considerable levels of ``debt intolerance'' and high volatility of taxes, compared to an analogous AMSS economy.

%
%

\bigskip

\textbf{Parametrization and functional forms.} The utility function is given by $u(c,1-n) = c + C_{1} \frac{(1-n)^{1-\sigma}}{1 -
  \sigma}$. For this preference specification, it is easy to see that whether the fiscal authority has the ability to commit to tax policies or not renders the same equilibrium results. In this parametrization, we assume that the government expenditure $g_{t}$ follows an AR(1) log-normal process
  \begin{equation*}
  \log g_{t}=(1-\rho)\mu+\rho \log g_{t-1}+\sigma_{\varepsilon} \varepsilon_{t},~with~\varepsilon_{t} \sim N(0,1),
  \end{equation*}
  which is approximated by an 11-state Markov chain using \cite{TAUCHEN_EL86} procedure.\footnote{The debt state space $\mathbb{B}$ is constructed by discretizing $[0,0.4]$ into $800$ grid-points The one-period gross risk-free rate $1+r^{f}$ is equal to the reciprocal of the households' discount factor $\beta$ and bond spreads are computed as the differential between bond returns and the risk-free rate. Finally, as in AMSS, we rule out negative lump-sum transfers to the households. The model is solved numerically using value function iterations with a discrete state space and an ``outer'' loop that iterates on prices until convergence.}


We choose the parameters of the model as follows: $\beta = 0.97$, $\psi = 2$, $\kappa =0.998$ and $C_{1}=0.15$.\footnote{While the exogenous ad hoc output cost of default $1-\kappa$ is only 0.2 percent, the endogenous drop of consumption is significantly larger. Indeed, if the government runs a balanced budget, the autarkic level of consumption is over 1 percent lower than its counterpart in repayment for high government spending.} The parameter values for the stochastic process of the government expenditure are $\mu=0.114$, $\rho=0.56$ and $\sigma_{\varepsilon}=0.037$. For the debt restructuring process we choose a probability of receiving an offer $\lambda = 0.47$, and consider ten equally-probable renegotiation offers with equidistant haircuts ranging from $0.45$ to $0.9$.

We perform 5,000 Monte Carlo (MC) iterations, each consisting of sample
paths of 2,500 observations for which the first 500 observations were disregarded in order to eliminate the effect of the initial
conditions. We then compute the statistics across MC simulations.

\bigskip

\textbf{Behavior of debt.} In the numerical simulations, default occurs at an average frequency of around 1.8 percent.\footnote{The unconditional default frequency is computed as the sample mean of the number of default events in the simulations.} Households in the model anticipate the default strategies in equilibrium and demand higher returns to hold the bond. Facing higher borrowing costs, the government responds by issuing less debt. Consequently, the average level of indebtedness is significantly lower in our environment than in AMSS model, as illustrated in figure \ref{fig:debthisto}.
The histograms (which were smoothed using kernel methods) of debt-to-output ratios in our model (solid red) and in AMSS (dotted blue) are shown for different values of the government expenditure.\footnote{To construct the histogram, we disregarded few observations of debt-to-output ratios corresponding to bond spread levels exceeding 50 percent.} First, as the government expenditure increases, in both economies more debt is optimally chosen in order to finance it due to tax-smoothing motives. In addition, the debt-to-output ratio in our model becomes more concentrated the higher the value of $g$ is. This clustering feature is due to the fact that, as $g$ increases, the government wants to issue more debt, but faces tighter (endogenous) borrowing limits resulting from higher default risk. This behavior is in stark contrast with the one for AMSS, where the debt-to-output ratios are more spread out and can take higher values. In fact, for the bottom panels there is almost full separation between the histogram of our model and that of AMSS. It is worth to note as well, that even for low values of $g$ (e.g. the top panel) our model exhibits ``thinner right tail'' for the distribution of debt-to-output ratio, again illustrating the latent borrowing limits.

\begin{figure}
	\centering
	\includegraphics[height=3.1in,width=5.5in]{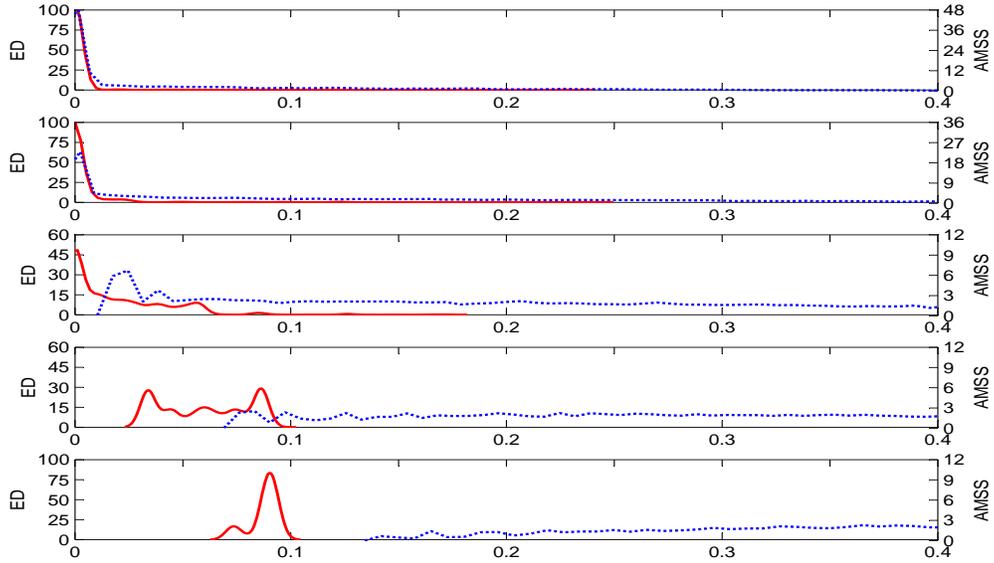}
	\caption{Histograms of debt-to-output ratio for our model (solid red) and AMSS (dotted blue) conditioned on different values of  $g$. From top to bottom, the five panels correspond to the second, fourth, sixth, eighth and tenth gridpoint of $g$.}
	\label{fig:debthisto}
\end{figure}

\medskip

\begin{figure}
  \centering
  \includegraphics[height=2.5in,width=4.6in]{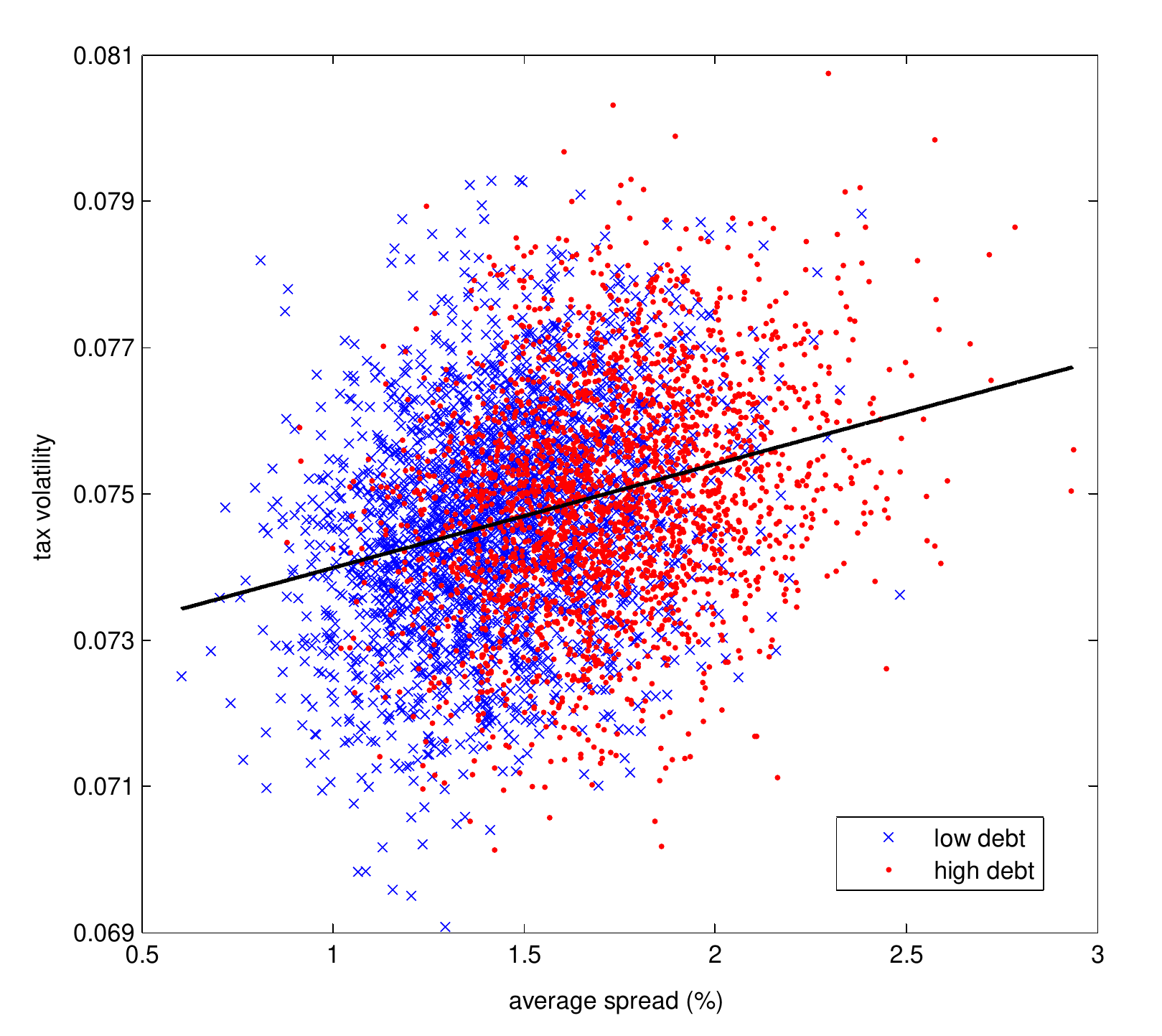}
  \caption{Standard deviation of tax rates and mean bond spreads in financial access for low debt (blue) and high debt (red), and fitted OLS line between them (black line).}
  \label{fig:taxvol_vs_spread_by_030714}
\end{figure}

\begin{figure}
  \centering
  \includegraphics[height=3.25in,width=5.5in]{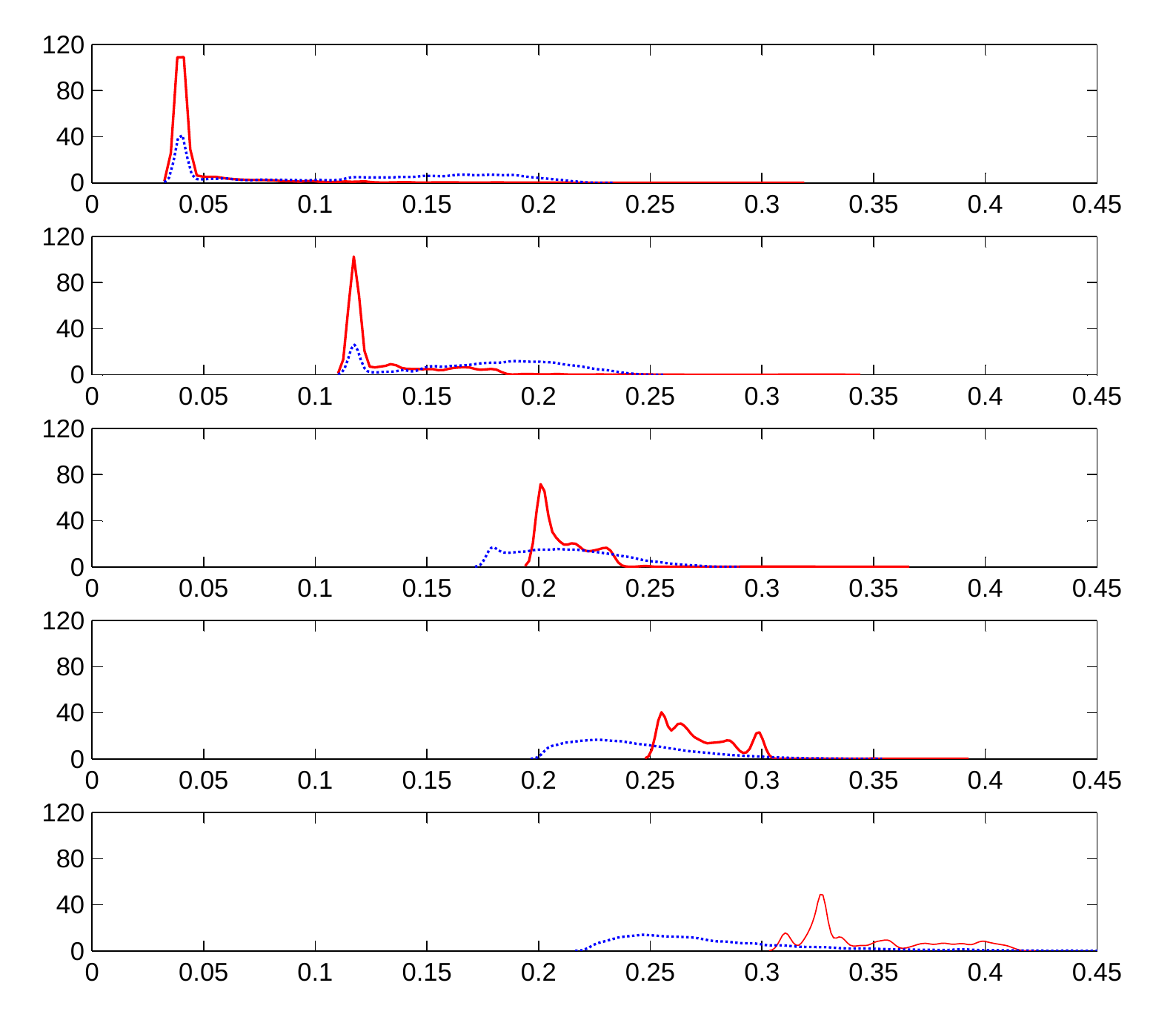}
  \caption{Histogram of tax rates in financial access for our model (solid red) and AMSS (dotted blue) conditioned on different values of  $g$. From top to bottom, the five panels correspond to the second, fourth, sixth, eighth and tenth realization of $g$.}
  \label{fig:taxeshistogram_030714}

\end{figure}

\textbf{Behavior of taxes.} As mentioned in section \ref{sec:bench}, we typically observe in the data that default risk and tax volatility are positively correlated and that both are higher for high levels of debt-to-output
ratio. Figure \ref{fig:taxvol_vs_spread_by_030714} shows that our model is able to generate this pattern. The blue (red) dots indicate the standard deviation of taxes and spreads for low (high) levels of debt, respectively. Each dot corresponds to the subsample of periods with financial access in a MC simulation. The solid black line represents the regression line. For low (high) debt we consider debt-to-output ratios below (over) the median of its asymptotic distribution. For both cases we see a positive
relationship between spreads and tax volatility, reflecting the government's limited ability to spread out the tax burden across time in the context of default risk.
In the simulations, the standard deviation of tax rates in financial access is roughly 50 percent higher than in AMSS model. Another feature that stands out in the figure is that the red dotted-cloud is shifted to the upper right corner of the graph with respect to the blue crossed-cloud, thus indicating that both spreads and tax volatility are higher for higher level of debts.

In order to shed more light about the behavior of taxes when there
is risk of default, we compare in figure \ref{fig:taxeshistogram_030714} the histograms (again smoothed using kernel methods) of tax rates in our model (solid red) with that in AMSS (dotted blue), for different values of government expenditure. First, as observed in the bottom panels, for high values of $g$ the distribution of taxes in our model
  is shifted to the right compared to that in AMSS model. This difference arises from the fact that in our environment, with default risk, debt is too costly for the government to finance high government expenditure and hence it has to resort to higher tax rates. In contrast, in the top panels, when the $g$ realization is low, the situation is
  reversed and now the distribution of taxes in AMSS model is shifted to the right relative to our model. During those states, the government repays the outstanding debt, which typically is higher in AMSS than in our credit-constrained economy. In addition, in our model taxes are more concentrated around a single peak (which shifts to the right with the level of government expenditure), reflecting more limited borrowing. In contrast, in AMSS the distribution of taxes is more spread-out for each $g$ realization but at the same time more less sensitive to changes in $g$, a clear reflection of more tax smoothing.

\medskip


\textbf{Debt renegotiation.} In table \ref{tab:debtrec} we present some statistics regarding the debt renegotiation process for different values of $\lambda$.
\begin{table}[h]
	\centering
	\begin{tabular}{c|cccccc}
		\hline  \hline $\lambda $ & 0.2  & 0.4  & 0.6  & 0.8  & 1.0\\ \hline
		Avg. offer accepted  & 0.60  & 0.59  & 0.59  & 0.58  & 0.57 \\
		Avg. duration $|$ High debt  & 10.08   & 6.69  & 6.03  & 5.19 & 5.06 \\
		Avg. duration $|$ Low  debt  & 9.46   & 5.82  & 3.42  & 3.16  & 2.92 \\
		\hline  \hline
	\end{tabular}
		\caption{MC Statistics for Debt Renegotiation for Different Values of $\lambda$. }
	\label{tab:debtrec}
\end{table}

In the first row, as the probability of receiving an offer increases, the average offer accepted decreases. This result follows because as $\lambda$ increases the option value of staying in financial autarky increases and thus the government accepts less offers. The last two rows in the table show the average duration, conditioning on the fact that the defaulted debt is ``high'' (second row) and ``low'' (third row).\footnote{As low (high) defaulted debt we consider debt-to-output ratios in the default episodes below (above) the unconditional median.} For both cases, it decreases as the probability of receiving an offer increases, but more importantly it shows that for ``high'' levels of debt we have, on average, longer financial autarky spells. In fact, the difference can be as large as 75 percent higher for intermediate values of $\lambda$. This result coincides with the implications of proposition \ref{pro:opt-dec}. Moreover, as we see from the table, differences in the duration can be non-negligible.


\medskip

\textbf{Impulse responses.} Figure \ref{fig:ir4-crop} presents the impulse response
for debt, fiscal primary surplus and taxes for our model and AMSS. The path of government expenditure is plotted in the first panel: the government expenditure is low and equal to $0.0915$ except for $t=2, 3, 4$ where it is high and equal to 0.159. The initial debt level is set to zero. While in both economies the government accumulates debt in periods of higher government expenditure, in
our model it does so to a lesser extent due to the presence of endogenous borrowing limits. From $t=5$
onwards, when government expenditure becomes low, the debt is gradually paid back, eventually reaching zero.

Taxes behave analogously: in our economy taxes are higher than AMSS during the periods of
high government expenditure since borrowing is more limited and costly, but decrease more rapidly
when the realization of government expenditure becomes lower (see the third panel). Overall, not surprisingly, a
smoother behavior for taxes is observed in AMSS than in our economy, resulting from the absence of default risk. In both economies, a hike in the tax rate reduces after-tax wages inducing substitution from labor to leisure (recall that with quasi-linear preferences the labor supply curve is invariant to changes in wealth). Consequently, output drops. Debt is more persistent than tax rates in both models, a feature that stands out under incomplete markets, as stressed by \cite{MS_09}. Also, in line with the findings of \cite{MS_09}, both economies experience a significant fiscal deficit as $g$ increases and debt grows. In AMSS model, since the government has more borrowing capacity, it therefore issues more debt, and hence the fiscal deficit in the short run is relatively larger.


The last panel plots the behavior of the Lagrange multiplier $\nu_{t}$ of the implementability constraint, studied in subsection \ref{sec:LOM}. The fact that ours is above the one of AMSS for periods of high government expenditure reflects the ``mark-up'' effect mentioned in subsection
\ref{sec:LOM}. Also note that the Lagrange multiplier increases during these periods, reflecting the
fact that the shadow cost of debt is increasing in the level of debt. From $t=4$ onwards, as
debt decreases, the Lagrange multiplier in our model falls and eventually converges to the one of AMSS. 

\begin{figure}
  \centering
  \includegraphics[height=3.1in,width=5.5in]{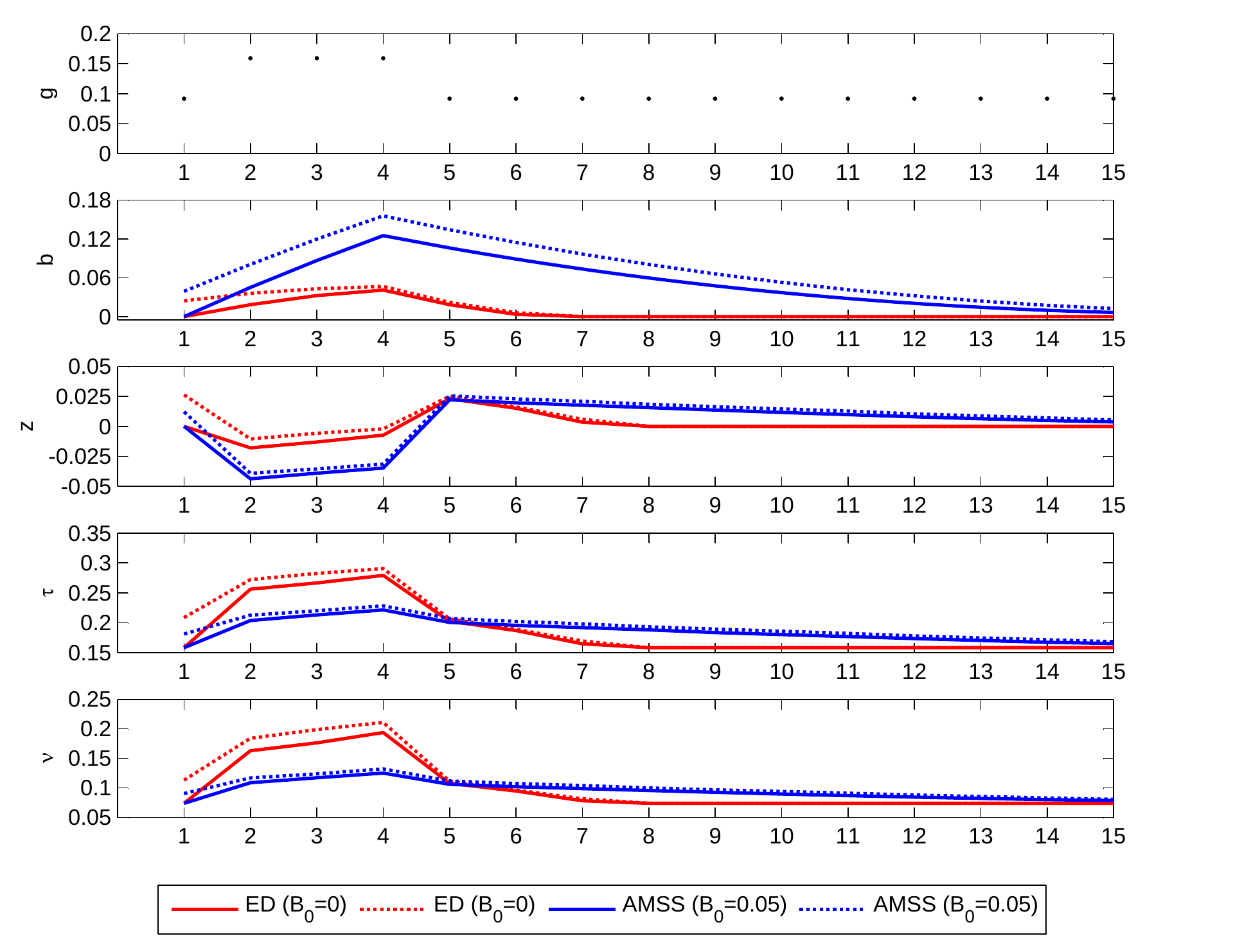}
  \caption{Impulse responses for our model (red) and AMSS (blue). Realization of government expenditure (first panel), debt path (second panel), primary surplus path (third panel), tax path (fourth panel), and Lagrange multiplier path (fifth path).}
  \label{fig:ir4-crop}
\end{figure}

\medskip

\textbf{Dynamics around default episodes. }To better understand the role of default and its implications for fiscal policy, we show here the evolution of taxes around a 9-period window around the default event. To do so, we pick 1,000 default episodes in our simulations that were preceded by at least 6 periods of access to financial markets and that were followed by at least 4 periods of financial exclusion. Figure \ref{fig:pp-def} presents the dynamics of the cross-sectional medians of government spending and tax rates in AMSS and in our model around these default events. We also compute a counterfactual tax rate in our economy assuming that the government is not allowed to default at time 0 nor in any of the following 4 periods. The dashed lines correspond to 25 and 75th-percentile bands for each series.

As shown in the top panel, default episodes are preceded by high government expenditure realizations, which peak two periods ahead of the announcement. Even though government expenditure drops in the subsequent two periods, it does not do so enough to prevent default from occurring. Taxes remain on the rise to finance the government spending and the growing debt built over time. In period 0 the government finds optimal to default and tax rates are cut and continue declining thereafter. Had the government not defaulted, taxes would have jumped by almost 50 percent and increase further in the next period, as debt continued rising. The fact that the tax rate is lower upon defaulting than in the counterfactual scenario follows from lemma \ref{lem:roll_over1} in the appendix. This lemma shows that in the states in which default is optimal, if the government repaid, it would not be able to raise any funds from the debt management for any available debt contract. Rather, it would be subject to additional outlays. Consequently, taxes would necessarily be higher in that case. Thus, by defaulting, the government avoids transiently higher tax distortions. In the counterfactual scenario, as shown in the figure, after period 1 tax rates start declining along with the debt level as the government expenditure becomes lower. In sharp contrast with our economy, in AMSS taxes remain quite stable and more persistent, as illustrated in the bottom panel. As expected, taxes  after $t=0$ do not decline as much as in our economy, given the high outstanding debt that the government has to repay.  

\begin{figure}
	\centering
	\includegraphics[height=2.4in,width=5.5in]{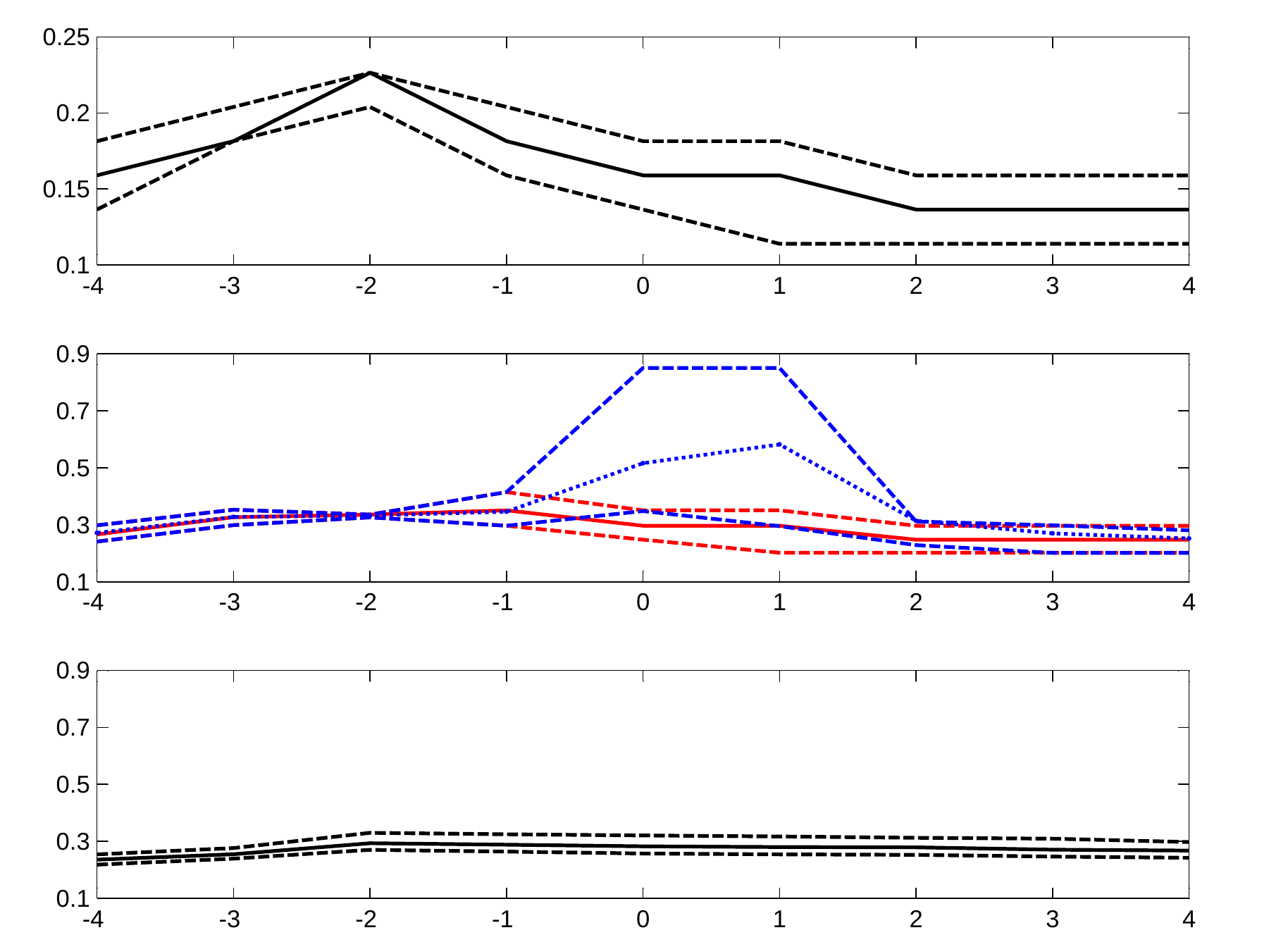}
	\caption{Default episode windows in model simulations. Top panel: Cross-sectional medians of government expenditure. Middle panel: Actual (solid) and counterfactual (dotted) tax rates in our economy. Bottom panel: Tax rates in AMSS. Dashed lines correspond to 25 and 75th-percentile bands.}
	\label{fig:pp-def}

\end{figure}

\medskip

\textbf{Aversion to consumption risk.} In our previous parametrization we assumed that households were consumption-risk neutral by endowing them with quasi-linear preferences. In what follows we relax this assumption by considering a balanced-growth preference specification, as in \cite{AMSS_JPE02} and \cite{FARHI_WP07}. More specifically, we assume $u(c,1-n)=\log c + C_{1} \frac{(1-n)^{1-\sigma}}{1-\sigma}$, with $C_{1}=0.05$ and $\sigma=3$. 
For computational purposes, in this new numerical exercise, the government spending shock can only take three values: $g^{L}$ (``low''), $g^{I}$ (``normal'') and $g^{H}$ (``expansionary''), such that $g^{H}>g^{I}>g^{L}$, with transition probability matrix\footnote{Under this Markov chain for the government expenditure, borrowing is widely used to smooth taxes as the economy fluctuates between the two (likely) states of $g^{L}$ and $g^{I}$. Default typically occurs when the economy switches to the persistent expansionary state $g^{H}$ carrying in high debt and remains in this state for a sufficiently large number of periods.}
\begin{equation*}
\left(
  \begin{array}{ccc}
    0.6 & 0.4 & 0 \\
    0.5 & 0.45 & 0.05 \\
    0.1 & 0.1 & 0.8
  \end{array}
\right)
\end{equation*}

There is only one debt renegotiation offer $\delta=0.7$ with arrival probability $\lambda=0.5$ and an ad hoc output cost for default of $1-\kappa$ equal to 0.4 percent. In this setup, low unconditional probability of occurrence of $g^{H}$ together with a non-negligible output loss in autarky are sufficient to deter the government from defaulting too often. The time discount factor $\beta$ is set to $0.98$.

Again, this numerical exercise is not designed to replicate key features in the data, but to simply provide guidance regarding the robustness of our main insights to nonlinear utility in consumption. A novel feature of defaulting in this environment with commitment to fiscal policies is that it allows for a resetting of taxes. While in our economy after every default episode the fiscal authority has the chance eventually to review and reset its tax policy, in the AMSS model this is not the case, as the government is bound by previous marginal utility promises when $u_{c}\neq 1$. Consequently, fiscal policy will tend to exhibit even more history dependence in the latter environment, as manifested, for example, in high and persistent taxes in states with high indebtedness as the economy hovers around the \emph{ad hoc} borrowing limits.

As illustrated in figure \ref{fig:taxvol_vs_spread_by_310316}, even when the household is risk averse to consumption fluctuations, the positive relationship between tax volatility and average default risk is observed in the simulations.
Additionally, both statistics are higher for high debt.
Not surprisingly, debt-to-output ratios are again lower on average in our economy than in the AMSS model, as shown in figure \ref{fig:debthisto_ra}.
Finally, in figure \ref{fig:pp-def-ra} we present the dynamics of our model around default episodes. In contrast with the quasi-linear case, we do not conduct the counterfactual exercise computing the tax dynamics should the government were not allowed to default. The reason is simple: the marginal utility promised in the period before the default announcement, i.e. $t=-1$, is state-contingent and captures the fact that at $t=0$, upon the government expenditure realization, default occurs and hence the marginal utility corresponds to the financial autarky one. This autarkic marginal utility is not necessarily an equilibrium object when the government repays, i.e. it may not belong to $\Omega(1,B_{0},g_{0},1,1)$. As implied by the dynamics of the median and percentile bands of $g$ in the top panel, all default episodes occur upon two consecutive realizations of the highest government expenditure. In the second and bottom panels we see that taxes are higher on average in our economy than in AMSS, and, more relevant, they grow faster during the periods of financial access preceding the default announcement.

\begin{figure}
	\centering
	\includegraphics[height=2.2in,width=4.6in]{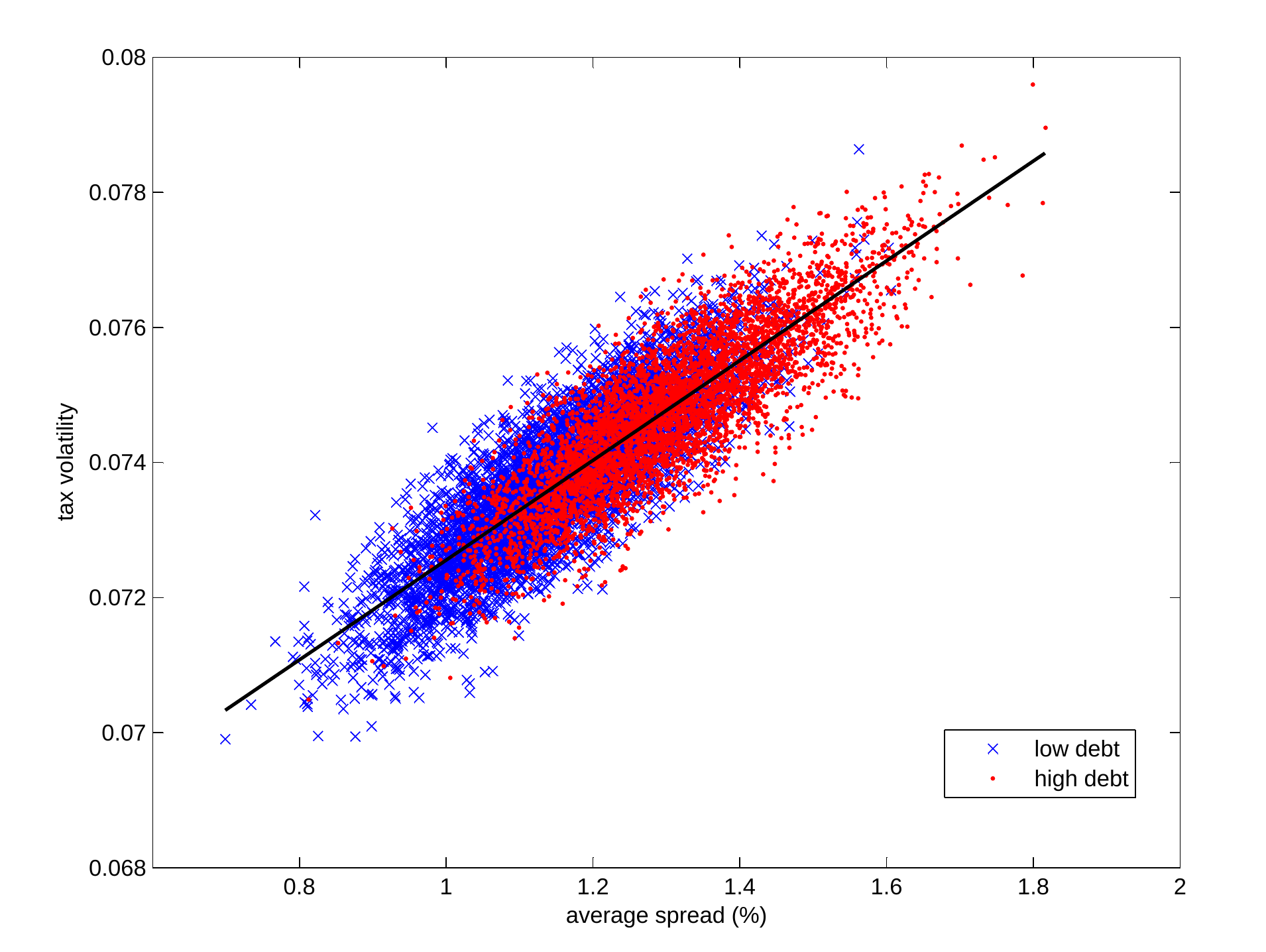}
	\caption{ Standard deviation of tax rates and mean bond spreads in financial access for low debt (blue) and high debt (red), and fitted OLS line between them (black line) in economy with aversion to consumption risk.}
	\label{fig:taxvol_vs_spread_by_310316}
\end{figure}

\begin{figure}
	\centering
	\includegraphics[height=3.0in,width=5.5in]{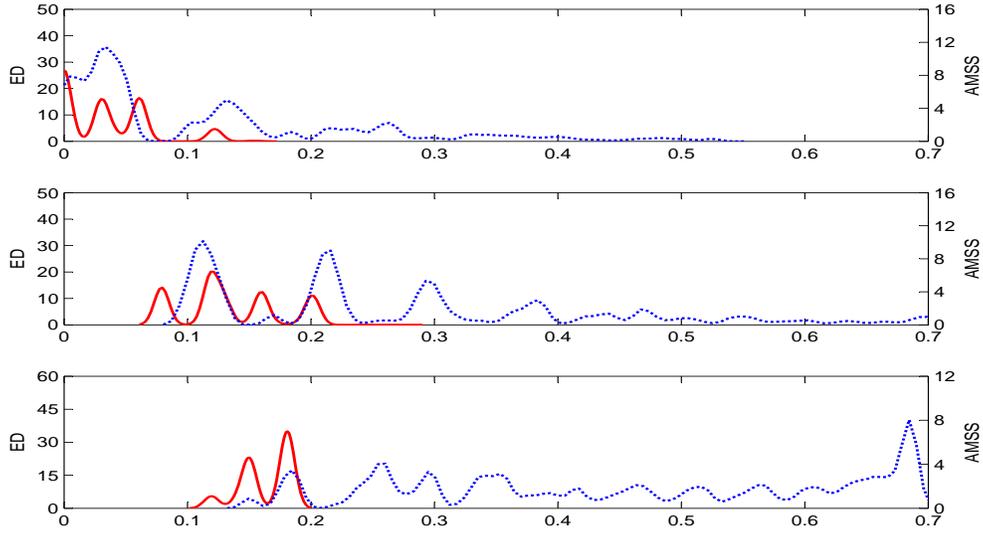}
	\caption{Histograms of debt-to-output ratio for our model (solid red) and AMSS (dotted blue) conditioned on different values of  $g$ for an economy with aversion to consumption risk. From top to bottom: low, intermediate and high realization of $g$.}
	
	\label{fig:debthisto_ra}
\end{figure}

\begin{figure}
	\centering
	\includegraphics[height=2.6in,width=5.5in]{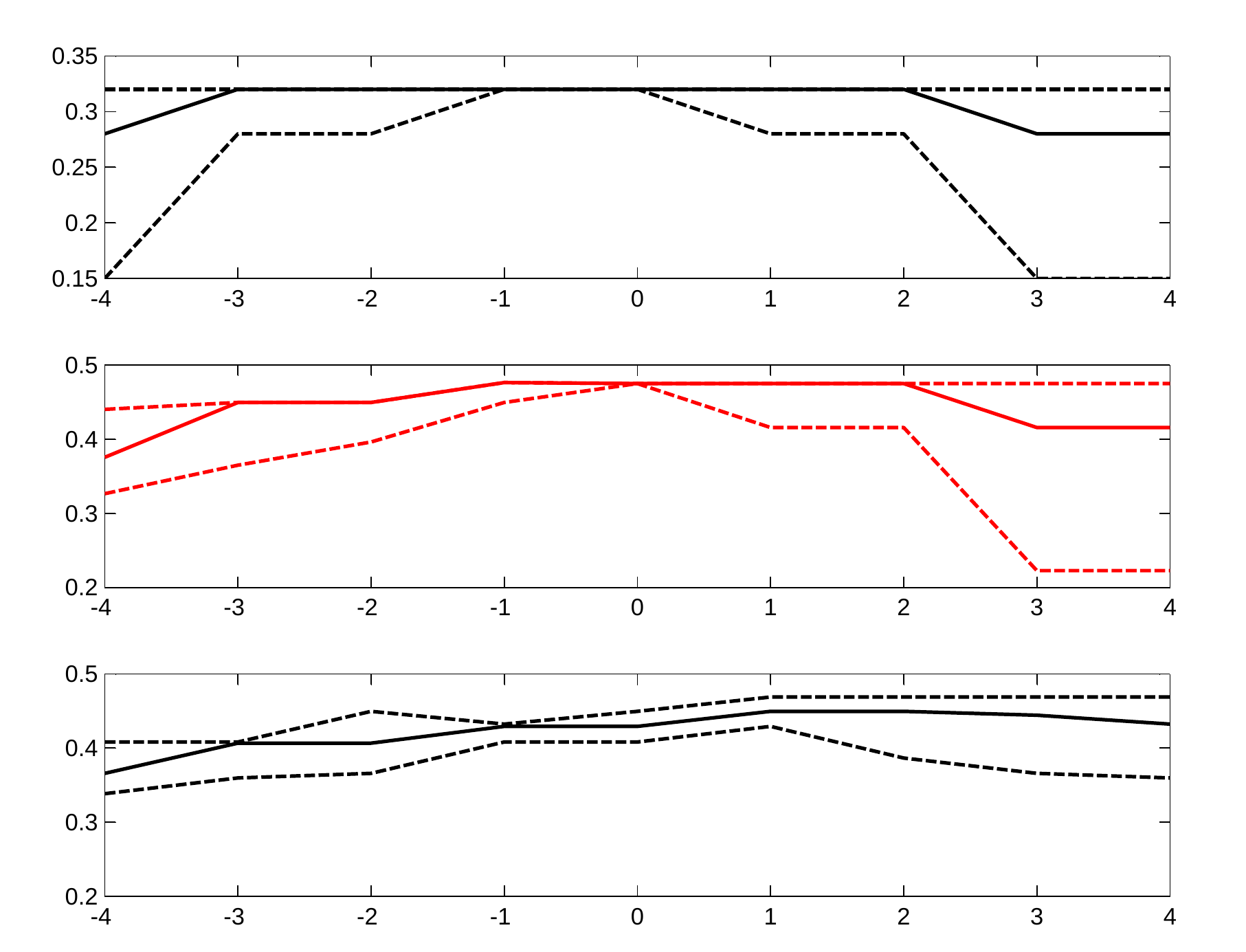}
	\caption{Default episode windows in model simulations for economy with aversion to consumption risk. Top panel: Cross-sectional medians of government expenditure. Middle panel: actual tax rates (solid) in our economy. Bottom panel: tax rates in AMSS. Dashed lines correspond to 25 and 75th-percentile bands.}
	\label{fig:pp-def-ra}

\end{figure}

\section{Conclusion}
\label{sec:conclusion}

Our model extends the results on optimal taxation under incomplete markets results to environments where the government can default on its debt. We study how default risk and the actual default event affect tax policies and vice-versa. By defaulting the government avoids higher tax distortions in the future that would come along with the service of the debt. The presence of default risk, however, gives rise to endogenous credit limits that hinder the government's ability to smooth shocks using debt. As a result, taxes are more volatile and less serially correlated than in the standard incomplete market setting. We view our model as a suitable framework to study government policies for economies that are (or were) prone to default and to restructure their debt. Some examples throughout history include France and the U.S. in the 18th century, and emerging economies nowadays.

Our model also provides a novel device that allows us to study asset prices of domestic debt both during periods of financial access and autarky. Further research could fully explore the pricing implications of this device for general government bonds held by uncertainty or risk averse creditors. In addition, our reduced-form debt restructuring process could be used to study more in detail some observed features in recent debt renegotiation episodes.




{\fontsize{11}{13} \selectfont
\bibliography{mybib_amssa33}

\newpage

\appendix
{\fontsize{10}{12} \selectfont
\renewcommand{\baselinestretch}{1}
\begin{center}
{\huge{Appendix}}
\end{center}

\renewcommand{\theequation}{\Alph{section}.\arabic{equation}}
\renewcommand{\thetable}{\Alph{section}.\arabic{table}}
\renewcommand{\thefigure}{\Alph{section}.\arabic{figure}}
\section{Notation and Stochastic Structure of the Model}

Throughout the appendix for a generic mapping $f$ from a set $S$ to
$T$, we use $s \mapsto f(s)$ or $f : S \rightarrow T$ to denote it. For
the case that a mapping depends on many variables, the notation $s_{1}
\mapsto f(s_{1},s_{2})$ is used to denote the function $f$ \emph{only}
as a function of $s_{1}$, keeping $s_{2}$ fixed. Also, for a generic
set $A$, $|A|$ denotes the cardinality of $A$.

\section{Optimization Problem for the Households}
\label{app:HOUSE}

The Lagrangian associated to the household's problem is given by
\begin{align*}
    \mathcal{L}(\{c_{t},n_{t},b_{t+1},\nu_{t},\mu_{t},\psi_{t}\}_{t=0}^{\infty}) &\equiv \sum_{t=0}^{\infty} \beta^{t} E_{\Pi(\cdot|\omega_{0})} \left[ \left\{      u(c_{t}(\omega^{t}),1-n_{t}(\omega^{t}))\right.\right.\\
 & \left.\left.- \nu_{t}(\omega^{t}) \{ c_{t}(\omega^{t}) - (1-\tau_{t}(\omega^{t})) \kappa_{t}(\omega^{t}) n_{t}(\omega^{t}) + p_{t}(\omega^{t}) b_{t+1}(\omega^{t}) - \varrho_{t}(\omega^{t}) b_{t}(\omega^{t-1}) \} \right . \right. \\
& \left. + \left. \Psi_{t}(\omega^{t})c_{t}(\omega^{t})  + \psi_{1t}(b_{t+1}(\omega^{t}) - \underline{b}) + \psi_{2t} (\overline{b} - b_{t+1}(\omega^{t})) \right\} \right],
\end{align*}
where $\nu_{t}$ and $\Psi_{t}$ are the Lagrange multipliers associated to the budget constraint and to the non-negativity restrictions for consumption, and $\psi_{it}$ $i=1,2$ are the Lagrange multipliers associated to the debt limits.


Assuming interiority of the solutions, the first order conditions (FONC) are given by:
\begin{align*}
c_{t}(\omega^{t}) \colon &  u_{c}(c_{t}(\omega^{t}),1-n_{t}(\omega^{t})) - \nu_{t}(\omega^{t}) = 0 \\
n_{t}(\omega^{t}) \colon &  - u_{l}(c_{t}(\omega^{t}),1-n_{t}(\omega^{t})) + \nu_{t}(\omega^{t})(1-\tau_{t}(\omega^{t})) \kappa_{t}(\omega^{t}) = 0 \\
b_{t+1}(\omega^{t}) \colon &  p_{t}(\omega^{t})  \nu_{t}(\omega^{t}) - E_{\Pi( \cdot | \omega^{t})} [ \beta  \nu_{t+1}(\omega^{t+1}) \varrho_{t+1}(\omega^{t+1})  ] = 0.
\end{align*}

Then, using $u_{j}(\omega^{t})$ for $u_{j}(c_{t}(\omega),1-n_{t}(\omega))$ with $ j \in \{ c,l\}$, it follows
\begin{align}\label{eqn:eq_tax_app}
  \frac{u_{l}(\omega^{t})}{u_{c}(\omega^{t})} &= (1-\tau_{t}(\omega^{t})) \kappa_{t}(\omega^{t}),\\
  p_{t}(\omega^{t}) &= E_{\Pi( \cdot | \omega^{t})} \left[ \beta  \frac{u_{c}(\omega^{t+1})}{u_{c}(\omega^{t})} \varrho_{t+1}(\omega^{t+1})  \right]. \label{eqn:eq_price_app}
\end{align}

From the definition of $\varrho$, equation \ref{eqn:eq_price_app} implies, for $\phi_{t}=1$ ,
\begin{align*}
  p_{t}(\omega^{t}) & = E_{\Pi( \cdot | \omega^{t})} \left[ \beta \frac{u_{c}(\omega^{t+1})}{u_{c}(\omega^{t})} (1-d_{t+1}(\omega))  \right] + E_{\Pi( \cdot | \omega^{t})} \left[ \beta \frac{u_{c}(\omega^{t+1})}{u_{c}(\omega^{t})} d_{t+1}(\omega^{t+1}) q_{t+1}(\omega^{t+1})   \right].
\end{align*}

For $\phi_{t} = 0$, (where in this case recall that $p_{t} = q_{t}$)
\begin{equation*}
  p_{t}(\omega^{t}) = \lambda E_{\Pi( \cdot | \omega^{t})} \left[ \beta \frac{u_{c}(\omega^{t+1})}{u_{c}(\omega^{t})} a_{t+1}(\omega^{t+1}) \delta_{t+1}  \right] + E_{\Pi( \cdot | \omega^{t})} \left[  \beta \frac{u_{c}(\omega^{t+1})}{u_{c}(\omega^{t})} \{1-\lambda + \lambda (1-a_{t+1}(\omega^{t+1}))\} q_{t+1} (\omega^{t+1})  \right].
\end{equation*}

\section{Proofs for Section \ref{sec:charc-CE}}
\label{app:charc-CE}

The next lemma characterizes the set of competitive equilibria as a sequence of restrictions involving FONC and budget constraints. The proof is relegated to the end of the section.

\begin{lemma}\label{lem:CEG_charac}
Suppose assumption \ref{ass:U_prop} holds. The tuple
$(c_{t}, g_{t}, n_{t}, b_{t+1}, p_{t})_{t=0}^{\infty}$ and $\boldsymbol{\sigma}$ is a competitive equilibrium iff given a $B_{0} = b_{0}$, for all $\omega^{t} \in \Omega^{t}$, for all $t$,
\begin{align}\label{eqn:charac-CEG_1b}
  c_{t}(\omega^{t}) =  \kappa_{t}(\omega^{t})n_{t}(\omega^{t}) - g_{t},~and~B_{t+1}(\omega^{t}) =  b_{t+1}(\omega^{t}),\\
  \kappa_{t}(\omega^{t}) \tau_{t}(\omega^{t}) = \left( \kappa_{t}(\omega^{t}) - \frac{u_{l}(\omega^{t})}{u_{c}(\omega^{t})} \right), \label{eqn:charac-CEG_2b}\\
Z_{t}(\omega^{t}) + \phi_{t}(\omega^{t}) \{ p_{t}(\omega^{t}) B_{t+1}(\omega^{t}) - \delta_{t} B_{t}(\omega^{t}) \} \geq 0,\label{eqn:charac-CEG_3b}
\end{align}
where
\begin{align}\label{eqn:charac-CEG_4b}
  p_{t}(\omega^{t}) = E_{\Pi( \cdot | \omega^{t})} \left[ \beta  \frac{u_{c}(\omega^{t+1})}{u_{c}(\omega^{t})} \varrho_{t+1}(\omega^{t+1})  \right],
\end{align}
and if $\phi_{t}(\omega^{t}) = 0,~B_{t+1}(\omega^{t}) = B_{t}(\omega^{t-1})$.
\end{lemma}

\begin{proof}[Proof of Theorem \ref{thm:charac-CEG}]

We now show the ``$\Rightarrow$'' direction. Consider an outcome path $(d_{t},a_{t},B_{t+1},n_{t})_{t=0}^{\infty}$ that is consistent. This means by lemma \ref{lem:CEG_charac} that the tuple $(c_{t}, g_{t}, n_{t}, b_{t+1}, p_{t})_{t=0}^{\infty}$ and $\boldsymbol{\sigma}$ is a competitive equilibrium iff given a $B_{0} = b_{0}$, for all $\omega^{t} \in \Omega^{t}$, for all $t$, equations \ref{eqn:charac-CEG_1b}, \ref{eqn:charac-CEG_2b}, \ref{eqn:charac-CEG_3b} and \ref{eqn:charac-CEG_4b} hold.
Equations  \ref{eqn:charac-CEG_2b} and \ref{eqn:charac-CEG_4b} imply equations \ref{eqn:eq_tax}-\ref{eqn:eq_price}. Equations \ref{eqn:charac-CEG_2b}-\ref{eqn:charac-CEG_4b} imply condition \ref{eqn:IC-time-3}.

We now show the ``$\Leftarrow$'' direction. Suppose now that the outcome path satisfies, for all $\omega^{t} \in \Omega^{t}$, the following equations:  \ref{eqn:feas}, \ref{eqn:eq_tax}, \ref{eqn:eq_price}, and \ref{eqn:IC-time-3}. By using $B_{t+1}(\omega^{t}) =  b_{t+1}(\omega^{t})$, equations \ref{eqn:eq_tax}, \ref{eqn:eq_price} and the feasibility condition, we can enlarge the outcome path by $(c_{t},p_{t},b_{t+1},\tau_{t},g_{t})_{t=0}^{\infty}$. Clearly, restrictions \ref{eqn:charac-CEG_1b}, \ref{eqn:charac-CEG_2b} and \ref{eqn:charac-CEG_4b} hold. By replacing equations \ref{eqn:eq_tax} and \ref{eqn:eq_price} on \ref{eqn:IC-time-3}, it is easy to see that equation \ref{eqn:charac-CEG_2b} holds too.
\end{proof}

\subsection{Proofs of Supplementary Lemmas}

For the proof of Lemma \ref{lem:CEG_charac} we need the following lemma (the proof is relegated to the end of the section).

\begin{lemma}\label{lem:FONC_suff}
 Suppose assumption \ref{ass:U_prop} holds. Then first order conditions \ref{eqn:eq_tax} and \ref{eqn:eq_price} are also sufficient.
\end{lemma}

\begin{proof}[Proof of Lemma \ref{lem:CEG_charac}]
  Take $\boldsymbol{\sigma}$ and $(c_{t},g_{t},n_{t}, b_{t+1})_{t=0}^{\infty}$, and a price schedule $(p_{t})_{t}$ that satisfy the equations. It is easy to see that feasibility and market clearing holds (conditions 3 and 4). Also, by lemma \ref{lem:FONC_suff} optimality of the households is also satisfied.

  To check attainability of the government policy (condition 2), observe that equations \ref{eqn:charac-CEG_1b} -
  \ref{eqn:charac-CEG_3b} imply for all $\omega^{t} \in \Omega^{t}$,
  \begin{align*}
      g_{t} + \phi_{t}(\omega^{t})  \delta_{t} B_{t}(\omega^{t-1}) - \phi_{t}(\omega^{t}) p_{t}(\omega^{t}) B_{t+1}(\omega^{t}) \leq \kappa_{t}(\omega^{t}) \tau_{t}(\omega^{t}) n_{t}(\omega^{t}).
  \end{align*}

  Finally, we check optimality of the households. We first check that the sequences satisfy the budget constraint. Observe that by equations  \ref{eqn:charac-CEG_1b} - \ref{eqn:charac-CEG_3b}
\begin{align*}
    -c_{t}(\omega^{t}) + \kappa_{t}(\omega^{t})n_{t}(\omega^{t}) + \phi_{t}(\omega^{t}) \{ \delta_{t} B_{t}(\omega^{t-1}) - p_{t}(\omega^{t}) B_{t+1}(\omega^{t}) \} \leq \kappa_{t}(\omega^{t}) \tau_{t}(\omega^{t}) n_{t}(\omega^{t}).
\end{align*}

If $\phi_{t}(\omega^{t}) = 1$, then equation \ref{eqn:charac-CEG_3b} implies that $b_{t+1}(\omega^{t}) = B_{t+1}(\omega^{t})$ for all $t$ (and for $b_{0}$ we assume it is equal to $B_{0}$) and thus
\begin{align*}
   & - c_{t}(\omega^{t}) + \kappa_{t}(\omega^{t})n_{t}(\omega^{t}) + \delta_{t} b_{t}(\omega^{t-1}) - p_{t}(\omega^{t}) b_{t+1}(\omega^{t}) \leq \kappa_{t}(\omega^{t}) \tau_{t}(\omega^{t}) n_{t}(\omega^{t}).
\end{align*}
This coincides with the budget constraint of the household.

If $d_{t}(\omega^{t}) = 1$, but $a_{t}(\omega^{t}) = 0$, equations \ref{eqn:charac-CEG_1b} and \ref{eqn:charac-CEG_3b} imply that $b_{t}(\omega^{t-1}) = b_{t+1}(\omega^{t}) = 0$ for all $t$, so $- c_{t}(\omega^{t}) + \kappa_{t}(\omega^{t})n_{t}(\omega^{t})  = \kappa_{t}(\omega^{t}) \tau_{t}(\omega^{t}) n_{t}(\omega^{t})$, which is the budget constraint of the household.\\
Take $\boldsymbol{\sigma}$ and $(c_{t},g_{t},n_{t}, b_{t+1},p_{t})_{t=0}^{\infty}$ being a competitive equilibrium. Then it is easy to see that it satisfies the equations.
\end{proof}

\begin{proof}[Proof of Lemma \ref{lem:FONC_suff}]
  Under assumption \ref{ass:U_prop} the objective function of the household optimization problem is strictly concave. The budget constraints and debt constraint form a convex set of constraints. Thus, if the transversality condition holds, the FONC are sufficient; this follows from a simple adaptation of the results in \cite{SLP_book89} Ch. 4.5.

  In order to verify the transversality condition, it suffices to show that for any $\zeta_{t}(\omega^{t})$ such that $b_{t}(\omega^{t}) + \zeta_{t}(\omega^{t}) \in \mathbb{B}$,
\begin{align*}
  \lim_{T \rightarrow \infty} \beta^{T} E_{\Pi} [ u_{c}(\kappa_{T}(\omega^{T})n_{T}(\omega^{T}) - g_{T}, 1- n_{T}(\omega^{T})) \varrho_{T}(\omega^{T})\zeta_{T}(\omega^{T}) ] = 0,
\end{align*}

which follows from \cite{Magill-Quinzii-ECMA94} Theorem 5.2 as debt is constrained by assumption.
\end{proof}

\section{Proofs for Section \ref{sec:govt}}
\label{app:govt}

In this section we provide formal definitions of the sets $\mathcal{S}(h_{0},\phi)$ and $\Omega(h_{0},\phi)$ introduced in the recursive representation of the government problem. We also provide the proof of Theorem \ref{thm:VF-rec}.

\bigskip

\textbf{Formal definition of $\mathcal{S}(h_{0},\phi_{0})$.} For any $h_{0} \in \mathbb{H}$ and $\phi_{0} \in \{0,1\}$, let
{\small{
		\begin{align*}
		\mathcal{S}(h_{0},\phi_{0}) \equiv &\left\{ \boldsymbol{\gamma} \hspace{-0.01in}:\hspace{-0.02in}~\forall~(h^{t},\phi_{t}) \in \mathbb{H}^{t} \times \{0,1\},~\boldsymbol{\gamma}|_{(h^{t},\phi_{t})}~renders~ (d_{\tau}(\boldsymbol{\gamma}),a_{\tau}(\boldsymbol{\gamma}),B_{\tau+1}(\boldsymbol{\gamma}),n_{\tau}(\boldsymbol{\gamma}))_{\tau=t}^{\infty} \in CE_{\phi_{t}}(\omega_{t},B), \right.\\
		&\left. with~B=(\delta_{t}\phi_{t}+(1-\phi_{t}))
		B_{t}(\boldsymbol{\gamma}) (h^{t-1},\phi_{t-1}(\boldsymbol{\gamma})(h^{t-1}))\right\}.
		\end{align*}}}
Analogously to \cite{CHANG_JET98}, by drawing the strategies $\gamma$ from $\mathcal{S}(h_{0},\phi_{0})$ we ensure that after any history following $(h_{0},\phi_{0})$ the continuation strategy delivers competitive equilibrium allocations. As it will become clearer later on, when making the default/repayment decision embedded in $\phi_{0}$, the default authority evaluates welfare after his alternative courses of action. To then compute the utility for any $\phi_{0}$, the candidate strategies $\gamma$ to be considered have to belong to the corresponding $\mathcal{S}(h_{0},\phi_{0})$.\footnote{Technically, it could be the case that $\phi_{t}=1$ but the debt that period is too high to be repaid in the competitive equilibrium. For this case, we simply set the per-period payoff at an arbitrary large negative value and thus ensure that this choice of $\phi_{t}$ will never arise as part of the optimal solution of the government problem.}

\bigskip

\textbf{Formal definition of $\Omega(h_{0},\phi)$.} Recall that for autarky ($\phi=0$), the ``promised'' marginal
utilities of consumption are trivially pinned down by the choice of labor that balances the government budget and maximizes the
per-period payoff; i.e., for any $g \in \mathbb{G}$, the ``promised'' marginal
utility of consumption equals $m_{A}(g) \equiv
u_{c}(\kappa \mathbf{n}^{\ast}_{0}(g)-g,1-\mathbf{n}^{\ast}_{0}(g))$ where\footnote{The lemma \ref{lem:charac-z}(1) ensures that $\mathbf{n}^{\ast}_{0}(g)$ exists and is
	unique for all $g$.}
\begin{align*}
\mathbf{n}^{\ast}_{0}(g) = \arg\max_{n \in [0,1]} \{ u(\kappa n-g,1-n) :  z(\kappa,n,g) = 0\}.
\end{align*}


We now proceed to formally define our object of interest. For any $h_{0} = (\phi_{-1},B_{0},g_{0},\delta_{0}) \in \mathbb{H}$ and $\phi \in \{0,1\}$, let
\begin{align*}
\Omega(h_{0},\phi)= & \left\{ (\mu,v)\in \mathbb{R}_{+}\times\mathbb{R}:\exists~\boldsymbol{\gamma}~ \in \mathcal{S}(h_{0},\phi),~and~(V_{\tau}(h^{\tau},0),V_{\tau}(h^{\tau},1))_{h^{\tau},\tau}~such~that: \right. \\
& \left. \mu = m_{A}(g)~if~\phi=0,~and~\mu=u_{c}(n_{0}(\boldsymbol{\gamma})(h_{0})-g_{0},1-n_{0}(\boldsymbol{\gamma})(h_{0}))~if~\phi=1, \right.\\
& \left. v = V_{0}(h_{0},\phi) \right.\\
&\left. (V_{\tau}(h^{\tau},\phi))_{h^{\tau},\tau}~satisfies~expression~\ref{eq:8}~for~any~\phi \in \{0,1\}, \right.\\
&\left. \boldsymbol{\gamma}^{D}|_{h_{0},\phi_{1}(\boldsymbol{\gamma})}~are~determined~by~expressions~\ref{eqn:PF-d1}-\ref{eqn:PF-a1} \right\},
\end{align*}

For each initial history $h_{0}$ and $\phi$, the set
$\Omega(h_{0},\phi)$ is given by all the values for marginal
utility and lifetime utility values at time zero that can be sustained
in a competitive equilibrium, wherein the default authority reacts
optimally from next period on. Each pair $(\mu,v)$ imposes
restrictions on the labor allocation at time 0  as well as on the lifetime utility at time 0, given $h_{0}$ and
$\phi$. Finally, note that the set $\Omega(h_{0},\phi)$ when $\phi$ is the value chosen by the government, contains the promised marginal utilities (and utility values) that can be delivered along the equilibrium path, while $\Omega(h_{0},\phi)$ when $\phi$ is the value not chosen by the government, contains off-equilibrium marginal utilities.


The correspondence $\Omega$ is an equilibrium object, endogenously determined, that can be computed using numerical methods as the largest fixed point of an appropriately constructed correspondence operator, in the spirit of \cite{APS_EMETRICA90}. Henceforth, we proceed to formulate and solve the recursive problem of the fiscal authority as if we already know $\Omega$.

\bigskip

\textbf{Proof of theorem \ref{thm:VF-rec}} The next lemma characterizes the government surplus function, the proof is relegated to the end of this section.

\begin{lemma}\label{lem:charac-z}
  Let $(\kappa,n,g) \mapsto z(\kappa,n,g) = \left(\kappa - \frac{u_{l}(n-g,1-n)}{u_{c}(n-g,1-n)} \right)n - g$. Then:

  \begin{enumerate}
  \item $\arg\max_{n\in [0,1]} \{ u(\kappa n - g, 1-n)~:~ z(\kappa,n,g) = 0\}$ exists and is unique.
  \item Suppose assumption \ref{ass:linear_c} holds  and let $\bar{n}(g) = \arg\max_{n \in [0,1]} z(1,n,g)$. Then, $n \mapsto z(1,n,g)$ is decreasing and strictly concave for all $n \in [\bar{n}(g),1]$
  \end{enumerate}
\end{lemma}

To show theorem \ref{thm:VF-rec} we need the following lemma whose proof is relegated to the end of this section.

\begin{lemma}\label{lem:rec-S}
  If, for any $h_{0} = (1,B,g,\delta) \in \mathbb{H}$ and $\phi_{0} \in \{0,1\}$, $\boldsymbol{\gamma} \in \mathcal{S}(h_{0},\phi_{0})$,
  then $\boldsymbol{\gamma}|_{h^{t},\phi} \in \mathcal{S}(h_{t},\phi)$
  for any $h^{t} \in \mathbb{H}^{t}$ and $\phi \in \{0,1\}$. Moreover,
{\small{\begin{align*}
    z(\kappa_{\phi_{0}},n_{0}(\boldsymbol{\gamma})(h_{0}),g)\mu_{0}(\boldsymbol{\gamma})(h_{0}) + \phi_{0} \{ \mathcal{P}_{\phi_{0}}(g,B_{1}(\boldsymbol{\gamma})(h_{0},\phi_{0}),\mu_{1}(\boldsymbol{\gamma})(h_{0},h_{1}(\cdot))) B_{1}(\boldsymbol{\gamma})(h_{0},\phi_{0}) - \delta \mu_{0}(\boldsymbol{\gamma})(h_{0}) B \} \geq 0
  \end{align*}}}
where $\kappa_{\phi} = \kappa (1-\phi) + \phi$ and $h_{1}(\cdot) \equiv (1,B_{1}(\boldsymbol{\gamma})(h_{0},1),\cdot,1)$ and for $t = 0,1$
\begin{align*}
  \mu_{t+1}(\boldsymbol{\gamma})(h^{t},h_{t+1}(g')) = u_{c}(n_{t+1}(\boldsymbol{\gamma})(h^{t},h_{t+1}(g'))-g',1-n_{t+1}(\boldsymbol{\gamma})(h^{t},h_{t+1}(g')))
\end{align*}
\end{lemma}

The proof of Theorem \ref{thm:VF-rec} is analogous to the standard proof of the principle of optimality for the single agent case, e.g. Theorem 9.2 in \cite{SLP_book89}.

\begin{proof}[Proof of Theorem \ref{thm:VF-rec}]  By definition of $V^{\ast}_{1}$, $V^{\ast}_{0}$ and
  $\overline{V}^{\ast}_{1}$, it follows that with $h_{0} = (\phi_{-1},B,g,\bar{\delta})$ and $\phi=0$
\begin{align} \label{eqn:VF0_c2}
 &V_{0}^{\ast}(g,B) = \sup_{\boldsymbol{\gamma}} V_{0}(\boldsymbol{\gamma})(h_{0},0) \\
 & ~ subject~to~ \boldsymbol{\gamma}=(\boldsymbol{\gamma}^{F},\boldsymbol{\gamma}^{D}) \in \mathcal{S}(h_{0},0) \label{eqn:VF0_strat}\\
 & ~~~~\boldsymbol{\gamma}^{D}|_{h_{0},\phi=0}~are~determined~by~(\ref{eqn:PF-a1})-(\ref{eqn:PF-d1})\\
 & ~~~~u_{c}(\kappa n_{0}(\boldsymbol{\gamma})(h_{0}) - g, 1-n_{0}(\boldsymbol{\gamma})(h_{0})) = m_{A}(g) \label{eqn:VF0_c1}
\end{align}
and similarly, with $h_{0} = (1,B,g,1)$ and $\mu \in \mathbb{R}_{+}$ and $\phi_{0} =1$
\begin{align} \label{eqn:VF1_c1}
  &V_{1}^{\ast}(g,B,\mu) = \sup_{\boldsymbol{\gamma}} V_{0}(\boldsymbol{\gamma})(h_{0},1) \\
 & ~ subject~to~ \boldsymbol{\gamma}=(\boldsymbol{\gamma}^{F},\boldsymbol{\gamma}^{D}) \in \mathcal{S}(h_{0},1) \\
 & ~~~~\boldsymbol{\gamma}^{D}|_{h_{0},\phi=1}~are~determined~by~(\ref{eqn:PF-a1})-(\ref{eqn:PF-d1}) \\
 & ~~~~u_{c}(n_{0}(\boldsymbol{\gamma})(h_{0}) - g, 1-n_{0}(\boldsymbol{\gamma})(h_{0})) = \mu. \label{eqn:VF1_c2}
\end{align}
Finally, with $h_{0} = (0,B,g,\delta)$ and $\phi_{0}=1$
\begin{align} \label{eqn:VFbar_c1}
  &\overline{V}_{1}^{\ast}(g,\delta B) = \sup_{\boldsymbol{\gamma}} V_{0}(\boldsymbol{\gamma})(h_{0},1) \\
 & ~ subject~to~ \boldsymbol{\gamma}=(\boldsymbol{\gamma}^{F},\boldsymbol{\gamma}^{D}) \in \mathcal{S}(h_{0},1) \\
 & ~~~~\boldsymbol{\gamma}^{D}|_{h_{0},\phi=1}~are~determined~by~(\ref{eqn:PF-a1})-(\ref{eqn:PF-d1}).
\end{align}

The first (sequential) problem consists of selecting $\boldsymbol{\gamma}$, consistent with competitive equilibrium and optimality for the default authority from $t=1$ on, to maximize the lifetime utility of households, conditional on $h_{0}=(\phi_{-1},B,g,\delta)$ and $\phi=0$. The solution is given by $V_{0}^{\ast}(g,B)$, which does not depend on $\delta$ nor $\mu$. Condition \ref{eqn:VF0_c1} ensures that the current marginal utility is equal to the autarkic value defined before.

Problem \ref{eqn:VF1_c1} is analogous to Problem \ref{eqn:VF0_c2} with
$\phi=1$ and $\phi_{-1}=1 $ instead. In this case, we impose through condition
\ref{eqn:VF1_c2} that the current marginal utility is $\mu$.

Henceforth, we refer to strategies that satisfy the restrictions on
the above programs as \emph{admissible}. We also assume that the
suprema are achieved; this assumption is to ease the exposition, if this were not the case the proof still goes
through by exploiting the definition of the supremum.

By definition, $\overline{V}_{1}^{\ast}(g,\delta B) \geq V_{0}(\boldsymbol{\gamma})(0,B,g,\delta,1)$ for all $\boldsymbol{\gamma} \in \mathcal{S}(0,B,g,\delta,1)$ and $\boldsymbol{\gamma}^{D}|_{h_{0},\phi=1}$ are determined by (\ref{eqn:PF-a1})-(\ref{eqn:PF-d1}). By definition of $\Omega(0,B,g,\delta,1)$, this implies that for all $(\mu,v) \in \Omega(0,B,g,\delta,1)$, $\overline{V}_{1}^{\ast}(g,\delta B) \geq v$. On the other hand, assuming that the there exists a strategy $\boldsymbol{\gamma}$ that achieves the supremum, it has to be true that there exists a $\mu$ such that $(\mu,\overline{V}_{1}^{\ast}(g,\delta B)) \in \Omega(0,B,g,\delta,1)$. Therefore,
\begin{align}\label{eqn:VF-rec-1}
  \overline{V}_{1}^{\ast}(g,\delta B) = \max\{ v | (\mu,v) \in \Omega(0,B,g,\delta,1)\}.
\end{align}

It is easy to see that the same result applies for any time $t$ and any history $(h^{t},0,B,g,\delta)$ (not just $t=0$ and $h_{0} = (0,B,g,\delta)$).

Let $h_{0} \equiv (\phi_{-1},B,g,\delta)$ and $\phi = 0$. Suppose that there
exists a strategy $\hat{\boldsymbol{\gamma}}$ that achieves the
supremum in program \ref{eqn:VF0_c2}. Then,\footnote{Henceforth we abuse notation and use
  $n_{0}(\boldsymbol{\gamma})(g)$ instead of $n_{0}(\boldsymbol{\gamma})(h_{0})$.}
  \begin{align} \notag
   V^{\ast}_{0}(g,B) = & u( \kappa n_{0}(\hat{\boldsymbol{\gamma}})(g) - g,1- n_{0}(\hat{\boldsymbol{\gamma}})(g)) \\ \notag
& + \beta\lambda \int_{\mathbb{G}} \int_{\Delta} \max\{ V_{1}(\hat{\boldsymbol{\gamma}})(h_{0},0,B,(g',\delta'),1), V_{1}(\hat{\boldsymbol{\gamma}})(h_{0},0,B,(g',\delta'),0) \} \pi_{\Delta}(d\delta')\pi_{\mathbb{G}}(dg'|g) \\
  & +  \beta (1-\lambda)\int_{\mathbb{G}} V_{1}(\hat{\boldsymbol{\gamma}})(h_{0},0,B,(g',\overline{\delta}),0)\pi_{\mathbb{G}}(dg'|g).   \notag
 \end{align}

Observe that, for any $g' \in \mathbb{G}$,
$V_{1}(\hat{\boldsymbol{\gamma}})(h_{0},0,B,(g',\delta'),0)$ is
constant with respect to $\delta'$. Also, note that
$\hat{\boldsymbol{\gamma}}|_{h^{1},\phi}$ is admissible by
lemma \ref{lem:rec-S}. It also follows that
$V_{1}(\boldsymbol{\gamma})(h_{0},0,B,(g',\delta'),0) =
V_{0}(\boldsymbol{\gamma})(0,B,(g',\delta'),0)$ for any strategy
$\boldsymbol{\gamma}$ and any $(h_{0},g',\delta')$. Thus
\begin{align}
  V_{1}(\hat{\boldsymbol{\gamma}})(h_{0},0,B,(g',\bar{\delta}),0) =
  V_{0}^{\ast}(g',B),~\forall g' \in \mathbb{G}.
\end{align}
Therefore,
\begin{align}\label{eqn:VF-rec-2}
   V^{\ast}_{0}(g,B) = & u( \kappa n_{0}(\hat{\boldsymbol{\gamma}})(g) - g,1- n_{0}(\hat{\boldsymbol{\gamma}})(g)) \\
& + \beta\lambda \int_{\mathbb{G}} \int_{\Delta} \max\{ V_{1}(\hat{\boldsymbol{\gamma}})(h_{0},0,B,(g',\delta'),1), V_{0}^{\ast}(g',B) \} \pi_{\Delta}(d\delta')\pi_{\mathbb{G}}(dg'|g) \notag \\
  & +  \beta (1-\lambda)\int_{\mathbb{G}} V_{0}^{\ast}(g',B) \pi_{\mathbb{G}}(dg'|g).\notag
\end{align}

By construction, $n_{0}(\hat{\boldsymbol{\gamma}})(g) = \mathbf{n}^{\ast}_{0}(g)$ and thus,
\begin{align}
  \notag
   V^{\ast}_{0}(g,B) = & u( \kappa \mathbf{n}^{\ast}_{0}(g) - g,1-\mathbf{n}^{\ast}_{0}(g)) \\ \notag
& + \beta\lambda \int_{\mathbb{G}} \int_{\Delta} \max\{ V_{1}(\hat{\boldsymbol{\gamma}})(h_{0},0,B,(g',\delta'),1), V^{\ast}_{0}(g',B)  \} \pi_{\Delta}(d\delta')\pi_{\mathbb{G}}(dg'|g) \\
  & +  \beta (1-\lambda)\int_{\mathbb{G}} V^{\ast}_{0}(g',B)\pi_{\mathbb{G}}(dg'|g).   \notag
\end{align}

Observe that at $(h_{0},\phi_{0}=0,B,g',\delta',\phi_{1}=1)$ a ``new'' fiscal authority begins at time $t=1$. By construction, this fiscal authority starts without binding promises regarding the marginal utility of consumption. Since $\hat{\boldsymbol{\gamma}}$ is optimal, it follows that $V_{1}(\hat{\boldsymbol{\gamma}})(h_{0},0,B,g',\delta',1) = \overline{V}_{1}^{\ast}(g',\delta' B)$. Therefore,
\begin{align*}
   V^{\ast}_{0}(g,B) = & u( \kappa \mathbf{n}^{\ast}_{0}(g) - g,1-\mathbf{n}^{\ast}_{0}(g)) + \beta\lambda \int_{\mathbb{G}} \int_{\Delta} \max\{ \overline{V}_{1}^{\ast}(g',\delta' B), V^{\ast}_{0}(g',B)  \} \pi_{\Delta}(d\delta')\pi_{\mathbb{G}}(dg'|g) \\
  & +  \beta (1-\lambda)\int_{\mathbb{G}}
  V^{\ast}_{0}(g',B)\pi_{\mathbb{G}}(dg'|g).   
\end{align*}

We now consider program \ref{eqn:VF1_c1}. With an slight abuse of notation, let $\hat{\boldsymbol{\gamma}}$ be the strategy that achieves the supremum in program \ref{eqn:VF1_c1}. Henceforth, let $\mu_{t}(\boldsymbol{\gamma})(h^{t}) \equiv
u_{c}(n_{t}(\boldsymbol{\gamma})(h^{t}) -
g_{t},1-n_{t}(\boldsymbol{\gamma})(h^{t}))$ for any strategy
$\boldsymbol{\gamma}$ and history $h^{t} \in \mathbb{H}^{t}$. Observe that, for $h_{1} =
(1,B_{1}(\hat{\boldsymbol{\gamma}})(h_{0},1),(g',1))$, due to lemma
\ref{lem:rec-S}, $\hat{\boldsymbol{\gamma}}|_{h^{1},\phi}$ is admissible (taking $\mu$ as $\mu_{1}(\hat{\boldsymbol{\gamma}})(h^{1})$), because $\hat{\boldsymbol{\gamma}}|_{h^{1},\phi} \in \mathcal{S}(h_{1},\phi)$,
and also $\boldsymbol{\gamma}^{D}|_{h^{1},\phi=1}$ are determined by
expressions \ref{eqn:PF-a1}-\ref{eqn:PF-d1}. Thus
\begin{align*}
   V_{1}(\hat{\boldsymbol{\gamma}})(h_{0},h_{1},1) \leq V^\ast_{1}(g',B_{1}(\hat{\boldsymbol{\gamma}})(h_{0},1),\mu_{1}(\hat{\boldsymbol{\gamma}})(h^{1}))\quad\text{  and  }\quad V_{1}(\hat{\boldsymbol{\gamma}})(h_{0},h_{1},0) \leq V^\ast_{0}(g',B_{1}(\hat{\boldsymbol{\gamma}})(h_{0},1)).
\end{align*}

Therefore, letting $h^{1}(g') = (1,B_{1}(\hat{\boldsymbol{\gamma}})(h_{0},1),(g',1)))$,
\begin{align*}
   V^{\ast}_{1}(g,\delta B, \mu) \leq & u( n_{0}(\hat{\boldsymbol{\gamma}})(g) - g,1- n_{0}(\hat{\boldsymbol{\gamma}})(g)) \\ \notag
& + \beta  \int_{\mathbb{G}} \max\{ V^{\ast}_{1}(g',B_{1}(\hat{\boldsymbol{\gamma}})(h_{0},1),\mu_{1}(\hat{\boldsymbol{\gamma}})(h^{1}(g'))), V^{\ast}_{0}(g',B_{1}(\hat{\boldsymbol{\gamma}})(h_{0},1)) \} \pi_{\mathbb{G}}(dg'|g).
\end{align*}

By lemma \ref{lem:rec-S}, $(n_{0}(\hat{\boldsymbol{\gamma}})(g),B_{1}(\hat{\boldsymbol{\gamma}})(h_{0},1),\mu_{1}(\hat{\boldsymbol{\gamma}})(h^{1}(\cdot))))$ are such that $u_{c}( n_{0}(\hat{\boldsymbol{\gamma}})(g) - g,1- n_{0}(\hat{\boldsymbol{\gamma}})(g)) = \mu$ and
\begin{align*}
  z(1,n_{0}(\hat{\boldsymbol{\gamma}})(g),g)\mu + \mathcal{P}^{\ast}_{1}(g,B_{1}(\hat{\boldsymbol{\gamma}})(h_{0},1),\mu_{1}(\hat{\boldsymbol{\gamma}})(h^{1}(\cdot))B_{1}(\hat{\boldsymbol{\gamma}})(h_{0},1)    \geq B \mu.
\end{align*}
Therefore,
\begin{align*}
  V^{\ast}_{1}(g,  B, \mu) \leq  \max_{(n',B',\mu'(\cdot)) \in \Gamma(g, B,\mu)} u( n - g,1- n)  + \beta  \int_{\mathbb{G}} \max\{ V^{\ast}_{1}(g',B',\mu'(g')), V^{\ast}_{0}(g',B') \} \pi_{\mathbb{G}}(dg'|g).
\end{align*}

We now show that the reversed inequality holds:
{\small{ \begin{align} \notag
		V^{\ast}_{1}(g, B, \mu) = & u( n_{0}(\hat{\boldsymbol{\gamma}})(g) - g,1- n_{0}(\hat{\boldsymbol{\gamma}})(g)) \\ \notag
		+ & \beta  \int_{\mathbb{G}} \max\{ V_{1}(\hat{\boldsymbol{\gamma}})(h_{0},1,B_{1}(\hat{\boldsymbol{\gamma}})(h_{0},1),(g',1),1), V_{1}(\hat{\boldsymbol{\gamma}})(h_{0},1,B_{1}(\hat{\boldsymbol{\gamma}})(h_{0},1),(g',1),0) \} \pi_{\mathbb{G}}(dg'|g) \\  \notag
		\geq & u( n_{0}(\boldsymbol{\gamma})(g) - g,1- n_{0}(\boldsymbol{\gamma})(g)) \label{eqn:VF-rec-3} \\ 
		+ & \beta  \int_{\mathbb{G}} \max\{ V_{1}(\boldsymbol{\gamma})(h_{0},1,B_{1}(\boldsymbol{\gamma})(h_{0},1),(g',1),1), V_{1}(\boldsymbol{\gamma})(h_{0},1,B_{1}(\boldsymbol{\gamma})(h_{0},1),(g',1),0) \} \pi_{\mathbb{G}}(dg'|g)
		\end{align}}}
where $h_{0} = (1,B,g,1)$ and the second line holds for any
$\boldsymbol{\gamma}$ admissible.

 For this we construct
the following strategy $\tilde{\boldsymbol{\gamma}}$: (1) $\tilde{\boldsymbol{\gamma}}^{D}$ are determined by
expressions \ref{eqn:PF-a1}-\ref{eqn:PF-d1}; (2) for any $\phi$ and $h_{1}$, $\tilde{\boldsymbol{\gamma}}^{F}(h_{0},\phi) = B_{1}(\tilde{\boldsymbol{\gamma}})(h_{0},\phi)$ and $\mu_{1}(\tilde{\boldsymbol{\gamma}})(h^{1})$ are such that
\begin{align}\label{eqn:rec-bc-1}
z(1,n_{0}(\tilde{\boldsymbol{\gamma}})(g),g)\mu + \phi\{ \mathcal{P}^{\ast}_{1}(g,B_{1}(\tilde{\boldsymbol{\gamma}})(h_{0},1),\mu_{1}(\tilde{\boldsymbol{\gamma}})(h_{0},\circ))B_{1}(\tilde{\boldsymbol{\gamma}})(h_{0},1) - B \mu\} \geq 0,
\end{align}
where $\circ$ stands for
$(1,B_{1}(\tilde{\boldsymbol{\gamma}})(h_{0},\phi),(\cdot,\phi))$, and
$B_{1}(\tilde{\boldsymbol{\gamma}})(h_{0},\phi)=B$ and
$\mu_{1}(\tilde{\boldsymbol{\gamma}})(h^{1}) = m_{A}(g')$ if $\phi=0$;
(3) the remaining components of the strategy
$\tilde{\boldsymbol{\gamma}}^{F}$ agree with
$\hat{\boldsymbol{\gamma}}^{F}$, i.e.,
$\tilde{\boldsymbol{\gamma}}^{F}|_{h^{1},\phi} =
\hat{\boldsymbol{\gamma}}^{F}|_{h^{1},\phi}$ for all history $h^{1}
\in \mathbb{H}^{1}$ and $\phi \in \{0,1\}$.


We now verify that $\tilde{\boldsymbol{\gamma}}$ is admissible, which boils down to
proving that $\tilde{\boldsymbol{\gamma}} \in
\mathcal{S}(h_{0},1)$. Observe that by our construction $(n_{0}(\tilde{\boldsymbol{\gamma}})(g),B_{1}(\tilde{\boldsymbol{\gamma}})(h_{0},1))$ satisfy the implementability constraint (equation \ref{eqn:rec-bc-1}) at time $t=0$ for a price given by $ \mathcal{P}^{\ast}_{1}(g,B_{1}(\tilde{\boldsymbol{\gamma}})(h_{0},1),\mu_{1}(\tilde{\boldsymbol{\gamma}})(h_{0},\circ))$ and it satisfies that $B_{1}(\tilde{\boldsymbol{\gamma}})(h_{0},0)=B$. Additionally, from lemma \ref{lem:rec-S},
$\hat{\boldsymbol{\gamma}}|_{h^{1},\phi} \in \mathcal{S}(h_{1},\phi)$, so these two results imply that $\tilde{\boldsymbol{\gamma}} \in
\mathcal{S}(h_{0},1)$.

Also for any $h^{1} \in \mathbb{H}^{2}$ and $\phi \in \{0,1\}$ with $h_{1} = (1,B_{1}(\tilde{\gamma})(h_{0},1),g',1)$, since $\hat{\boldsymbol{\gamma}}|_{h^{1},\phi} \in \mathcal{S}(h_{1},\phi)$, it follows that $V_{1}(\hat{\boldsymbol{\gamma}})(h^{1},1) = V^{\ast}_{1}(g', B_{1}(\tilde{\boldsymbol{\gamma}})(h_{0},1),\mu_{1}(\tilde{\boldsymbol{\gamma}})(h^{1}))$ and $V_{1}(\hat{\boldsymbol{\gamma}})(h^{1},0) = V^{\ast}_{0}(g',B_{1}(\tilde{\boldsymbol{\gamma}})(h_{0},1))$, otherwise there would be an admissible strategy that achieves a higher value for $V_{0}(\cdot)(h_{0},\phi)$ than $\hat{\boldsymbol{\gamma}}$.

Hence, evaluating display \ref{eqn:VF-rec-3} at $\tilde{\boldsymbol{\gamma}}$, it follows that
\begin{align*}
  V^{\ast}_{1}(g, B, \mu) \geq & u( n_{0}(\tilde{\boldsymbol{\gamma}})(g) - g,1- n_{0}(\tilde{\boldsymbol{\gamma}})(g)) \\ \notag
& + \beta  \int_{\mathbb{G}} \max\{ V^{\ast}_{1}(g',B_{1}(\tilde{\boldsymbol{\gamma}})(h_{0},1),\mu_{1}(\tilde{\boldsymbol{\gamma}})(h_{0},h_{1}(g'))), V^{\ast}_{0}(g',B_{1}(\tilde{\boldsymbol{\gamma}})(h_{0},1)) \} \pi_{\mathbb{G}}(dg'|g)
\end{align*}
where $h_{1}(g')$ stands for $(1,B_{1}(\tilde{\boldsymbol{\gamma}})(h_{0},1),(g',1))$. Since $(n_{0}(\tilde{\boldsymbol{\gamma}})(h_{0}),B_{1}(\tilde{\boldsymbol{\gamma}})(h_{0},1),\mu_{1}(\tilde{\boldsymbol{\gamma}})(h^{1}))$ are arbitrary (other than the fact that they belong to $\Gamma(g, B,\mu)$), it follows that
\begin{align*}
  V^{\ast}_{1}(g, B, \mu) \geq \max_{(n,B',\mu'(\cdot) \in \Gamma(g,\delta B,\mu)}  u( n - g,1- n)  + \beta  \int_{\mathbb{G}} \max\{ V^{\ast}_{1}(g',B',\mu'(g')), V^{\ast}_{0}(g',B') \} \pi_{\mathbb{G}}(dg'|g).
\end{align*}

\end{proof}

\subsection{Proofs of Supplementary Lemmas}

\begin{proof}[Proof of Lemma \ref{lem:charac-z}]
(1) Under assumption assumption \ref{ass:U_prop}, $n \mapsto u'(\kappa n-g,1-n) = u_{c}(\kappa n-g) - u_{l}(1-n) = 1 - (1-\tau)\kappa $ and since $\kappa < 1$ and $\tau \in [0,1]$ it implies that $u'(\kappa n-g,1-n)>0$. Also, $n \mapsto u(\kappa n-g,1-n)$ is continuous. Moreover, $\{ n : z(\kappa,n,g) = 0 \} = \{ n : \kappa (u_{c}(n-g,1-n) - u_{l}(n-g,1-n))n - u_{c}(n-g,1-n) g = 0  \}$. Under assumption  \ref{ass:U_prop}, $u_{c}$ and $u_{l}$ are continuous, and thus this set is closed (and bounded). Therefore it is compact. By the theorem of the maximum $\arg\max_{n \in [0,1]} \{u(\kappa n-g,1-n) : z(\kappa,n,g) = 0  \} $ exists. Uniqueness follows from the fact that $n \mapsto u(\kappa n-g,1-n)$ is increasing.

(2) First observe that $n \mapsto z(1,n,g) =  (1-H'(1-n))n -g$ (with $u_{c}=1$) is continuous and thus $\bar{n}(g)$ exists for all $g \in \mathbb{G}$ ($\mathbb{G}$ is such that for all $g \in \mathbb{G}$, $\max_{n \in [0,1]} z(1,n,g) \geq g$). Observe that $n \mapsto z'(1,n,g) = (1-H'(1-n)) + H''(1-n)n$ and $n \mapsto z''(1,n,g) = 2 H''(1-n) - H'''(1-n)n$. By assumption \ref{ass:linear_c}, $z''(1,n,g)<0$ and thus is strictly concave. We now show that $z$ is decreasing. If $\bar{n}(g) = 1$ then the statement is vacuous, so consider $\bar{n}(g) < 1$. Since $\bar{n}(g)$ is the ``argmax'', $z'(1,\bar{n}(g),g) \leq 0$. Since $z$ is strictly concave, $z'$ is a decreasing, hence $z'(1,n,g) < z'(1,\bar{n}(g),g) \leq 0$ for all $n > \bar{n}(g)$, and the result follows.
\end{proof}

\begin{proof}[Proof of Lemma \ref{lem:rec-S}]
  If $\boldsymbol{\gamma} \in \mathcal{S}(h_{0},\phi_{0})$ it follows that, for any public history $h^{t}$ with $h_{t} = (\phi_{t-1},B_{t},\omega_{t}=(g_{t},\delta_{t})) $ with $B_{t} = B_{t}(\boldsymbol{\gamma})(h^{t-1},\phi)$ and  any $\phi \in \{0,1\}$,
  \begin{align*}
    z(\kappa_{\phi},n_{t}(\boldsymbol{\gamma})(h^{t}),g_{t})u_{c}(\omega^{t}) + \phi \{ p_{t}(\omega^{t}) u_{c}(\omega^{t}) B_{t+1}(\boldsymbol{\gamma})(h^{t},\phi) - \delta_{t} u_{c}(\omega^{t}) B_{t} \} \geq 0
  \end{align*}
and $B_{t+1}(\boldsymbol{\gamma})(h^{t},0) = B_{t}$,
\begin{align} \notag
  p_{t}(\omega^{t}) u_{c}(\omega^{t}) = &\beta \int_{\mathbb{G}}
  d_{t+1}(\boldsymbol{\gamma})(h^{t},h_{t+1}(g'))
  \mu_{t+1}(\boldsymbol{\gamma})(h^{t},h_{t+1}(g'))
  \pi_{\mathbb{G}}(dg'|g_{t}) \\ \label{eqn:rec-S-1}
& + \beta \int_{\mathbb{G}}
  (1-d_{t+1}(\boldsymbol{\gamma})(h^{t},h_{t+1}(g')))
  m_{A}(g') q_{t+1}(\omega^{t},\bar{\delta},g') \pi_{\mathbb{G}}(dg'|g_{t})
\end{align}
where $h_{t+1}(g') \equiv (1,B_{t+1}(\boldsymbol{\gamma})(h^{t},1),g',1)$ and
\begin{align*}
  \mu_{t+1}(\boldsymbol{\gamma})(h^{t},h_{t+1}(g')) = u_{c}(n_{t+1}(\boldsymbol{\gamma})(h^{t},h_{t+1}(g'))-g',1-n_{t+1}(\boldsymbol{\gamma})(h^{t},h_{t+1}(g')))
\end{align*}
and $q_{t}$ is the ``secondary market'' price at time $t$, i.e.,
{\small{\begin{align}
  q_{t+1}(\omega^{t+1},\bar{\delta},g) \equiv & \beta \lambda \int_{\mathbb{G}} \int_{\Delta}
  a_{t+1}(\boldsymbol{\gamma})(h^{t},h_{t+1}(g',\delta'))
  \mu_{t+1}(\boldsymbol{\gamma})(h^{t},h_{t+1}(g',\delta')) \delta'
  \pi_{\Delta}(d\delta')  \pi_{\mathbb{G}}(dg'|g) \label{eqn:rec-S-2} \\  \notag
 & + \beta \int_{\mathbb{G}} \left( 1- \lambda + \lambda \int_{\Delta}
  (1-a_{t+1}(\boldsymbol{\gamma})(h^{t},h_{t+1}(g',\delta')))
  \pi_{\Delta}(d\delta') \right) m_{A}(g') q_{t+2}(\omega^{t+1},\bar{\delta},g') \pi_{\mathbb{G}}(dg'|g)
\end{align}}}
with $h_{t+1}(g',\delta') = (0,\delta '
B_{t+1}(\boldsymbol{\gamma})(h^{t},1) , g' , \delta')$.

From equation \ref{eqn:rec-S-1} it follows that $p_{t}(\omega^{t}) u_{c}(\omega^{t}) = \mathcal{P}_{1}(g_{t},B_{t+1}(\boldsymbol{\gamma})(h^{t},1),\mu_{t+1}(\boldsymbol{\gamma})(h^{t},h_{t+1}(\cdot)))$ and from equation \ref{eqn:rec-S-2} $q_{t+1}(\omega^{t},\bar{\delta},g') = \mathcal{P}_{0}(g',B_{t+1}(\boldsymbol{\gamma})(h^{t},1))$. Also, from these equations and the first display it is clear that if $\boldsymbol{\gamma} \in \mathcal{S}(h_{0},\phi_{0})$, then $\boldsymbol{\gamma}|_{h^{t},\phi} \in \mathcal{S}(\phi_{t-1},B_{t},\omega_{t},\phi)$.
\end{proof}

\section{Proofs for Section \ref{sec:bench}}
\label{app:bench}
In order to show proposition \ref{pro:opt-dec}, we need the following
lemmas (whose proofs are relegated to the end of this section). Throughout
this section we assume that assumption \ref{ass:linear_c} holds.

Throughout this section, let
\begin{align*}
  \Gamma_{\phi}(g,B) = \{ (n,B') : z(\kappa_{\phi},n,g) + \phi( \mathcal{P}^{\ast}_{\phi}(g,B')B' - B ) \geq 0~and~B'=B~if~\phi = 0 \}
\end{align*}
with $\kappa_{\phi} \equiv  \phi + \kappa (1-\phi)$.
\begin{lemma}\label{lem:V_bdd}
There exists a constant $\infty > C>0$, such that  $|V^{\ast}_{\phi}(g,B)| \leq C$ for all $(\phi,g,B)$ such that $\Gamma_{\phi}(g,B) \ne \{ \emptyset \}$.
\end{lemma}
\begin{lemma}\label{lem:V_inc}
$B \mapsto V^{\ast}_{1}(g,B)$ is non-increasing for all $g \in \mathbb{G}$.\footnote{This result clearly implies that $\delta \mapsto V^{\ast}_{1}(g,\delta B)$ is non-decreasing for all $g \in \mathbb{G}$ and $B>0$.}
\end{lemma}

\begin{lemma}\label{lem:Vbar_slope}
	There exists a $C>0$, such that $
  \max_{g' \in \mathbb{G}} \max_{B_{1},B_{2} \in \mathbb{B}^{2}} | V^{\ast}_{0}(g',B_{1}) - V^{\ast}_{0}(g',B_{2})  |  \leq \lambda \frac{\beta C}{1-\beta}$.
\end{lemma}

The previous lemma implies that, for any $\epsilon>0$, there exists a $\lambda(\epsilon) > 0 $ such that, for any $\lambda \in [0,\lambda(\epsilon)]$
   \begin{align}
    \max_{g' \in \mathbb{G}} \max_{B_{1},B_{2} \in \mathbb{B}^{2}} | V^{\ast}_{0}(g',B_{1}) - V^{\ast}_{0}(g',B_{2})  | \leq \epsilon.
   \end{align}
\begin{lemma}\label{lem:roll_over1}
  There exists $\bar{\lambda}
  > 0$ such that for all $\lambda \in [0,\bar{\lambda}]$, the
  following holds: For all $(g,B)$ such that $B>0$ and $\mathbf{d}^{\ast}(g,B) = 1$,
  $\mathcal{P}^{\ast}_{1}(g,B')B' \leq B$ for all $B' \in \mathbb{B}$.
\end{lemma}
We observe that for each $B \in \mathbb{B}$, $\mathcal{P}^{\ast}_{0}$ is the fixed point of the
following mapping
\begin{align*}
 q \mapsto & T^{\ast}_{B}[q](\cdot) \\
= &  \lambda \beta \int_{\mathbb{G} \times
   \Delta} \hspace{-0.05in}\mathbf{a}^{\ast}(g',\delta',B) \delta'
 \pi_{\Delta}(d\delta') \pi_{\mathbb{G}}(dg'|\cdot) +   \beta
 \int_{\mathbb{G}} \left( (1-\lambda) + \lambda \int_{\Delta} (1-\mathbf{a}^{\ast}(g',\delta',B))
 \pi_{\Delta}(d\delta')  \right) q(g') \pi_{\mathbb{G}}(dg'|\cdot) \\
= & \lambda \beta \int_{\mathbb{G} \times
   \Delta} \hspace{-0.05in}\mathbf{a}^{\ast}(g',\delta',B) \delta'
 \pi_{\Delta}(d\delta') \pi_{\mathbb{G}}(dg'|\cdot) +   \beta
 \int_{\mathbb{G}} \left( 1 - \lambda \int_{\Delta} \mathbf{a}^{\ast}(g',\delta',B)
 \pi_{\Delta}(d\delta')  \right) q(g') \pi_{\mathbb{G}}(dg'|\cdot)
\end{align*}
for any $B \in \mathbb{B}$, and $q \in \{ f : \mathbb{G} \rightarrow \mathbb{R} ~\text{uniformly bounded} \}$. We use this insight to derive properties
of $\mathcal{P}^{\ast}_{0}$. 
\begin{lemma}
\label{lem:T-q}
  Suppose assumption \ref{ass:linear_c} holds. Then:
  \begin{enumerate}
  \item For each $B \in \mathbb{B}$, $T^{\ast}_{B}$ is a contraction.
  \item For any $(g,B) \in \mathbb{G} \times \mathbb{B}$, $\mathcal{P}^{\ast}_{0}(g,B) \in \left[ 0 , \lambda
      \frac{\beta}{1-\beta} E_{\pi_{\Delta}}[\delta]   \right]$.
  \item If $g$ is $iid$ (distributed according to $\pi_{\mathbb{G}}(\cdot)$), then
    $\mathcal{P}^{\ast}_{0}(g,B)$ is constant in $g$ and given by
    \begin{align*}
      \mathcal{P}^{\ast}_{0}(g,B) = \frac{\lambda \beta \int_{\mathbb{G} \times
   \Delta} \mathbf{a}^{\ast}(g',\delta',B) \delta'
 \pi_{\Delta}(d\delta') \pi_{\mathbb{G}}(dg')}{1-\beta + \beta \lambda
\int_{\mathbb{G} \times
   \Delta} \mathbf{a}^{\ast}(g',\delta',B)
 \pi_{\Delta}(d\delta') \pi_{\mathbb{G}}(dg')}
    \end{align*}
and in this case $| \mathcal{P}^{\ast}_{0}(g,B) | \leq
\frac{\beta \lambda }{1-\beta + \beta \lambda } <1$.
  \end{enumerate}
\end{lemma}


\begin{proof}[Proof of Proposition \ref{pro:opt-dec}]
 \textbf{Part (1).} By lemma \ref{lem:V_inc}, $\delta \mapsto V^{\ast}_{1}(g,\delta B)$ is non-increasing, provided $B>0$ (but this is the only case it matters since the government will never default on savings $B<0$). On the other hand $V^{\ast}_{0}(g,B)$ is constant with respect to $\delta$. Therefore if for some $\delta \in \Delta$, $\mathbf{a}^{\ast}(g,\delta,B)=1$, then for all $\delta_{1} \leq \delta$ the same must hold. Thus, there exists a $\hat{\delta} : \mathbb{G} \times \mathbb{B} \rightarrow [0,1]$ such that
 \begin{align}
   \mathbf{a}^{\ast}(g,\delta,B) = \mathbf{1}_{\{\delta :  \delta \leq \hat{\delta}(g,B)\}}(\delta).
 \end{align}

We now show that $B \mapsto \hat{\delta}(g,B)$ is non-increasing, for all $g \in \mathbb{G}$. It
suffices to show that for any $\delta$ such that $\delta >
\hat{\delta}(g,B_{1})$ then $\delta > \hat{\delta}(g,B_{2})$ for any
$B_{1} < B_{2}$.

Since $\delta > \hat{\delta}(g,B_{1})$, it follows
that $V^{\ast}_{1}(g,\delta B_{1}) < V^{\ast}_{0}(g,B_{1})$. Let $\epsilon(g,B_{1},\delta) \equiv V^{\ast}_{0}(g,B_{1}) - V^{\ast}_{1}(g,\delta B_{1})$. It is easy to see that $\epsilon(g,B_{1},\delta)>0$ for any $(g,B_{1},\delta)$ such that $\delta > \hat{\delta}(g,B_{1})$. Moreover, since $g$, $B_{1}$ and $\delta$ belong to discrete sets, there exists a $\epsilon > 0$ such that $\epsilon \leq \epsilon(g,B_{1},\delta)$ for all  $(g,B_{1},\delta)$ such that $\delta > \hat{\delta}(g,B_{1})$.

Since $B \mapsto V^{\ast}_{1}(g,B)$ is non-increasing (see lemma
\ref{lem:V_inc}) for any $g \in \mathbb{G}$, it follows that $V^{\ast}_{1}(g,\delta B_{2}) \leq
V^{\ast}_{1}(g,\delta B_{1})$ for all $(g,\delta) \in \mathbb{G} \times
\Delta$ (observe that $\delta > 0$ always). Therefore,
\begin{align*}
  V^{\ast}_{1}(g,\delta B_{2}) - V^{\ast}_{0}(g,B_{2}) \leq&  V^{\ast}_{1}(g,\delta B_{1}) - V^{\ast}_{0}(g,B_{2})\leq  V^{\ast}_{1}(g,\delta B_{1}) - V^{\ast}_{0}(g,B_{1}) + \{ V^{\ast}_{0}(g,B_{1}) - V^{\ast}_{0}(g,B_{2})\},
\end{align*}
for all $(g,B_{1},B_{2},\delta)$ such that $\delta > \hat{\delta}(g,B_{1})$.

Hence, if $ |V^{\ast}_{0}(g,B_{1}) - V^{\ast}_{0}(g,B_{2})| < \epsilon$ for any $(g,B_{1},B_{2})$, the previous display implies that $V^{\ast}_{1}(g,\delta B_{2}) - V^{\ast}_{0}(g,B_{2}) < 0$ and the desired result follows. We now show that $| V^{\ast}_{0}(g,B_{1}) - V^{\ast}_{0}(g,B_{2})| < \epsilon$ for any $(g,B_{1},B_{2})$. By lemma \ref{lem:Vbar_slope}, there exists a $C>0$ such that $
  | V^{\ast}_{0}(g,B_{1}) - V^{\ast}_{0}(g,B_{2})| < \lambda \frac{\beta C}{1-\beta},~\forall (B_{1},B_{2},g)$. \\
Thus for any $\varepsilon>0$, there exists a $\lambda(\varepsilon)$, such that $|V^{\ast}_{0}(g,B_{1}) - V^{\ast}_{0}(g,B_{1})| < \varepsilon$ for all $\lambda \in [0,\lambda(\varepsilon)]$. By setting $\varepsilon = \epsilon $ and $\bar{\lambda}=\lambda(\epsilon)$, the desired result follows.\\
 \textbf{Part (2).} Following \cite{ARELLANO_AER08} we show the result in two steps. Throughout the proof $\mathbf{n}^{\ast}_{\phi}$ and $\mathbf{B}^{\ast}$ are the optimal policy functions for labor and debt.

\textbf{Step 1.} We show that for any $B_{1} < B_{2}$, $\mathbb{S}(B_{1}) \subseteq \mathbb{S}(B_{2})$ where $\mathbb{S}(B) = \{ g : \mathbf{d}^{\ast}(g,B) = 1 \}$. If $\mathbb{S}(B_{1}) = \{\emptyset\}$ the proof is trivial, so we proceed with the case in which this does not hold and let $\bar{g} \in \mathbb{S}(B_{1})$. If $B_{2}$ is not feasible, in the sense that there does not exist any $B'$ such that $ B_{2} - \mathcal{P}^{\ast}_{1}(g;B')B' - \max_{n \in [0,1]}z(1,n,\bar{g}) \leq 0$, then  $\mathbb{S}(B_{2}) = \mathbb{G}$. And the result holds trivially, so we proceed with the case that $B_{2}$ is feasible, given $\bar{g}$.

It follows (since we assume that under indifference, the government
chooses not to default) $V^{\ast}_{1}(\bar{g},B_{1}) <
V^{\ast}_{0}(\bar{g},B_{1})$. Since $B \mapsto V^{\ast}_{1}(\bar{g},B)$ is non-increasing (see lemma \ref{lem:V_inc}), it follows that
\begin{align*}
  V^{\ast}_{1}(\bar{g},B_{2}) \leq V^{\ast}_{1}(\bar{g},B_{1}),~for~all ~B_{1} < B_{2}.
\end{align*}
Therefore,
\begin{align*}
   V^{\ast}_{1}(\bar{g},B_{2}) - V^{\ast}_{0}(\bar{g},B_{2}) \leq& V^{\ast}_{1}(\bar{g},B_{1}) - V^{\ast}_{0}(\bar{g},B_{2})\leq V^{\ast}_{1}(\bar{g},B_{1}) - V^{\ast}_{0}(\bar{g},B_{1}) + \{ V^{\ast}_{0}(\bar{g},B_{1})  -  V^{\ast}_{0}(\bar{g},B_{2})      \}.
\end{align*}
Let $\epsilon(\bar{g},B_{1}) \equiv -\{V^{\ast}_{1}(\bar{g},B_{1}) - V^{\ast}_{0}(\bar{g},B_{1})\}$, observe that $\epsilon(\bar{g},B_{1}) > 0$ for any $(B_{1},\bar{g})\in Graph \{ \mathbb{S} \}$. Thus, if $V^{\ast}_{0}(\bar{g},B_{1})  - V^{\ast}_{0}(\bar{g},B_{2}) < \epsilon(\bar{g},B_{1})$, then $V^{\ast}_{1}(\bar{g},B_{2}) < V^{\ast}_{0}(\bar{g},B_{2})$ and the desired result follows.

Observe that $|\mathbb{B} \times Graph(\mathbb{S})| < \infty$, so there exists $\epsilon > 0$ such that $\epsilon \leq \epsilon(\bar{g},B_{1})$ for any $\bar{g}$ and $B_{1}$ in $Graph(\mathbb{S})$. By lemma \ref{lem:Vbar_slope} and our derivations in part (1), there exists a $\lambda(\epsilon) > 0$ such that
\begin{align*}
  |V^{\ast}_{0}(g,B_{1}) - V^{\ast}_{0}(g,B_{2})| < \epsilon ,~\forall \lambda \in [0,\lambda(\epsilon)]~and~(g,B_{1},B_{2}) \in \mathbb{G} \times \mathbb{B}^{2}.
\end{align*}
Hence, $V^{\ast}_{1}(\bar{g},B_{2}) - V^{\ast}_{0}(\bar{g},B_{2}) < 0$, thereby implying that $\bar{g} \in \mathbb{S}(B_{2})$.\\
\textbf{Step 2.}  We show that, for any $B \in \mathbb{B}$ and any
$g_{1} < g_{2}$ in $\mathbb{G}$, if $\mathbf{d}^{\ast}(g_{1},B) = 1$, then $\mathbf{d}^{\ast}(g_{2},B) =
1$. That is, we want to show that $V^{\ast}_{1}(g_{2},B) < V^{\ast}_{0}(g_{2},B)$. Since default occurs for $g_{1}$, it suffices to show that
\begin{align}\label{eqn:opt-dec-3}
  V^{\ast}_{1}(g_{2},B) - V^{\ast}_{0}(g_{2},B) < V^{\ast}_{1}(g_{1},B) - V^{\ast}_{0}(g_{1},B)
\end{align}
or equivalently, $V^{\ast}_{1}(g_{2},B) - V^{\ast}_{1}(g_{1},B) < V^{\ast}_{0}(g_{2},B) - V^{\ast}_{0}(g_{1},B)$. Observe that
\begin{align}\label{eqn:opt-dec-3b}
  V^{\ast}_{0}(g_{2},B) - V^{\ast}_{0}(g_{1},B) = r(\mathbf{n}^{\ast}_{0}(g_{2}))  - r(\mathbf{n}^{\ast}_{0}(g_{1})) - (g_{2} - g_{1})
\end{align}
where $n \mapsto r(n) = n + H(1-n)$. And now take $\tilde{n}$ such that
\begin{align*}
  z(1,\tilde{n},g_{1}) = B - \mathcal{P}^{\ast}_{1}(\mathbf{B}^{\ast}(g_{2},B))\mathbf{B}^{\ast}(g_{2},B);
\end{align*}
i.e., $\tilde{n}$ is such that
$(\tilde{n},\mathbf{B}^{\ast}(g_{2},B))$ are feasible choices given
the state $(g_{1},B)$, and recall $z(1,n,g) \equiv (1- H'(1-n)) n - g$ and  $(g,B) \mapsto \mathbf{B}^{\ast}(g,B)$
is the optimal policy function for debt, when the government has
access to financial markets. Observe that if no such choice exists, then, since $z(1,\tilde{n},g_{1})  \geq z(1,\tilde{n},g_{2}) $,  trivially $\mathbf{d}^{\ast}(g_{2},B)=1$. Also, $\mathcal{P}^{\ast}_{1}$
does not depend on $g$ because of the i.i.d. assumption. Given this construction,
\begin{align*}
&  V^{\ast}_{1}(g_{2},B) - V^{\ast}_{1}(g_{1},B) \\
\leq & r(\mathbf{n}^{\ast}_{1}(g_{2},B)) - g_{2} + \beta
\int_{\mathbb{G}} \mathbb{V}^{\ast}(g',\mathbf{B}^{\ast}(g_{2},B))
\pi_{\mathbb{G}}(dg')   - \left\{ r(\tilde{n}) - g_{1}  + \beta
  \int_{\mathbb{G}}
  \mathbb{V}^{\ast}(g',\mathbf{B}^{\ast}(g_{2},B))
  \pi_{\mathbb{G}}(dg')   \right\} \\
= & r(\mathbf{n}^{\ast}_{1}(g_{2},B)) - r(\tilde{n}) - (g_{2} - g_{1})
\end{align*}
where $(g,B) \mapsto \mathbb{V}^{\ast}(g,B)
\equiv \max\{ V^{\ast}_{1}(g,B) , V^{\ast}_{0} (g,B) \}$. Given this and \ref{eqn:opt-dec-3b}, it suffices to show that
\begin{align}\label{eqn:opt-dec-3c}
  r(\mathbf{n}^{\ast}_{1}(g_{2},B)) - r(\tilde{n})  \leq r(\mathbf{n}^{\ast}_{0}(g_{2}))  - r(\mathbf{n}^{\ast}_{0}(g_{1})).
\end{align}
We now show this inequality. By construction of $\tilde{n}$,
\begin{align}
  z(1,\tilde{n},g_{1}) = z(1,\mathbf{n}^{\ast}_{1}(g_{2},B),g_{2})
\end{align}
where $(g,B) \mapsto \mathbf{n}^{\ast}_{1}(g,B)$ is the optimal policy
function for labor, when the government has access to financial
markets. Since $n \mapsto z(1,n,g)$ is non-increasing in the relevant
domain (by relevant domain we mean the interval of $n$ which are in ``correct side of the Laffer curve''; see lemma \ref{lem:charac-z}(2)) and $g_{1} < g_{2}$, $\tilde{n} \geq \mathbf{n}^{\ast}_{1}(g_{2},B)$. By analogous arguments, it follows that $\mathbf{n}^{\ast}_{0}(g_{1}) >
\mathbf{n}^{\ast}_{0}(g_{2})$.

Also, note that
\begin{align}
  z(1,\tilde{n},g_{1}) - z(1,\mathbf{n}^{\ast}_{0}(g_{1}),g_{1}) = \mathcal{P}^{\ast}_{1}(\mathbf{B}^{\ast}(g_{2},B))\mathbf{B}^{\ast}(g_{2},B) = z(1,\mathbf{n}^{\ast}_{1}(g_{2},B),g_{2}) - z(1,\mathbf{n}^{\ast}_{0}(g_{2}),g_{2}),
\end{align}
or equivalently, with $n \mapsto \rho(n) = (1-H'(1-n))n$
\begin{align}\label{eqn:rho-1}
  \rho(\tilde{n}) - \rho(\mathbf{n}^{\ast}_{0}(g_{1})) = \rho(\mathbf{n}^{\ast}_{1}(g_{2},B)) - \rho(\mathbf{n}^{\ast}_{0}(g_{2})).
\end{align}
Since $n \mapsto z(1,n,g)$ (and thus $\rho$) is concave and non-increasing (see lemma \ref{lem:charac-z}(2)), it follows $\tilde{n} > (<) \mathbf{n}^{\ast}_{0}(g_{1})$ iff $\mathbf{n}^{\ast}_{1}(g_{2},B) > (<) \mathbf{n}^{\ast}_{0}(g_{2})$.

Putting all these observations together, we have the following possible orders
\begin{align*}
(I) &:~  \mathbf{n}^{\ast}_{0}(g_{1}) \geq \tilde{n} \geq \mathbf{n}^{\ast}_{0}(g_{2}) \geq \mathbf{n}^{\ast}_{1}(g_{2},B)\\
(II) &:~ \mathbf{n}^{\ast}_{0}(g_{1}) \geq \mathbf{n}^{\ast}_{0}(g_{2}) \geq \tilde{n} \geq  \mathbf{n}^{\ast}_{1}(g_{2},B)\\
(III) &:~ \tilde{n} \geq \mathbf{n}^{\ast}_{0}(g_{1})  \geq  \mathbf{n}^{\ast}_{1}(g_{2},B) \geq \mathbf{n}^{\ast}_{0}(g_{2})\\
(IV)&:~ \tilde{n} \geq   \mathbf{n}^{\ast}_{1}(g_{2},B) \geq \mathbf{n}^{\ast}_{0}(g_{1})  \geq \mathbf{n}^{\ast}_{0}(g_{2}).
\end{align*}
Moreover, since in $(g_{1},B)$ the
government defaults, it follows from the proof of lemma
\ref{lem:roll_over1} that $B - \mathcal{P}^{\ast}_{1}(B')B' \geq 0$ for
any $B' \in \mathbb{B}$, in particular for $B'=
\mathbf{B}^{\ast}(g_{2},B)$. Therefore, $z(1,\tilde{n},g_{1}) >
z(1,\mathbf{n}^{\ast}_{0}(g_{1}),g_{1}) $, and thus $\tilde{n} \leq
\mathbf{n}^{\ast}_{0}(g_{1})$, and consequently $\mathbf{n}^{\ast}_{1}(g_{2},B) \leq \mathbf{n}^{\ast}_{0}(g_{2})$. Hence, cases (III) and (IV) are ruled out.

We now study cases (I) and (II). Since $n \mapsto z(1,n,g)$ is strictly concave and non-increasing (see lemma \ref{lem:charac-z}), equation \ref{eqn:rho-1} and (I) and (II) imply
\begin{align}
  \mathbf{n}^{\ast}_{0}(g_{1}) - \tilde{n} \leq \mathbf{n}^{\ast}_{0}(g_{2}) - \mathbf{n}^{\ast}_{1}(g_{2},B).
\end{align}
Since $n \mapsto r(n) \equiv n + H(1-n)$ is concave and increasing under our assumptions, the previous inequality implies that
\begin{align}\label{eqn:opt-dec-3}
  r(\mathbf{n}^{\ast}_{0}(g_{1})) - r(\mathbf{n}^{\ast}_{0}(g_{2})) \leq r(\tilde{n}) - r(\mathbf{n}^{\ast}_{1}(g_{2},B) )
\end{align}
for both case (I) and (II), or equivalently
\begin{align*}
   r(\mathbf{n}^{\ast}_{1}(g_{2},B)) - r(\tilde{n}) \leq
   r(\mathbf{n}^{\ast}_{0}(g_{2})) - (\mathbf{n}^{\ast}_{0}(g_{1})),
\end{align*}
which is precisely equation \ref{eqn:opt-dec-3c}.\\
Hence, step 2 establishes that $\mathbf{d}^{\ast}$ is of the threshold type, since it shows that, for any $B$, if $\mathbf{d}^{\ast}(g,B)=1$, the same is true for any $g' > g$. That is $\{g :  \mathbf{d}^{\ast}(g,B)=1\}$ is of the form $\{g : g \geq \bar{g}(B)\}$. Step 1 shows that the $\bar{g}$ ought to be non-increasing.
\end{proof}
\begin{proof}[Proof of Proposition \ref{pro:preqn-mon0}]
We first establish the result for $i=0$. From lemma \ref{lem:T-q}(3), observe that
\begin{align*}
  \mathcal{P}^{\ast}_{0}(B) = \frac{\beta\lambda \int_{\mathbb{G}} \int_{\Delta} \mathbf{1}_{\{ \delta \leq \hat{\delta}(g',B) \}}(\delta) \pi_{\Delta}(d\delta) \pi_{\mathbb{G}}(dg')}{1 - \beta + \beta\lambda \int_{\mathbb{G}} \int_{\Delta} \mathbf{1}_{\{ \delta \leq \hat{\delta}(g',B) \}}(\delta) \pi_{\Delta}(d\delta) \pi_{\mathbb{G}}(dg')} \frac{\int_{\mathbb{G}}  \int_{\Delta} \mathbf{1}_{\{ \delta \leq \hat{\delta}(g',B) \}}(\delta) \delta \pi_{\Delta}(d\delta) \pi_{\mathbb{G}}(dg')}{\int_{\mathbb{G}} \int_{\Delta} \mathbf{1}_{\{ \delta \leq \hat{\delta}(g',B) \}}(\delta) \pi_{\Delta}(d\delta) \pi_{\mathbb{G}}(dg') }.
\end{align*}
Note that the first term in the RHS is an increasing function (namely $x \mapsto
\frac{x}{1-\beta + x}$) of \linebreak $\beta \lambda \int_{\mathbb{G}}
\int_{\Delta} \mathbf{1}_{\{ \delta \leq \hat{\delta}(g',B) \}}(\delta)\pi_{\Delta}(d\delta) \pi_{\mathbb{G}}(dg')$. Since $B \mapsto \hat{\delta}(g,B)$ is non-increasing (proposition
\ref{pro:opt-dec}), it follows that $B \mapsto \int_{\Delta} \mathbf{1}_{\{ \delta \leq \hat{\delta}(g',B) \}}(\delta)
\pi_{\Delta}(d\delta) $ is also non-increasing, this in turn implies
that the first term in the RHS is also
non-increasing as a function of $B$.

By our assumption $\pi_{\Delta}(\cdot) = 1_{\delta_{0}}(\cdot)$, the second term in the RHS is given by
\begin{align*}
  \frac{\int_{\mathbb{G}}  \int_{\Delta} \mathbf{1}_{\{ \delta \leq \hat{\delta}(g',B) \}}(\delta) \delta \pi_{\Delta}(d\delta) \pi_{\mathbb{G}}(dg')}{\int_{\mathbb{G}} \int_{\Delta} \mathbf{1}_{\{ \delta \leq \hat{\delta}(g',B) \}}(\delta) \pi_{\Delta}(d\delta) \pi_{\mathbb{G}}(dg') } = \delta_{0} \frac{\int_{\mathbb{G}}  \mathbf{1}_{\{ \delta_{0} \leq \hat{\delta}(g',B) \}}(\delta)  \pi_{\mathbb{G}}(dg')}{\int_{\mathbb{G}}  \mathbf{1}_{\{ \delta_{0} \leq \hat{\delta}(g',B) \}}(\delta) \pi_{\mathbb{G}}(dg') } = \delta_{0}
\end{align*}
and thus constant. Hence, $B \mapsto \mathcal{P}^{\ast}_{0}(B) $ is non-increasing.

For $i=1$, observe
  that for any $B_{1} \leq B_{2}$,
  \begin{align*}
    \mathcal{P}^{\ast}_{1}(B_{1}) =& \beta \int_{\mathbb{G}}
    \mathbf{1}_{\{ g' \leq \bar{g}(B_{1}) \}}(g')
    \pi_{\mathbb{G}}(dg') + \beta \int_{\mathbb{G}} \mathbf{1}_{\{
      g' > \bar{g}(B_{1}) \}}(g')
    \pi_{\mathbb{G}}(dg') \mathcal{P}^{\ast}_{0}(B_{1}) \\
\geq & \beta \int_{\mathbb{G}}
    \mathbf{1}_{\{ g' \leq \bar{g}(B_{2}) \}}(g')
    \pi_{\mathbb{G}}(dg') + \beta \int_{\mathbb{G}} \mathbf{1}_{\{
      g' > \bar{g}(B_{2}) \}}(g')
    \pi_{\mathbb{G}}(dg')  \mathcal{P}^{\ast}_{0}(B_{1}) \\
\geq & \beta \int_{\mathbb{G}}
    \mathbf{1}_{\{ g' \leq \bar{g}(B_{2}) \}}(g')
    \pi_{\mathbb{G}}(dg') + \beta \int_{\mathbb{G}} \mathbf{1}_{\{
      g' > \bar{g}(B_{2}) \}}(g')
    \pi_{\mathbb{G}}(dg') \mathcal{P}^{\ast}_{0}(B_{2}) \\
 = & \mathcal{P}^{\ast}_{1}(B_{2})
  \end{align*}
where the first inequality follows from the fact that $B \mapsto
\bar{g}(B)$ is non-increasing (proposition \ref{pro:opt-dec}) and $ \mathcal{P}^{\ast}_{0}(B) <
1$ for any $B \in \mathbb{B}$ (see lemma \ref{lem:T-q}(3)) ; the second inequality follows from the fact that $\mathcal{P}^{\ast}_{0}$ is non-increasing.
\end{proof}

\subsection{Proofs of Supplementary Lemmas}

\begin{proof}[Proof of Lemma \ref{lem:V_bdd}]
For any $(\phi_{-},g,\delta,B) \in \{0,1\} \times \mathbb{G} \times \Delta \times \mathbb{B}$, and any function $(\phi_{-},g,\delta,B) \mapsto F(\phi_{-},g,\delta,B)$ we define the following operator
\begin{align}\label{eqn:V_bdd-0}
  T[F](\phi_{-},g,\delta,B) = \max_{(a,d) \in D(\phi_{-},\delta)} T_{1}[F](\phi_{-}(1-d) + a (1-\phi_{-}),g,\delta,\varphi(B,\delta,a,d))
\end{align}
with $D(0,\delta) = \{0,1\} \times \{1\}$ if $\delta \ne \bar{\delta}$ and $D(0,\bar{\delta}) = \{0\} \times \{1\}$, also $D(1,\delta) = \{1\} \times \{0,1\}$; $\varphi(B,\delta,a,1) = \delta B a + (1-a) B $ and $\varphi(B,\delta,0,d)= B$; and
\begin{align}\label{eqn:V_bdd-1}
  T_{1}[F](\phi,g,\delta,B) = \max_{(n,B') \in \Gamma_{\phi}(g,B)}  \left\{ \kappa_{\phi} n -g + H(1-n)  + \beta \int_{\mathbb{G}} \int_{\bar{\Delta}} F(\phi,g ',\delta',B')  \pi_{\bar{\Delta}}(d\delta'|\phi) \pi_{\mathbb{G}}(dg')  \right\},
\end{align}
where $\pi_{\bar{\Delta}}(\cdot|\phi) = \mathbf{1}_{\{1 \}}(\cdot)$ if $\phi = 1$ and $\pi_{\bar{\Delta}}(\cdot | \phi) = (1-\lambda) \mathbf{1}_{\{\bar{\delta}\}}(\cdot) + \lambda \pi_{\Delta}(\cdot)$ if $\phi = 0$.

A fixed point of the $T$ operator is given by
\begin{align}\label{eqn:V_bdd-2}
  \mathbb{V}^{\ast}(\phi_{-},g,\delta,B) = \max_{(a,d) \in D(\phi_{-},\delta)} V^{\ast}_{\phi_{-}(1-d) + a (1-\phi_{-})}(g,\varphi(B,\delta,a,d))
\end{align}
and for any $\phi \in \{0,1\}$
  \begin{align}\label{eqn:V_bdd-3}
    V^{\ast}_{\phi}(g,B) = \max_{(n,B') \in \Gamma_{\phi}(g,B)}  \left\{ \kappa_{\phi} n -g + H(1-n)  + \beta \int_{\mathbb{G}} \int_{\bar{\Delta}} \mathbb{V}^{\ast}(\phi,g ',\delta',B')  \pi_{\bar{\Delta}}(d\delta'|\phi) \pi_{\mathbb{G}}(dg')  \right\}.
  \end{align}
To verify equation \ref{eqn:V_bdd-3}, see that if $\phi=0$, $B'=B$ by the restrictions imposed on $\Gamma_{0}$, $\kappa_{0} = \kappa$ and
\begin{align*}
  \int_{\bar{\Delta}} \mathbb{V}^{\ast}(0,g ',\delta',B')  \pi_{\bar{\Delta}}(d\delta'|0) = & \lambda \int_{\Delta} \mathbb{V}^{\ast}(0,g ',\delta',B)  \pi_{\bar{\Delta}}(d\delta') + (1-\lambda) \mathbb{V}^{\ast}(0,g ',\bar{\delta},B)\\
 = & \lambda \int_{\Delta} \max_{a \in \{0,1\}} V^{\ast}_{a}(g ',B(\delta a + (1-a)))  \pi_{\bar{\Delta}}(d\delta') + (1-\lambda) \mathbb{V}^{\ast}(0,g ',\bar{\delta},B)\\
= & \lambda \int_{\Delta} \max \{ V^{\ast}_{1}(g ',B \delta ) ,  V^{\ast}_{0}(g ',B )\} \pi_{\bar{\Delta}}(d\delta') + (1-\lambda) V_{0}^{\ast}(g ',B)
\end{align*}
where the last line follows from the fact that $D(0,\bar{\delta}) = \{0\} \times \{1\}$. If $\phi=1$, then
\begin{align*}
  \int_{\bar{\Delta}} \mathbb{V}^{\ast}(1,g ',\delta',B')  \pi_{\bar{\Delta}}(d\delta'|1) = &  \mathbb{V}^{\ast}(1,g ',1,B') \\
= & \max_{d \in \{0,1\}} V^{\ast}_{(1-d)}(g',\varphi(B',1,1,d))\\
 = & \max\{ V^{\ast}_{1}(g',\varphi(B',1,0,0)),V^{\ast}_{0}(g',\varphi(B',1,0,1))\}\\
 = & \max\{ V^{\ast}_{1}(g',B'),V^{\ast}_{0}(g',B')\}.
\end{align*}
Observe that from this fixed point we can derive the functions $V^{\ast}$ by using equation \ref{eqn:V_bdd-3}.

We now show that the operator $T$ maps bounded functions onto bounded functions. Take $F$ such that $|F(\phi_{-},g,\delta,B)| \leq C $ for all $(\phi_{-},g,\delta,B)$ and for some finite constant $C>0$. Then
\begin{align*}
 | T[F](\phi_{-},g,\delta,B) | = | \max_{(a,d) \in D(\phi_{-},\delta)} T_{1}[F](\phi_{-}(1-d) + a (1-\phi_{-}),g,\delta,\varphi(B,\delta,a,d)) |.
\end{align*}
If $(g,\delta,B)$ are such that $\Gamma_{1}(g,\delta B) = \{ \emptyset \}$, then by convention, $\phi_{-}(1-d) + a (1-\phi_{-})=0$ (i.e., there is default/no repayment) and thus $\max_{(a,d) \in D(\phi_{-},\delta)} T_{1}[F](\phi_{-}(1-d) + a (1-\phi_{-}),g,\delta,\varphi(B,\delta,a,d)) = F(0,g,\delta,\varphi(B,\delta,0,1)) = F(0,g,\delta,B)$ and since by our assumptions over $\mathbb{G}$, $\Gamma_{0}(g,B) \ne \{\emptyset\}$ for any $(g,B)$, there exists a finite $c' > 0$ such that $ | \max_{n \in \Gamma_{0}(g,B)} \kappa n - g + H(1-n) | \leq c'$. This implies that in this case $| T[F](\phi_{-},g,\delta,B) | \leq c' + \beta C$.

Similarly, if $(g,\delta,B)$ are such that $\Gamma_{1}(g,\delta B) \ne \{ \emptyset \}$ then $ | \max_{n \in \Gamma_{1}(g,\delta B)}  n - g + H(1-n) | \leq c' $ and it follows that $| T[F](\phi_{-},g,\delta,B) | \leq c' + \beta C$. Hence, by letting $C = \frac{c'}{1-\beta}$ we showed that $T$ maps bounded functions onto bounded functions.

The fix point $ \mathbb{V}^{\ast}$ inherits this property, i,.e., $| \mathbb{V}^{\ast}(\phi_{-},g,\delta,B) | \leq C$ for all $(\phi_{-},g,\delta,B)$. This result, the fact that $| \max_{n \in \Gamma_{0}(g,B)} \kappa n - g + H(1-n) | \leq c'$ and equation \ref{eqn:V_bdd-3} implies that there exists a finite constant $C''>0$, such that $|V^{\ast}_{0}(g,B)| \leq C''$. An analogous result holds for $V^{\ast}_{1}(g,B)$ provided that $(g,B)$ are such that $\Gamma_{1}(g,B) \ne \{ \emptyset \}$.
\end{proof}
\begin{proof}[Proof of Lemma \ref{lem:V_inc}]
  It is easy to see that $\Gamma_{1}(g,B_{1}) \subseteq \Gamma_{1}(g,B_{2})$ for any $B_{1} \geq B_{2}$ and this immediately implies that
  \begin{align*}
   V^{\ast}_{1}(g,B_{1}) = &  \max_{(n,B') \in \Gamma_{1}(g,B_{1})} \{ n-g + H(1-n) + \beta \int_{\mathbb{G}} \max\{ V^{\ast}_{0}(g',B') , V^{\ast}_{1}(g',B') \} \pi_{\mathbb{G}}(dg')\\
 \leq &  \max_{(n,B') \in \Gamma_{1}(g,B_{2})} \{ n-g + H(1-n) + \beta \int_{\mathbb{G}} \max\{ V^{\ast}_{0}(g',B') , V^{\ast}_{1}(g',B') \} \pi_{\mathbb{G}}(dg')\\
= & V^{\ast}_{1}(g,B_{2})
  \end{align*}
and the result follows for $V^{\ast}_{1}$.
\end{proof}

\begin{proof}[Proof of Lemma \ref{lem:Vbar_slope}]
  Observe that, for any $(g,B_{1},B_{2}) \in \mathbb{G} \times \mathbb{B}^{2}$,
  \begin{align*}
 | V^{\ast}_{0}(g,B_{1}) - V^{\ast}_{0}(g,B_{2})| \leq &  \lambda \beta \int_{\mathbb{G}} \int_{\Delta} \mathbf{a}^{\ast}(g,\delta,B) |V^{\ast}_{1}(g',\delta B_{1}) - V^{\ast}_{1}(g',\delta B_{2})| \pi_{\Delta}(d\delta)  \pi_{\mathbb{G}}(dg'|g) \\
& + \beta \int_{\mathbb{G}} \{ (1-\lambda) + \lambda \int_{\Delta} (1-\mathbf{a}^{\ast}(g,\delta,B)) \pi_{\Delta}(d\delta)  \}  | V^{\ast}_{0}(g',B_{1}) - V^{\ast}_{0}(g',B_{2})  | \pi_{\mathbb{G}}(dg'|g) \\
\leq & \lambda \beta \int_{\mathbb{G}} \int_{\Delta} \mathbf{a}^{\ast}(g,\delta,B) |V^{\ast}_{1}(g',\delta B_{1}) - V^{\ast}_{1}(g',\delta B_{2})| \pi_{\Delta}(d\delta)  \pi_{\mathbb{G}}(dg'|g) \\
& + \beta \max_{g' \in \mathbb{G}} | V^{\ast}_{0}(g',B_{1}) - V^{\ast}_{0}(g',B_{2})  |  \\
\leq & \lambda \beta C + \beta \max_{g' \in \mathbb{G}} | V^{\ast}_{0}(g',B_{1}) - V^{\ast}_{0}(g',B_{2})  |
  \end{align*}
where the last line follows from lemma \ref{lem:V_bdd} and the fact that if $(g,\delta,B)$ are such that $\Gamma_{1}(g,\delta B) = \{ \emptyset\}$ then $ \mathbf{a}^{\ast}(g,\delta,B) = 0$. Therefore,
\begin{align*}
  \max_{g' \in \mathbb{G}} \max_{B_{1},B_{2} \in \mathbb{B}^{2}} | V^{\ast}_{0}(g',B_{1}) - V^{\ast}_{0}(g',B_{2})  |  \leq \lambda \frac{\beta C}{1-\beta}.
\end{align*}
\end{proof}

\begin{proof}[Proof of Lemma \ref{lem:roll_over1}]
  Suppose not. That is, for any $\lambda$, there exists a $(g,B)$ with
  $B>0$ such that $\mathbf{d}^{\ast}(g,B) = 1$ but there exists a $B'$ such that
  $\mathcal{P}^{\ast}_{1}(g,B')B' > B$.

First observe that for any $(g,B,B')$ such that $\mathcal{P}^{\ast}_{1}(g,B')B' > B$,
\begin{align*}
    z(1,\mathbf{n}(g,B,B'),g) <  z(1,\mathbf{n}^{\ast}_{0}(g),g)
  \end{align*}
where $\mathbf{n}(g,B,B')$ is the level of labor that solves $z(1,n,g)
+ \mathcal{P}^{\ast}_{1}(g,B')B' = B$. Since $n \mapsto z(1,n,g)$
  is non-increasing in the relevant domain (see lemma \ref{lem:charac-z}(2)), it
  follows that $\mathbf{n}(g,B,B') > \mathbf{n}^{\ast}_{0}(g)$,
  thereby implying that the per-period payoff is greater under no
  default, i.e.,
  \begin{align}\label{eqn:roll-over1-1}
    r(\mathbf{n}(g,B,B')) -g - \{r(\mathbf{n}^{\ast}_{0}(g)) -g\} >0
  \end{align}
where $n \mapsto r(n) = n + H(1-n)$ is increasing by our assumptions. Let $U \equiv \{ (g,B,B') \in \mathbb{G} \times \mathbb{B}^{2} : equation~\ref{eqn:roll-over1-1}~holds \}$. Under our assumptions $|U| < \infty$, so there exists a $\epsilon'>0$ such that $r(\mathbf{n}(g,B,B')) -g - \{r(\mathbf{n}^{\ast}_{0}(g)) -g\} \geq \epsilon' $ for all $(g,B,B') \in U$. Consider any $\lambda \in [0,\lambda(0.5\epsilon')]$ where $\epsilon \mapsto \lambda(\epsilon)$ is such that
\begin{align}
  \lambda(\epsilon) | \int_{\mathbb{G}}  \{ \int_{\Delta} \max\{
 V^{\ast}_{1}(g',\delta B') - V^{\ast}_{0}(g',B'), 0 \}  \pi_{\Delta}(d\delta)  \} \pi_{\mathbb{G}}(dg'|g)| \leq \epsilon;
\end{align}
such $\lambda$ exists by lemma \ref{lem:V_bdd}. By our hypothesis, there exists a $(g,B,B')$ with
  $B>0$ such that $\mathbf{d}^{\ast}(g,B) = 1$ and
  $\mathcal{P}^{\ast}_{1}(g,B')B' > B$. And thus $(g,B,B') \in U$. By our choice of $\lambda$,
\begin{align*}
 \int_{\mathbb{G}}  V^{\ast}_{0}(g',B')
 \pi_{\mathbb{G}}(dg'|g) + 0.5 \epsilon' \geq \int_{\mathbb{G}}  \{ \lambda  \int_{\Delta} \max\{
 V^{\ast}_{1}(g',\delta B'), V^{\ast}_{0}(g',B') \}  \pi_{\Delta}(d\delta) + (1-\lambda)
 V^{\ast}_{0}(g',B') \} \pi_{\mathbb{G}}(dg'|g).
\end{align*}
By definition of $\epsilon'$ and the fact that $(g,B,B') \in U$, it follows that
\begin{align*}
  & r(\mathbf{n}(g,B,B')) -g + \beta \int_{\mathbb{G}} \max\{ V^{\ast}_{1}(g',B') , V^{\ast}_{0}(g',B') \}
 \pi_{\mathbb{G}}(dg'|g) \\
> & r(\mathbf{n}^{\ast}_{0}(g)) -g + 0.5 \epsilon' + \beta \int_{\mathbb{G}} V^{\ast}_{0}(g',B')
 \pi_{\mathbb{G}}(dg'|g) \\
\geq & r(\mathbf{n}^{\ast}_{0}(g)) -g  + \beta \{\int_{\mathbb{G}} V^{\ast}_{0}(g',B')
 \pi_{\mathbb{G}}(dg'|g) + 0.5 \epsilon' \} \\
\geq & V^{\ast}_{0}(g,B).
\end{align*}
Since $V^{\ast}_{1}(g,B)$ is larger or equal than the LHS, we conclude that for $(g,B)$ the government decides not to default, but this is a contradiction to the fact that $\mathbf{d}^{\ast}(g,B)=1$.
\end{proof}

\begin{proof}[Proof of Lemma \ref{lem:T-q}]
  \textbf{Part 1.} To show part 1 we show that for each $B \in \mathbb{B}$, $T^{\ast}_{B}$ satisfies the Blackwell sufficient conditions. Henceforth, consider $B \in \mathbb{B}$ given, observe that $T^{\ast}_{B}$ is of the form
  \begin{align}
    T^{\ast}_{B}[q](g) = A_{B}(g) +  \beta
 \int_{\mathbb{G}} C_{B}(g') q(g') \pi_{\mathbb{G}}(dg'|g)
  \end{align}
where $A_{B}(\cdot)~ \equiv ~\lambda ~\beta~ \int_{\mathbb{G} \times
   \Delta} ~\mathbf{a}^{\ast}(g',\delta',B) ~\delta'~
 \pi_{\Delta}(d\delta') ~\pi_{\mathbb{G}}(dg'|\cdot) $,  and $C_{B}(g) ~\equiv ~\left( ~(~1-~\lambda~) + \right.$
$\lambda~ \int_{\Delta} (1~- $

 \hspace{-0.3in}$ \left.\mathbf{a}^{\ast}(g',\delta',B))
 \pi_{\Delta}(d\delta')  \right)$ is non-negative and less than one. Hence for any $g \in \mathbb{G}$ and for any $q \leq q'$, $T^{\ast}_{B}[q](g)  \leq T^{\ast}_{B}[q'](g) $ and $T^{\ast}[q+a](g) = A_{B}(g) +  \beta
 \int_{\mathbb{G}} C_{B}(g') q(g') \pi_{\mathbb{G}}(dg'|g) + \beta  \int_{\mathbb{G}} C_{B}(g') q(g') \pi_{\mathbb{G}}(dg'|g) a \leq  A_{B}(g) +  \beta
 \int_{\mathbb{G}} C_{B}(g') q(g') \pi_{\mathbb{G}}(dg'|g) + \beta a = T^{\ast}_{B}[q](g) + \beta a$. Therefore $T^{\ast}_{B}$ is a contraction by Blackwell sufficient conditions, see \cite{SLP_book89}, moreover its modulus is given by $\beta$ which does not depend on $B$. \\

\textbf{Part 2.} Consider $C \equiv \beta \lambda \frac{E_{\pi_{\Delta}}[\delta]}{1-\beta}$ such that $|q(g)| \leq C$, then
\begin{align}
  |T^{\ast}_{B}[q](g)| \leq |A_{B}(g)| + \beta C \leq \beta \lambda E_{\pi_{\Delta}}[\delta] + \beta C = \beta \lambda E_{\pi_{\Delta}}[\delta] \{ 1 + \frac{\beta }{1-\beta} \} = \beta \lambda E_{\pi_{\Delta}}[\delta] \frac{1 }{1-\beta},
\end{align}
so in fact $T^{\ast}_{B}$ maps functions bounded by $C$ into themselves; and this holds for any $B \in \mathbb{B}$. Thus the fixed point of $T^{\ast}_{B}$ also satisfies the inequality.\\

\textbf{Part 3.} Since $\pi_{\mathbb{G}}(\cdot|g)$ are constant
with respect to $g$ it follows that
\begin{align*}
  \mathcal{P}^{\ast}_{0}(g,B) = \lambda \beta \int_{\mathbb{G} \times
   \Delta} \mathbf{a}^{\ast}(g',\delta',B) \delta'
 \pi_{\Delta}(d\delta') \pi_{\mathbb{G}}(dg') +   \beta
 \int_{\mathbb{G}} \left( 1 - \lambda \int_{\Delta} \mathbf{a}^{\ast}(g',\delta',B)
 \pi_{\Delta}(d\delta')  \right) \mathcal{P}^{\ast}_{0}(g',B)
\pi_{\mathbb{G}}(dg')
\end{align*}
and thus $\mathcal{P}^{\ast}_{0}(g,B)$ is constant with respect to
$g$, abusing notation we denote it as
$\mathcal{P}^{\ast}_{0}(B)$. From the display above it follows that
\begin{align*}
  \mathcal{P}^{\ast}_{0}(B) =&  \frac{\lambda \beta \int_{\mathbb{G} \times
   \Delta} \mathbf{a}^{\ast}(g',\delta',B) \delta'
 \pi_{\Delta}(d\delta') \pi_{\mathbb{G}}(dg') } { 1-\beta
 \int_{\mathbb{G}} \left( 1 - \lambda \int_{\Delta} \mathbf{a}^{\ast}(g',\delta',B))
 \pi_{\Delta}(d\delta')  \right) \pi_{\mathbb{G}}(dg')} = \frac{\lambda \beta \int_{\mathbb{G} \times
   \Delta} \mathbf{a}^{\ast}(g',\delta',B) \delta'
 \pi_{\Delta}(d\delta') \pi_{\mathbb{G}}(dg') } { 1-\beta + \beta \lambda
\left(\int_{\mathbb{G}} \int_{\Delta} \mathbf{a}^{\ast}(g',\delta',B)
 \pi_{\Delta}(d\delta')  \pi_{\mathbb{G}}(dg') \right) }.
\end{align*}
Since $\delta \in \Delta$ is such that $\delta \leq 1$, $|\mathcal{P}^{\ast}_{0}(B)| \leq \frac{\lambda \beta \left(\int_{\mathbb{G}} \int_{\Delta} \mathbf{a}^{\ast}(g',\delta',B)
 \pi_{\Delta}(d\delta')  \pi_{\mathbb{G}}(dg') \right) } { 1-\beta + \beta \lambda
\left(\int_{\mathbb{G}} \int_{\Delta} \mathbf{a}^{\ast}(g',\delta',B)
 \pi_{\Delta}(d\delta')  \pi_{\mathbb{G}}(dg') \right) } \leq
\frac{\beta \lambda }{1-\beta + \beta \lambda } <1 $.

\end{proof}


\subsection{Derivation of Equation \ref{eqn:LM-lom}}
\label{app:LM-LoM}

In this setting, to default or not, boils down to choosing a $T$ (contingent on $\omega^{\infty}$) such that for all $t< T(\omega^{\infty})$ there is no default and for $t \geq T(\omega^{\infty})$ there is financial autarky. Recall that under our assumptions $u(c,l) = c + H(l)$ and $g_{t} \sim iid \pi_{\mathbb{G}}$, also $\pi_{\mathbb{G}}$ has a density with respect to Lebesgue, which we denote as $f_{\pi_{\mathbb{G}}}$.

For any $\omega^{t} \in \Omega^{t}$ and $ t\leq T(\omega^{\infty})$,
\begin{align*}
  V_{1}^{\ast}(g_{t},B_{t}(\omega^{t-1})) = & \max_{(n,B') \in \Gamma(g_{t},B_{t}(\omega^{t-1}),1)}  n - g + H(1-n) + \beta \int_{\{g' : g' \leq \bar{g}(B')\}} \{ V_{1}^{\ast}(g',B') - V_{0}^{\ast}(g')\} \pi_{\mathbb{G}}(dg') \\
 & + \beta \int V_{0}^{\ast}(g') \pi_{\mathbb{G}}(dg')
\end{align*}
and let $\nu_{t}(\omega^{t})$ is the Lagrange multiplier of the restriction, $z(1,n,g_{t}) + \mathcal{P}^{\ast}_{1}(B') B' - B_{t}(\omega^{t-1}) \geq 0$. By assumption, the solution of $B'$ is in the interior. So the optimal choice $((n_{t}(\omega^{t}))^{\infty}_{t=0},(B_{t+1}(\omega^{t}))^{\infty}_{t=0})$ satisfy
\begin{align*}
  1 - H'(1 - n_{t}(\omega^{t})) + \nu_{t}(\omega^{t}) \left( \frac{dz(1,n_{t}(\omega^{t}),g_{t})}{dn}  \right) = 0
\end{align*}
or equivalently
\begin{align}\label{eqn:nu-labor}
  \nu_{t}(\omega^{t}) \equiv \boldsymbol{\nu}(n_{t}(\omega^{t})) = - \frac{1 - H'(1 - n_{t}(\omega^{t}))}{1 - H'(1 - n_{t}(\omega^{t})) + H''(1 - n_{t}(\omega^{t}))n_{t}(\omega^{t})},
\end{align}
and
\begin{align*}
&  \nu_{t}(\omega^{t}) \left\{  \mathcal{P}^{\ast}_{1}(B_{t+1}(\omega^{t})) + \frac{d \mathcal{P}^{\ast}_{1}(B_{t+1}(\omega^{t}))}{dB_{t+1}} B_{t+1}(\omega^{t}) \right\} \\
= & \beta \frac{d \int_{\{g' : g' \leq \bar{g}(B_{t+1}(\omega^{t}))\}} \{ V_{1}^{\ast}(g',B_{t+1}(\omega^{t})) - V_{0}^{\ast}(g')\} \pi_{\mathbb{G}}(dg')}{dB_{t+1}}\\
= & \beta \int_{\{g' : g' \leq \bar{g}(B_{t+1}(\omega^{t}))\}} \frac{d V_{1}^{\ast}(g',B_{t+1}(\omega^{t}))}{dB_{t+1}} \pi_{\mathbb{G}}(dg') \\
& + \beta  \{ V_{1}^{\ast}(\bar{g}(B_{t+1}(\omega^{t})),B_{t+1}(\omega^{t})) - V_{0}^{\ast}(\bar{g}(B_{t+1}(\omega^{t})))\} f_{\pi_{\mathbb{G}}}(\bar{g}(B_{t+1}(\omega^{t}))) \frac{d\bar{g}(B_{t+1}(\omega^{t}))}{dB_{t+1} }.
\end{align*}
Since $V_{1}^{\ast}(\bar{g}(B_{t+1}(\omega^{t})),B_{t+1}(\omega^{t})) - V_{0}^{\ast}(\bar{g}(B_{t+1}(\omega^{t}))) = 0$, the last term in the RHS is naught. Also, $\frac{d V_{1}^{\ast}(g_{t},B_{t}(\omega^{t-1}))}{dB_{t}} = \nu_{t}(\omega^{t}) $ and thus
\begin{align}
  \nu_{t}(\omega^{t}) \left\{  \mathcal{P}^{\ast}_{1}(B_{t+1}(\omega^{t})) + \frac{d \mathcal{P}^{\ast}_{1}(B_{t+1}(\omega^{t}))}{dB_{t+1}} B_{t+1}(\omega^{t}) \right\} = \beta \int_{\{g' : g' \leq \bar{g}(B_{t+1}(\omega^{t}))\}} \nu_{t+1}(\omega^{t},g')  \pi_{\mathbb{G}}(dg')
\end{align}
We now show that $\boldsymbol{\nu}$ is decreasing. For this it is easier to first establish that $\boldsymbol{\nu}^{-} \equiv 1/\boldsymbol{\nu}$ is increasing. Observe that
\begin{align*}
 \boldsymbol{\nu}^{-}(n) = - 1 - \frac{H''(1 - n)n}{1 - H'(1 - n)}
\end{align*}
and thus
\begin{align*}
  \frac{d\boldsymbol{\nu}^{-}(n)}{dn} = - \frac{-H'''(1 - n)n + H''(1 - n)}{1 - H'(1 - n)} - \frac{(H''(1 - n))^{2}n}{(1 - H'(1 - n))^{2}}.
\end{align*}
Since $-H'''(1 - n)n + H''(1 - n)<0$ by assumption and $1 - H'(1 - n) = \tau > 0$, then the first term in the RHS is negative; the second term in the RHS is also negative. Hence $\boldsymbol{\nu}^{-}$ is increasing, which readily implies that $\boldsymbol{\nu}$ is decreasing.

\newpage

\newpage
\setcounter{page}{1}

\begin{center}
  \Huge{Supplementary Online Material}
\end{center}

\section{Stylized Facts: Emerging vs. Industrialized Economies}
\label{sec:facts}

Throughout the paper, we mention that our theoretical model is capable of replicating qualitatively several stylized facts observed for a wide range of economies. In this section, we present these stylized facts regarding the domestic
government debt-to-output ratio and central government revenue-to-output ratio of several countries: industrialized
economies (IND), emerging economies (EME) and a subset of these: Latin American (LAC).\footnote{For government
revenue-to-output ratios, we used the data from
\cite{KRV_WP04}, and for the domestic
government debt-to-output ratios, we used the data from
\cite{PANIZZA_WP08}. We thank Ugo Panizza and Carmen Reinhart for kindly sharing their datasets with us. See appendix \ref{app:econometrics} for a detailed description of the data.}

In the dataset set which covers the period 1800-2010, no default event is observed for IND, whereas EME/LAC
(LAC in particular) do exhibit several defaults. Thus, we take the former group as a proxy for
economies with access to risk-free debt and the latter group as a proxy for economies without commitment to repay. It is worth to point out however that we are not presuming that IND economies are a type of economy that would never
never default. In turn, we are just using the fact that in our dataset IND economies do not show default events, to use them as a proxy for the type of economy modeled in AMSS, that is, one with risk-free government debt. 

Several stylized facts that stand out in our dataset. First, in EME/LAC economies default is more likely than in IND economies and within the former group, the default risk is much higher for highly indebted economies. Second, EME and LAC economies exhibit tighter debt ceilings than IND, as also reported by \cite{RR_WP03}. Third, economies with higher default risk tend to exhibit more volatile tax revenues than those with low default risk, and this fact is particularly notable for the group of EME/LAC economies. \cite{Bauducco} documents a similar finding.

As shown in section \ref{sec:bench}, our theory predicts that endogenous borrowing limits are more active for a high level of indebtedness. That is, when the government debt is high relative to output, the probability of default next period is higher, thus implying tighter borrowing limits and higher bond spreads. As the government's ability to smooth its needs for funds using debt is hindered, the volatility of taxes turns out to be higher. But when debt is low, default is an unlikely event, thereby implying slacker borrowing limits, lower spreads and therefore lower volatility in taxes. Hence,
implications in the upper tail of the domestic debt-to-output ratio distribution can be different from those in the ``central
part'' of it. Therefore, the mean and even the variance of the distribution may not be too informative, as they are affected by the central part of the distribution. Quantiles are better suited for recovering the information in the tails of the
distribution.\footnote{We refer the reader to \cite{KOENKER_book95} for a thorough treatment of quantiles and
quantile-based econometric models.}

 Figure \ref{fig:debt_risk_041913scatter} plots the percentiles of the domestic government debt-to-output ratio and of a measure of default risk for three groups: IND (black triangle), EME (blue square) and LAC (red circle).\footnote{This type of graph is not the conventional QQplot as the axis have the value of the random variable which achieves a certain quantile and not the quantile itself. For our purposes, this representation is more convenient.}
The X-axis plots the time series averages of domestic government debt-to-output ratio, and the Y-axis plots the values
of the measure  of default risk.\footnote{The measure of default risk is constructed as the spread using the EMBI+ real index
 from J.P. Morgan for countries for which it is  available and using the 3-7 year real government bond yield for
 the rest, minus the return of a US Treasury bond of similar maturity. Although bond returns are not entirely driven by default risk but also respond to other factors related to risk appetite, uncertainty and liquidity, for our purpose they constitute a valid conventional proxy of default risk. Furthermore, our spreads are an imperfect measure of default risk for domestic debt since EMBI+ considers mainly foreign debt. However, it is still informative since domestic default are positively correlated with defaults on sovereign debt, at least for the period from 1950's onwards. See figure 10 in \cite{RR_WP08}. } For each group, the last point on the right
corresponds to the 95 percentile, the second to last to the 90 percentile and so on; these are comparable between groups as all of them represent a percentile of the corresponding distribution. EME and LAC have lower domestic debt-to-output ratio levels than IND; in fact the domestic debt-to-output ratio value of around 50 percent that pertains to the 95 percentile for EME and LAC, corresponds roughly to only the 85 percentile for IND. \emph{Thus, economies that are prone to default (EME and LAC) exhibit tighter debt ceilings than economies that do not default (in this dataset, represented by IND).}

Figure \ref{fig:debt_risk_041913scatter} also shows that for the IND group, the default risk measure is low and roughly constant for different levels of debt-to-output ratios. On the other hand, the default risk measure for the EME group is not only higher, but increases substantially for high levels of indebtness. We consider this as evidence that for EME economies higher default risk is more prevalent for high levels of debt-to-output ratios.

\begin{figure}
  \centering
 \includegraphics[height=3.25in,width=5.5in]{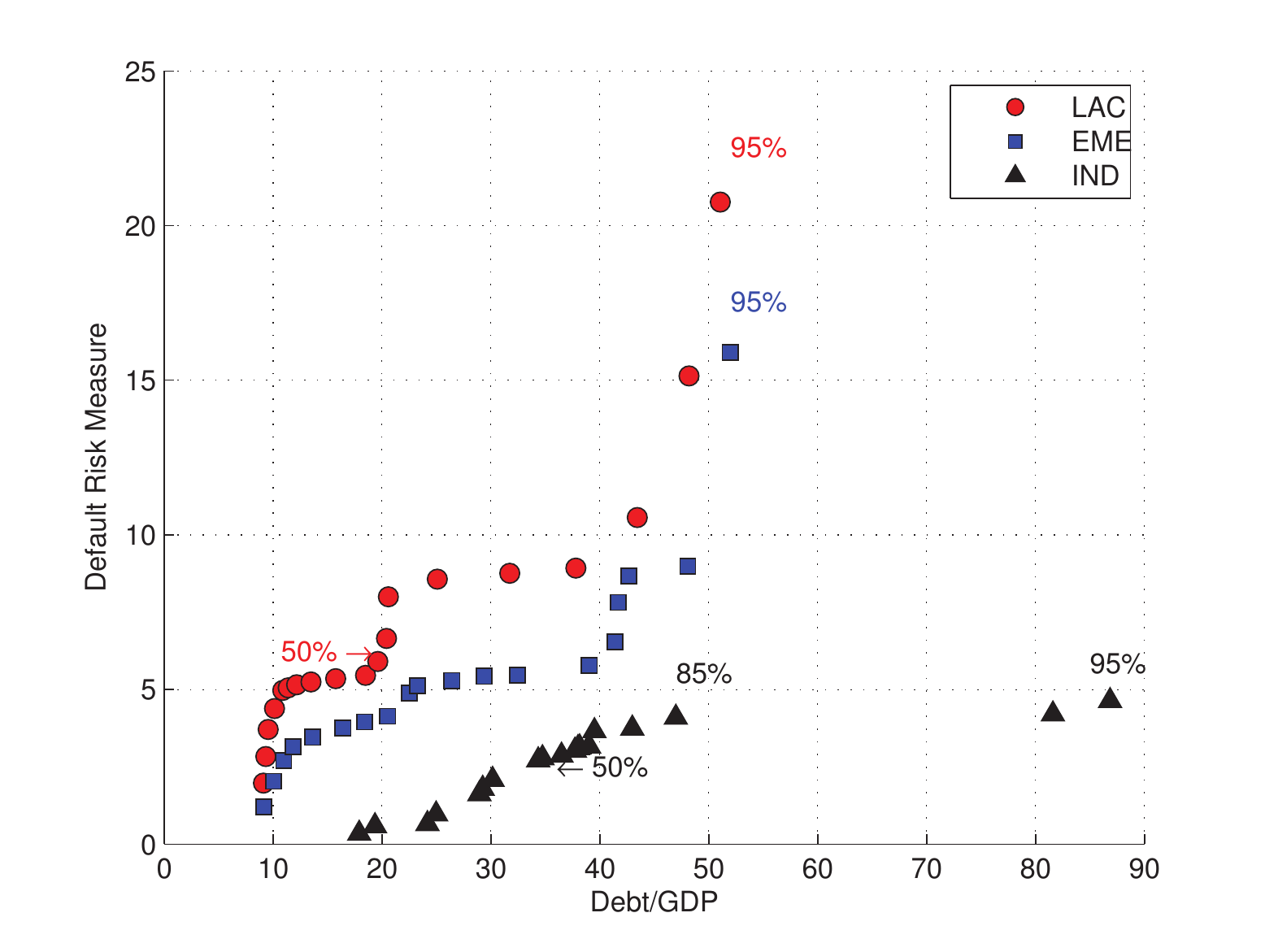}
  \caption{The percentiles of the domestic government debt-to-output ratio and of a measure of default risk for three groups: IND (black triangle), EME (blue square) and LAC (red circle)}
  \label{fig:debt_risk_041913scatter}
\end{figure}

\begin{table}[h]\label{tab:spread}\label{tab:vol_rev}
  \centering
 \caption{(A) Measure of default risk (\%) for EME and IND groups for
    different levels of debt-to-output ratio (\%); (B) standard deviation of central government revenue over GDP (\%) for EME and IND groups for different levels of default risk.}
  \begin{tabular}{ccc|c|ccc}
    \multicolumn{3}{c|}{(A)} & & \multicolumn{3}{c}{(B)} \\ \hline\hline
  Debt/GDP   & EME & IND & & Default Risk & EME & IND \\ \hline
$ 25$ & 5.4  & 2.0  & & $ 25$ & 0.9 & 1.4 \\
$ 75 $ & 10.7  & 2.9 & & $ 75$ & 2.5 & 1.7  \\ \hline\hline
  \end{tabular}
\end{table}

Table \ref{tab:spread}(A) compares the measure of default risk between IND and EME for low and high debt-to-output ratio levels. That is, for both groups (IND and EME) we select economies with debt-to-output ratio below the 25th percentile (low debt-to-output) for which we compute the average risk measure. We proceed analogously with those economies with debt-to-output ratio above the 75th percentile (high debt-to-output). For the case of low debt-to-ouput levels, the EME group presents higher (approximately twice as high) default risk than the IND group. For high debt-to-output ratio economies, however, this difference is quadrupled. \emph{Thus, economies that are prone to default (EME and LAC) exhibit higher default risk than economies that do not default (in this dataset, represented by IND), and, moreover, the default risk is much higher for economies in the former group that have high levels of debt-to-output ratio.}

Table \ref{tab:vol_rev}(B) compares the standard deviation
of the central government revenue-to-output ratio between IND and EME for low and high default risk levels. It indicates that for IND there is little variation of the volatility across low and high levels of default risk. For EME, however, the standard deviation of the central government revenue-to-output ratio is dramatically higher for economies with high default risk.\footnote{We looked also at the inflation tax as a proxy for tax policy; results are qualitatively the same.}. It is worth pointing out that all the EME with high default risk levels defaulted at least once during our sample period. \emph{Thus, economies with higher default risk exhibit more volatile tax revenues than economies with low default risk.} This is particularly notable for the group of EME/LAC economies.

These stylized facts establish a link between (a) default risk/default events, (b) debt ceilings and (c) volatility of tax revenues. In particular, the evidence suggests that economies that show higher default risk, also exhibit lower debt ceilings and more volatile tax revenues.  The theory behind our model helps shed light upon the forces driving these facts.\footnote{It is important to note that we are not arguing any type of causality; we are just illustrating co-movements. In fact, in the model below all three features are endogenous outcomes of equilibrium.}

\section{Description of the Data}
\label{app:econometrics}

In this section we describe how we constructed the figures presented
in section \ref{sec:facts}.

The industrialized economies group consists
of AUSTRALIA (1990-1999), AUSTRIA  (1990-1999), BELGIUM
(1990-2001), CANADA (1990-2003), DENMARK (1990-2003), FINLAND
(1994-1998), FRANCE (1990-2003), GERMANY (1990-1998), GREECE (1990-2001),
IRELAND (1995-2003), ITALY (1990-2003), JAPAN (1990-1993),
NETHERLANDS (1990-2001), NEW ZEALAND (1990-2003), NORWAY
(1990-2003), PORTUGAL   (1990-2001), SPAIN (1990-2003), SWEDEN (1990-2003),
SWITZERLAND  (1990-2003), UNITED KINGDOM (1990-2003) and UNITED
STATES (1990-2003).

The emerging economies group consists of ARGENTINA$^{1}$
(1998-2003), BOLIVIA$^{1}$  (2001-2003), BRAZIL$^{1}$ (1997-2003),
CHILE$^{1}$ (1993-2003), COLOMBIA$^{1}$ (1999-2003), ECUADOR$^{1}$
(1998-2003), EL SALVADOR$^{1}$ (2000-2003), HONDURAS$^{1}$
(1990-2003), JAMAICA$^{1}$ (1990-2003), MEXICO$^{1}$ (1990-2003),
PANAMA$^{1}$ (1997-2003), PERU$^{1}$ (1998-2003), VENEZUELA$^{1}$
(1997-2003), ALBANIA (1995-2003), BULGARIA (1991-2003), CYPRUS
(1990-2003), CZECH REPUBLIC (1993-2003), HUNGARY (1991-2003),
LATVIA (1990-2003), POLAND (1990-2003), RUSSIA (1993-2003), TURKEY
(1998-2003), ALGERIA (1990-2003), CHINA (1997-2003), EGYPT
(1993-2003), JORDAN (1990-2003), KOREA (1990-2003), MALAYSIA
(1990-2003), MAURITIUS (1990-2003), MOROCCO (1997-2003), PAKISTAN
(1990-2003), PHILIPPINES (1997-2003), SOUTH AFRICA (1990-2003),
THAILAND (1999-2003) and TUNISIA (1994-2003). The LAC group is
conformed by the countries with ``$^{1}$''.

For section \ref{sec:facts} we constructed the data as follows.
First, for each country, we computed time average, or time standard
deviations or any quantity of interest (in parenthesis is the
number of observations use to construct these). Second, once we
computed these averages, we group the countries in IND, EME and
LAC. We do this procedure for (a) central government domestic debt
(as \% of output) ; (b) central government expenditure (as \% of
output) ; (c) central government revenue (as \% of output) , and
(d) Real Risk Measure. The data for (a) is taken from
\cite{PANIZZA_WP08} ; the data for (b-c) is taken from
\cite{KRV_WP04} ; finally the data for (d) is taken from
www.globalfinancialdata.com.

For Greece and
Portugal we use central government public debt because central
government domestic debt was not available. For Sweden, Ecuador
and Thailand we use general government expenditure because central
government expenditure was not available. For Albania, Bulgaria,
Cyprus, Czech Rep., Hungary, Latvia, Poland and Russia no measure
of government expenditure was available and thus were excluded
from the sample for the calculations of this variable. The same
caveats apply to the central government revenue sample. For Argentina,
Brazil, Colombia, Ecuador, Egypt, Mexico, Morocco, Panama, Peru,
Philippines, Poland, Russia, Turkey and Venezuela we used the real
EMBI+ as a measure of real risk. For the rest of the countries we
used government note yields of 1-5 years maturity, depending on
availability.

}

\end{document}